\documentclass[11pt, a4paper]{article}
\usepackage{microtype}
\usepackage{amsmath, amsthm, amsfonts, amssymb, amsbsy, bigstrut, graphicx, enumerate, upref, longtable, comment, booktabs, array, caption, subcaption}
\usepackage{thmtools, thm-restate}
\usepackage[utf8]{inputenc} 
\usepackage[T1]{fontenc} 
\usepackage{geometry} 
\let\oldtextsection\textsection
\renewcommand{\textsection}{\ifmmode\mathsection\else\textup{\oldtextsection}\fi}
\newcommand{\secref}[1]{\textup{\oldtextsection\ref{#1}}}
\newcommand{\secrefrange}[2]{\textup{\oldtextsection\ref{#1}--\ref{#2}}}

\usepackage[%
    colorlinks      = true,
  linktocpage     = true,
  linkcolor       = blue,
  citecolor       = blue,
  urlcolor        = blue,
  breaklinks =true,
  hypertexnames=false,
]{hyperref}

\usepackage{setspace}
\usepackage[affil-it]{authblk}

\usepackage[numbers, sort&compress]{natbib}

\newcommand{\I}{\mathrm{i}}

\newcommand{\ep}{\epsilon}

\newcommand{\mc}{\mathcal{C}}

\newtheorem{thm}{Theorem}[section]
\newtheorem{lmm}[thm]{Lemma}
\newtheorem{cor}[thm]{Corollary}
\newtheorem{prop}[thm]{Proposition}

\theoremstyle{definition}
\newtheorem{remark}[thm]{Remark}

\newcommand{\ma}{\mathcal{A}}
\newcommand{\md}{\mathcal{D}}

\newcommand{\mh}{\mathcal{H}}

\newcommand{\rr}{\mathbb{R}}

\newcommand{\inn}[1]{\langle #1 \rangle}

\newcommand{\dist}{\operatorname{dist}}

\numberwithin{equation}{section}

\renewcommand{\Re}{\operatorname{Re}}

\newcommand{\dd}{\mathrm{d}}

\usepackage[utf8]{inputenc}
\usepackage{amsmath}
\usepackage{amsfonts}
\usepackage{bbm}
\usepackage{mathtools}
\usepackage{tikz}
\usetikzlibrary{decorations.pathreplacing}
\usepackage{amsthm}
\usepackage[english]{babel}
\usepackage{fancyhdr}
\usepackage{amssymb}
\usepackage{enumitem}
\usepackage{qtree}
\usepackage{mathrsfs}

\DeclareMathOperator{\supp}{supp}

\newcommand{\Z}{\mathbb{Z}}

\newcommand{\R}{\mathbb{R}}

\newcommand{\C}{\mathbb{C}}

\newcommand{\E}{\mathbb{E}}

\renewcommand{\bar}{\overline}

\renewcommand{\hat}{\widehat}

\newcommand{\bt}{\mathbb{T}}

\renewcommand{\i}{\mathrm{i}}
\renewcommand{\Im}{\operatorname{Im}}
\newcommand{\normord}[1]{:\mathrel{#1}:}

\newcommand{\mo}{\mathcal{O}}

\let\subset\subseteq

\begin{document}
\title{A Lorentzian construction of timelike Liouville field theory on the cylinder}
\author{Sourav Chatterjee\thanks{Department of Statistics, Stanford University, USA. Email: \href{mailto:souravc@stanford.edu}{\tt souravc@stanford.edu}. 
}}
\affil{Stanford University}

\maketitle

\begin{abstract}
Timelike Liouville field theory is a candidate model for positive curvature
two-dimensional quantum gravity, but a mathematically controlled Lorentzian
formulation has remained elusive. In this paper we construct the theory on the
cylinder $\mathbb{R}\times \mathbb{S}^1$ in the integer screening sector for a
natural algebra of renormalized exponential observables. Starting from a
renormalized finite-volume torus regularization, we construct infinite-volume
Euclidean correlation functions, prove analytic continuation in the time
variables, and identify the resulting Lorentzian boundary values by explicit
contour formulas. This yields exact Lorentzian correlators for a natural class
of exponential observables. We then prove locality:
spacelike separated vertex operators commute in the Lorentzian theory. For
smeared observables generated by the integer-charge fields $e^{2nb\phi}$,
 these Lorentzian expectation
values define a vacuum functional on an ordered
$*$-algebra and support an AQFT-type quantization without positivity. More
precisely, we obtain isotone local algebras, a complete locally convex space
$\mathcal H$ with dense algebraic subspace $\mathcal H_0$ carrying a
nondegenerate Hermitian form (shown to be indefinite for
$b<8^{-1/2}$), a continuous cyclic representation,
operator-topologically closed represented local algebras, an action of cylinder
translations by continuous linear homeomorphisms, and locality for the
represented local net. The construction does not produce a
Hilbert space or a Haag--Kastler net of local von Neumann algebras in the
usual sense, but it shows that a substantial part of the
Euclidean-to-Lorentzian and algebraic reconstruction mechanism survives in
this nonpositive setting for timelike Liouville theory on the cylinder.
\newline
\newline
\noindent {\scriptsize {\it Key words and phrases.} Timelike Liouville field theory, two-dimensional quantum gravity, algebraic quantum field theory.}
\newline
\noindent {\scriptsize {\it 2020 Mathematics Subject Classification.} 81T08, 83C45, 81T40.}
\end{abstract}

\clearpage

\tableofcontents

\section{Introduction}

Liouville field theory has long occupied a central position in two-dimensional quantum gravity and conformal field theory. In the spacelike regime, a coherent nonperturbative picture has emerged through a combination of conformal bootstrap, random geometry, and probabilistic methods. The timelike regime is subtler. There, the sign reversal in the kinetic term obstructs the usual Euclidean probabilistic construction. Even after a nonstandard Euclidean construction is carried out~\cite{chatterjee25}, the passage from Euclidean to Lorentzian requires explicit analytic continuation,  because the machinery of constructive field theory~\cite{glimmjaffe87} is not expected to be available for this problem (in particular, reflection positivity is expected to fail). Perhaps for this reason, a mathematically controlled Lorentzian formulation has remained largely open.

The present paper continues the program initiated in \cite{chatterjee25} and extended in \cite{chatterjee26}. The starting point of that program was a rigorous Euclidean treatment of timelike Liouville correlation functions, based on a finite-volume regularization and a precise way of handling the wrong-sign field; the second paper pushed the exact formulas beyond the charge neutral sector. Here we take the next step. We move from Euclidean correlators to Lorentzian observables and show that, on the cylinder $\mathbb{R}\times\mathbb{S}^1$, the resulting Lorentzian theory supports both a locality statement and an AQFT-type quantization without positivity; in the range $b<8^{-1/2}$, the vacuum form is proved to be genuinely indefinite.

We work in the integer screening sector, where the Coulomb gas representation remains explicit enough to analyze. Our first main result is a Lorentzian construction of pointlike correlators. Starting from a renormalized finite-volume torus regularization, we construct infinite-volume Euclidean correlators on the cylinder, prove that they admit analytic continuation in the time variables, and identify their Lorentzian boundary values by explicit contour formulas. This yields exact Lorentzian correlation functions for a natural class of renormalized exponential observables.

Our second main result is locality. Because the Lorentzian correlators arise as ordered boundary values, symmetry under permutation of insertions is not built in, and commutativity is a theorem rather than a formal convention. We prove that spacelike separated vertex operators commute on the cylinder. We then pass to smeared observables generated by the integer-charge vertex fields $e^{2nb\phi}$ and show that the same Lorentzian data determine an AQFT-type structure with isotony, translation covariance, a cyclic vacuum vector, and locality for the represented local net.

The resulting quantization is not a positive Haag--Kastler theory in the usual sense. The natural vacuum form is Hermitian and nondegenerate on a dense algebraic core, and it is shown to be indefinite when $b<8^{-1/2}$. Accordingly, we do not obtain reflection positivity, a Hilbert space, abstract local $C^*$-algebras, or represented local von Neumann algebras. The point, rather, is that a substantial part of the Euclidean-to-Lorentzian and algebraic reconstruction mechanism survives beyond positivity and can be carried out rigorously for timelike Liouville theory on the cylinder.

This work fits into a broader literature. Liouville theory entered quantum gravity in two dimensions via Polyakov's conformal gauge formulation of noncritical strings~\cite{polyakov81}, and standard references emphasize the distinction between spacelike and timelike regimes \cite{teschner01review,nakayama04review,chatterjeewitten25}. For spacelike Liouville theory, the exact structure constants and bootstrap framework were developed in \cite{dornotto94,teschner95,zamolodchikovzamolodchikov96}, while the Gaussian free field and Gaussian multiplicative chaos approach led to rigorous constructions and exact formulas \cite{davidetal16,duplantiersheffield11,kupiainenetal20,lacoinetal22}. In the timelike direction, longstanding proposals and debates concern the interpretation of amplitudes and analytic continuation \cite{stromingertakayanagi03,sen02rolling,sen02tachyon,schomerus03,harlowetal11,giribet12,BautistaBawane2022,GiribetSivilotti2026,ribaultsantachiara15,anninosetal21,muhlmann22,collieretal24,collieretal25,muhlmannetal25}.

Recent rigorous work shows that exact Euclidean timelike constructions are nevertheless possible. Compactified and noncompactified imaginary Liouville theories have been developed in Euclidean signature \cite{guillarmouetal23,usciatietal25}. In particular, \cite{chatterjee25} initiated a rigorous program for timelike Liouville theory by constructing Euclidean correlation functions with exact formulas and semiclassical consistency checks, and \cite{chatterjee26} extended that picture beyond charge neutrality. For a related rigorous construction of a different exponential field theory on the infinite cylinder, see also \cite{guillarmouetal24sinh}. For a recent cosmological perspective on timelike Liouville as a tractable model of positive curvature quantum gravity, see~\cite{AnninosHertogKarlsson2025}. The present paper continues this program in the Lorentzian direction.

Methodologically, the paper combines the finite-volume regularization and wrong-sign field technology from \cite{chatterjee25} with infinite-volume limits on the cylinder, sector-by-sector analytic continuation of Coulomb gas integrals, and contour representations of Lorentzian boundary values. On the algebraic side, it borrows the representation theoretic logic of Gelfand--Naimark--Segal \cite{gelfandnaimark43,segal47irred,segal47postulates,kadisonringrose83,brattelirobinson87} and the representation-dependent side of AQFT \cite{haagkastler64,haag96,kadisonringrose83,brattelirobinson87}, but adapts both to an ordered algebra and a locally convex state space with a Hermitian form that need not be positive.

The remainder of the paper is organized as follows. 
\secref{sec:results} presents the main Euclidean and Lorentzian correlation formulas, works out some simple exact calculations, introduces smeared observables, and states the AQFT-type quantization results without positivity, including the locality theorems for vertex operators and for the represented local net.
\secref{sec:math} develops the Euclidean construction and infinite-volume limit. 
\secref{sec:anacont} proves the analytic continuation statements and derives the contour representation of Lorentzian correlations.
\secref{sec:quantization} proves the results on smeared observables and quantization, including the AQFT-type net, the represented local algebras, the state space, translation covariance, and the locality theorem for the represented local net.
The appendix collects background material used for orientation and comparison, including a brief recap of the standard AQFT framework.

\section{Results}\label{sec:results}
In this section we present the main results of the paper. We begin with the heuristic description of timelike Liouville field theory on the infinite cylinder that motivates the rigorous construction.
The discussion then follows the same order as the construction itself: first pointlike correlators and their Lorentzian continuation, then smeared observables, and finally the AQFT-type quantization statements built from them, including the locality theorems for vertex operators and for the represented local net.
\subsection{Heuristic description of the model}
We begin by recalling the formal Lorentzian model on the cylinder and isolating the class of observables on which the rest of the paper will focus.
Fix the spacetime domain
\[
M:=\R\times[-1,1],
\]
with periodic identification in the second coordinate, that is, 
$(t,-1)\equiv (t,1)$. We endow this manifold with the flat metric. We will refer to this as the infinite cylinder 
throughout the paper. For coupling
$b>0$ and cosmological constant $\mu>0$, the formal timelike Liouville action on $M$ under Lorentzian signature is
\begin{equation*}
S(\phi) :=\frac{1}{4\pi}\int_M
\biggl\{-(\partial_t\phi)^2+(\partial_x\phi)^2 - 4\pi\mu e^{2b\phi}\biggr\}
\,\dd t\, \dd x.
\end{equation*}
The objective is to calculate vacuum expectation values of the form
\begin{align*}
\inn{F} := \int F(\phi) e^{\i S(\phi)}\,\md \phi,
\end{align*}
where $\int \ldots \md \phi$ denotes uniform integration over the space of all fields. 
Our theory will allow computation of such expectations for $F$ that are linear combinations of functions of the form
\begin{align}\label{eq:obsvbles}
  F(\phi) = \prod_{j=1}^k e^{2\alpha_j \phi(t_j,x_j)}, \quad \Re(\alpha_j)> -\frac{1}{2b} \text{ for each } j, \quad  w:= -\frac{\sum_{j=1}^k \alpha_j}{b}\in \Z.
\end{align}
The integer condition is the natural screening charge regime for the present construction. In the
finite-volume regularized theory, after the $\i\phi$ rotation and removal of the
zero mode regulator, only the term with exactly $w$ screening insertions
survives; this is the origin of the $w$-fold integral formula below. For the
classical screening perspective in Liouville theory, see~\cite{GoulianLi1991}.
In the smeared observable algebra constructed later, the individual charges are
restricted to integer multiples of $b$. Then the total screening number may also
be a negative integer; in that case the same zero-mode argument shows that both
the Euclidean and Lorentzian correlators vanish.
Throughout, we use the usual empty-index conventions: sums over empty index
sets are $0$, products over empty index sets are $1$, and any $0$-fold integral
(for example over $M^0$, $\R^0$, $[-1,1]^0$, or $\gamma^0$) is interpreted as $1$.

\subsection{Euclidean correlators}
To rigorously define the vacuum expectation values, we follow the usual route of first constructing the Euclidean theory and then applying analytic continuation to the time coordinate. The Euclidean action is obtained by  replacing $t$ by $-\i t$ in the action; this results in $\partial_t \phi$ being replaced by $\i \partial_t \phi$ and $\dd t$ being replaced by $-\i \dd t$. With these replacements, the unnormalized Euclidean vacuum expectation value of a functional $F$ is 
\begin{align*}
\inn{F}_{\mathrm{E}} = \int F(\phi) e^{- S_{\mathrm{E}}(\phi)}\,\md \phi,
\end{align*}
where $S_{\mathrm{E}}$ is the Euclidean action
\begin{equation*}
S_{\mathrm{E}}(\phi) :=\frac{1}{4\pi}\int_M
\biggl\{-(\partial_t\phi)^2-(\partial_x\phi)^2 + 4\pi\mu e^{2b\phi}\biggr\}
\,\dd t\, \dd x.
\end{equation*}
We define the Euclidean theory on $M$ by first defining it on the finite torus $M_T := [-T,T] \times [-1,1]$, with periodic boundaries in both coordinates. The expectations are defined for this finite-volume theory after renormalizing the exponential term in the action, as well as the exponentials in the observables, as follows. The term $e^{2b \phi(t,x)}$ in the action is renormalized by replacing it with 
\[
  \exp(2b \phi(t,x) + 2b^2 g_T(0,0)),
\] 
where $g_T$ is the mean-zero Green's function on $M_T$ (see~\secref{sec:green} for details). More generally, the observable $e^{2\alpha \phi(t,x)}$ is renormalized by replacing it with 
\begin{align*}
 \exp\biggl(2\alpha \phi(t,x) + 2\alpha^2 g_T(0,0)+\frac{1}{3}\pi T\alpha(b-\alpha)\biggr).
\end{align*}
Even after restricting to finite volume and inserting these renormalizations, a fundamental obstruction remains: the kinetic term in the Euclidean action has the ``wrong sign'', preventing us from defining the Euclidean theory as a proper stochastic model. This is overcome by replacing $\phi$ with $\i \phi$, a step that is rigorously justified by the theory developed in \cite{chatterjee25}. After these steps and the additional limiting argument in \secref{sec:math}, we can take $T\to \infty$ and obtain the following result for the Euclidean theory on the cylinder $M$.
\begin{thm}[Euclidean correlators]\label{thm:euc}
For $(t,x)\in \R^2$ with $(t,x)\notin \{0\} \times 2\Z$, let 
\[
  g(t,x) := -\frac{\pi |t|}{2} - \log |1-e^{-\pi |t| + \pi \i x}|.
\]
Take any $k$ and any distinct $q_1,\ldots,q_k\in M$. Take any $\alpha_1,\ldots, \alpha_k\in \C$ such that $\Re(\alpha_j)>-\frac{1}{2b}$ for each $j$, and $w := -\frac{1}{b}\sum_{j=1}^k \alpha_j$ is an integer. If $w$ is a nonnegative integer, then after renormalizing as above and taking $T\to\infty$, the Euclidean timelike Liouville theory on the cylinder $M$ has 
\begin{align*}
&\biggl\langle \prod_{j=1}^k e^{2\alpha_j \phi(q_j)}\biggr\rangle_{\mathrm{E}} \\
&= \frac{(-\mu)^w}{w!}
\exp\biggl(
-4\sum_{1\le j<j'\le k}\alpha_j\alpha_{j'} g(q_j-q_{j'})
\biggr) \\ 
&\qquad \cdot \int_{M^w} \exp\biggl(
    -4b\sum_{j=1}^k\sum_{l=1}^w \alpha_j g(q_j-u_l)
    -4b^2\sum_{1\le l<l'\le w}g(u_l-u_{l'})
  \biggr)\,\dd u_1\cdots \dd u_w,
\end{align*}
where the integral is interpreted as $1$ if $w=0$, and is  absolutely convergent otherwise. If $w$ is a negative integer, then the same procedure gives
\[
\biggl\langle \prod_{j=1}^k e^{2\alpha_j \phi(q_j)}\biggr\rangle_{\mathrm{E}} = 0.
\]
\end{thm}
This theorem, proved in \secref{sec:math}, provides the Euclidean input for the Lorentzian continuation carried out next.

\subsection{Lorentzian correlators}
We now want to use the Euclidean theory to construct the theory in Lorentzian signature, by analytically continuing in the time coordinate to replace $t$ by $\i t$. For this, let us first define 
\begin{align}\label{eq:ceuc}
C_{\alpha_1,\ldots,\alpha_k}(t_1,x_1;\ldots; t_k,x_k) := \biggl\langle \prod_{j=1}^k e^{2\alpha_j \phi(t_j,x_j)}\biggr\rangle_{\mathrm{E}},
\end{align}
where the right side is given by Theorem \ref{thm:euc} when
$w:= -b^{-1}\sum_{j=1}^k\alpha_j$ is an integer.
The standard way to analytically continue in the time coordinate is to analytically continue the above function of $(t_1,\ldots,t_k)$ (fixing $\alpha_1,\ldots,\alpha_k$ and $x_1,\ldots,x_k$) to the region 
\begin{align}\label{eq:omegakplus}
  \Omega_k := \{(\tau_1,\ldots,\tau_k)\in \C^k:&\Re(\tau_1)<\cdots < \Re(\tau_k)\},
\end{align}
and then use continuity to obtain $ C_{\alpha_1,\ldots,\alpha_k}(\i t_1,x_1;\i t_2,x_2;\ldots; \i t_k,x_k)$, by taking $(\tau_1,\ldots,\tau_k)$ to $(\i t_1,\ldots, \i t_k)$ through the region $\Omega_k$. If this can be done, it gives us the Lorentzian expectation $C_{\alpha_1,\ldots,\alpha_k}^{\mathrm{L}}(t_1, x_1;\ldots;t_k,x_k)$ associated with the index ordering $1,\ldots,k$. In general this Lorentzian
boundary value is not symmetric under permutations of the insertions;
permuting the labels corresponds to taking the boundary value from the
correspondingly permuted ordered tube.

The first step in this direction is to write down the analytic continuation of the function $g$ in the time coordinate. It turns out (shown in \secref{sec:anagreen}) that $g$ can be continued analytically in the time coordinate in the region $\Re(\tau)>0$, as  
\begin{align}\label{eq:gdefine0}
g(\tau,x) &= -\frac{\pi\tau}{2}-\frac{1}{2}\log(1-e^{-\pi \tau+\pi \i x}) - \frac{1}{2}\log(1-e^{-\pi \tau-\pi \i x}),
\end{align}
where $\log$ denotes the analytic branch of logarithm on $\C \setminus(-\infty,0]$ that is real-valued on the real line. This extends to the imaginary axis by continuity, as  
\begin{align}\label{eq:gimag}
g(\i t,x)
=
-\frac{\pi\i t}{2}
-\frac{1}{2}\log(1-e^{-\pi\i(t+x)})
-\frac{1}{2}\log(1-e^{-\pi\i(t-x)}),
\end{align}
leaving it undefined at $(\i t,x)$ where $t\pm x\in 2\Z$. We will now use this $g$ to construct the
analytic continuation of $C_{\alpha_1,\ldots,\alpha_k}$.

In the nonnegative screening sector, fix $\alpha_1,\ldots,\alpha_k$ satisfying $\Re(\alpha_j)> -\frac{1}{2b}$ for each $j$, and let $w := -\frac{1}{b}\sum_{j=1}^k \alpha_j$. Take $q_1,\ldots,q_k\in M$ such that 
\begin{align}\label{eq:lightcond}
q_j- q_{j'} \notin \{(t,x)\in \R^2: t+x\in 2\Z \text{ or } t - x\in 2\Z\} \text{ for all } 1\le j<j'\le k.
\end{align}
In other words, none of the differences $q_j -q_{j'}$ lie on the boundary of the light cone on the cylinder. 
We will henceforth refer to this as the ``non-light-cone condition''.

Let $(t_j, x_j)$ denote the coordinates of $q_j$. 
Let $\gamma$ be the contour in $\C$ that comes in horizontally from $\i t_1 - \infty$ to $\i t_1$, then moves vertically to $\i t_2$, then to $\i t_3$, and so on, until it reaches $\i t_k$. Then, it continues horizontally to $\i t_k +\infty$. We are not assuming that the $t_j$'s are in increasing or decreasing order; so the contour may change directions multiple times while on the imaginary axis. A schematic is shown in Figure~\ref{fig:gamma-contour}. 

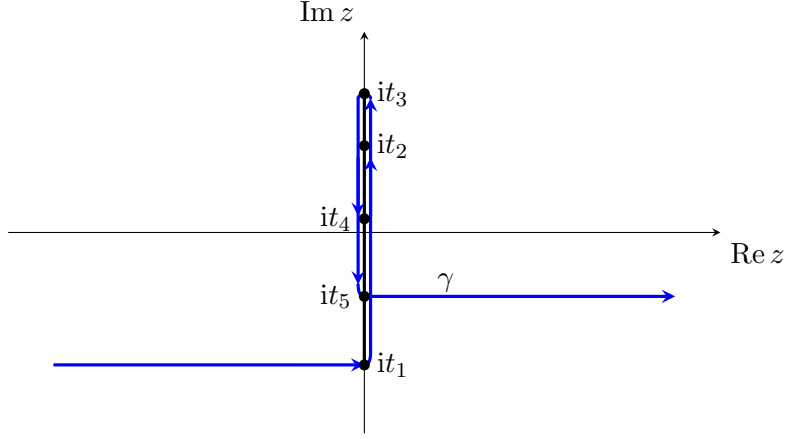
\begin{figure}[t!]
\centering
\begin{tikzpicture}[x=1.3cm,y=1.05cm,>=stealth,scale=1.15]
  \draw[->] (-3.15,0) -- (3.15,0) node[below right] {$\Re z$};
  \draw[->] (0,-2.2) -- (0,2.2) node[above left] {$\Im z$};

  \coordinate (T1) at (0,-1.45);
  \coordinate (T2) at (0,0.95);
  \coordinate (T3) at (0,0.15);
  \coordinate (T4) at (0,1.52);
  \coordinate (T5) at (0,-0.7);

  \draw[very thick,blue,->] (-2.75,-1.45) -- (T1);
  \draw[line width=1.2pt,line cap=round] (T1) -- (T4);
  \draw[very thick,blue,->] (T5) -- (2.75,-0.7);

  \draw[very thick,blue,line cap=round] (T1) to[out=0,in=-90] (0.055,-1.32);
  \draw[very thick,blue,->] (0.055,-1.32) -- (0.055,0.82);
  \draw[very thick,blue,->] (-0.055,0.82) -- (-0.055,0.18);
  \draw[very thick,blue,->] (0.055,0.18) -- (0.055,1.47);
  \draw[very thick,blue,line cap=round] (0.055,1.47) to[out=90,in=90] (-0.055,1.47);
  \draw[very thick,blue,->] (-0.055,1.47) -- (-0.055,-0.57);
  \draw[very thick,blue,line cap=round] (-0.055,-0.57) to[out=-90,in=180] (T5);

  \fill (T1) circle (1.8pt) node[right=1pt] {$\i t_1$};
  \fill (T2) circle (1.8pt) node[right=1pt] {$\i t_2$};
  \fill (T3) circle (1.8pt) node[left=1pt] {$\i t_4$};
  \fill (T4) circle (1.8pt) node[right=1pt] {$\i t_3$};
  \fill (T5) circle (1.8pt) node[left=1pt] {$\i t_5$};
  \node at (0.72,-.55) {$\gamma$};
\end{tikzpicture}
\caption{Schematic contour $\gamma$ used in the definition of Lorentzian correlations, illustrated here for $k=5$ with $t_1<t_5<t_4<t_2<t_3$. The contour enters horizontally at $\i t_1$, visits $\i t_2,\i t_3,\i t_4,\i t_5$ in index order, and exits horizontally from $\i t_5$. The contour itself is continuous and lies on the imaginary axis; the tiny left-right offsets of the blue vertical arrows are only a visual aid to distinguish upward and downward traversals along the same axis.}
\label{fig:gamma-contour}
\end{figure}

Fix a piecewise $C^1$ parametrization of $\gamma$ with $s$ going from $-\infty$ to $\infty$ and $\gamma'(s)\ne 0$ for all $s$ in the interiors of the pieces. For $u = (\gamma(s),x)$ and $u' = (\gamma(s'),x')$, where $x,x'\in [-1,1]$, define the \emph{time-ordered difference}
\[
  \mathcal{T}(u-u') :=
  \begin{cases}
  u-u' &\text{ if } s> s',\\ 
  u' - u &\text{ if } s< s'.
  \end{cases}
\]
On the null set where two screening variables have the same contour parameter, the value may be assigned arbitrarily, say $u-u'$. Let $v_j := (\i t_j,x_j)$ for $j=1,\ldots,k$. The external insertions are treated as marked visits of the contour, ordered by their labels; this breaks ties when two labels have the same contour parameter. Thus, for external points,
\[
  \mathcal{T}(v_j-v_{j'})=
  \begin{cases}
  v_{j'}-v_j &\text{if } j<j',\\
  v_j-v_{j'} &\text{if } j>j',
  \end{cases}
\]
and the same marked order is used for mixed external-screening differences, with arbitrary choices on the measure-zero set of ties. It is important to note that $\mathcal{T}(u-u')$ is not determined solely by the points $u,u'$; we actually need to know their contour parameters (and, for external insertions, their labels), since $\gamma$ is not necessarily injective as a function.

Next, let $M_\gamma:= \gamma\times [-1,1]$. We define contour integration on the tube $M_\gamma^w$ as 
\begin{align*}
\int_{M_\gamma^w} &f(u_1,\ldots,u_w)\, \dd u_1 \cdots \dd u_w \\ 
&:= \int_{[-1,1]^w}\int_{\R^w} f((\gamma(s_1),b_1),\ldots,(\gamma(s_w),b_w)) \\ 
&\hskip2in \cdot \gamma'(s_1)\cdots \gamma'(s_w)\, \dd s_1\cdots \dd s_w\, \dd b_1 \cdots \dd b_w.
\end{align*}
It is not difficult to see that this definition is independent of the parametrization of $\gamma$. 
With these notations, we define 
\begin{align}\label{eq:clorentz}
C^{\mathrm{L}}_{\alpha_1,\ldots,\alpha_k}(v_1,\ldots,v_k) &:= \frac{(-\mu)^w}{w!}
\exp\biggl\{
-4\sum_{1\le j<j'\le k}\alpha_{j}\alpha_{j'} g(\mathcal{T}(v_j - v_{j'}))
\biggr\} \notag \\ 
&\qquad \qquad \cdot \int_{M_\gamma^w} \exp\biggl\{
    -4b\sum_{j=1}^k\sum_{l=1}^w \alpha_j g(\mathcal{T}(v_j - u_l)) \notag\\ 
&\qquad \qquad \qquad  \qquad   -4b^2\sum_{1\le l<l'\le w}g(\mathcal{T}(u_l-u_{l'}))
  \biggr\}\,\dd u_1\cdots \dd u_w.
\end{align}
This is our candidate for the Lorentzian correlator. The following theorem shows that this prescription  gives the right answer. 
\begin{thm}[Lorentzian correlators]\label{thm:main}
Take $\alpha_1,\ldots,\alpha_k\in\C$ such that
$\Re(\alpha_j)>-\frac{1}{2b}$ for each $j$, and suppose that $w:=-\frac{1}{b}\sum_{j=1}^k\alpha_j$ 
is an integer. If $w\ge0$, then the Euclidean correlator 
$C_{\alpha_1,\ldots,\alpha_k}(t_1, x_1;\ldots; t_k,x_k)$, as a function of
$t_1,\ldots,t_k$, admits a unique analytic continuation to the region
$\Omega_k$ defined in equation \eqref{eq:omegakplus}. This continuation extends
continuously to the domain of $C_{\alpha_1,\ldots,\alpha_k}^{\mathrm{L}}$, and
the result is given by the function displayed in equation~\eqref{eq:clorentz}, and the integral appearing in the formula is absolutely convergent. If $w< 0$, then the same procedure gives $C_{\alpha_1,\ldots,\alpha_k}^{\mathrm{L}}\equiv 0$.
\end{thm}
The case $w<0$ follows trivially from Theorem \ref{thm:euc}; the
nonnegative case is proved in \secref{sec:anacont}. Very briefly, the idea of
the proof is that we decompose the Coulomb integral in the Euclidean correlator into ordered sectors, prove
analytic continuation sector by sector, reassemble the result as a $w$-fold contour integral on the ordered tube, and finally identify the Lorentzian boundary value through the contour representation in equation~\eqref{eq:clorentz}.

\subsection{Locality for vertex operators}\label{sec:vertexcomm}

By Theorem~\ref{thm:euc}, the Euclidean correlator is symmetric under
simultaneous permutations of the labeled insertions
\((\alpha_1,q_1),\ldots,(\alpha_k,q_k)\). The Lorentzian correlator is
different. It is defined as an ordered boundary value from the ordered tube
\(\Omega_k\), so permuting the labeled insertions means taking the boundary
value from a different ordered tube. In this sense,
\(C^{\mathrm L}_{\alpha_1,\ldots,\alpha_k}(u_1,\ldots,u_k)\) is naturally an
ordered correlator rather than an a priori symmetric function of its
arguments. Indeed, we will see in an explicit counterexample in the next subsection which shows that it may not be symmetric in its arguments.

For this reason, and also because exponential fields are the standard local
observables in Liouville theory, it is more natural to use the language of
vertex operators. We define a vertex operator 
\[ 
  V_\alpha(t,x):=e^{2\alpha\phi(t,x)},
\]
so that we have the formal identity
\begin{align*}
C^{\mathrm L}_{\alpha_1,\ldots,\alpha_k}(u_1,\ldots,u_k) = \inn{V_{\alpha_1}(u_1)\cdots V_{\alpha_k}(u_k)}.
\end{align*}
Since these pointlike objects are formal,
by saying that two  vertex operators commute we mean that exchanging
their consecutive positions in the ordered Lorentzian correlator leaves the
value unchanged. 

Whether or not two operators commute is tied to one of the central features of relativistic quantum field theories, known variously as  \emph{microcausality}, \emph{locality}, or \emph{Einstein causality}. This principle says that observables localized in spacelike separated regions
should commute. Conceptually, it encodes the compatibility of operations
performed in causally disjoint regions and rules out ``faster than light'' influence between them; in the algebraic approach it is one of the basic structural
axioms of a local net of observables~\cite{haagkastler64,haag96}.

We now specialize this locality principle to the Lorentzian cylinder.
On the cylinder \(M=\R\times\mathbb S^1\), two points \(u=(t,x)\) and
\(v=(s,y)\) are said to be spacelike separated~if
\[
  |t-s|<d_1(x-y),
  \qquad
  d_1(\xi):=\min_{m\in\Z}|\xi+2m|.
\]
Equivalently, after lifting the points to the universal cover so that the
spatial separation is minimal, their Minkowski square is positive. Thus, in
the present geometry, locality asks for commutation whenever the time
separation is smaller than the shortest spatial separation around the circle. The following theorem, proved in \secref{sec:locproof}, shows that this holds for our theory. In words, it
says that if two neighboring vertex operators are spacelike separated,
then exchanging their order does not change the Lorentzian correlator.

\begin{thm}[Locality for vertex operators]
\label{thm:microcausality}
Two spacelike separated vertex operators commute. Explicitly, this means the following. Take any $k\ge 2$, $\alpha_1,\ldots,\alpha_k\in\C$, and $v_1,\ldots,v_k\in M$ satisfying the conditions of Theorem \ref{thm:main}. Suppose that for some $1\le p\le k-1$, $v_p$ and $v_{p+1}$ are spacelike separated on the cylinder. Then
\begin{align}\label{eq:commute}
  &C^{\mathrm L}_{\alpha_1,\ldots,\alpha_{p-1},\alpha_p,
  \alpha_{p+1},\alpha_{p+2},\ldots,\alpha_k}
  (v_1,\ldots,v_{p-1},v_p,v_{p+1},v_{p+2},\ldots,v_k)
  \notag \\
  &\qquad =
  C^{\mathrm L}_{\alpha_1,\ldots,\alpha_{p-1},\alpha_{p+1},
  \alpha_p,\alpha_{p+2},\ldots,\alpha_k}
  (v_1,\ldots,v_{p-1},v_{p+1},v_p,v_{p+2},\ldots,v_k).
\end{align}
\end{thm}

The AQFT-type version of this result, formulated for smeared observables and the represented local net, is stated later in \secref{sec:einstein-causality}.

\subsection{Two simple exact Lorentzian calculations}

We record two explicit consequences of the general formula, since they
give useful benchmarks for Lorentzian timelike Liouville theory.

\paragraph{Neutral two-point function and timelike exchange phase.}
Let \(u=(t,x)\), assume \(t\pm x\notin 2\mathbb Z\), and take
\(0<\alpha<\frac{1}{2b}\).  In the neutral sector
\(\alpha+(-\alpha)=0\), the screening number is \(w=0\).  Therefore
Theorem~\ref{thm:main} gives
\begin{align*}
G^{12}_\alpha(t,x)
&:=
C^{\mathrm L}_{\alpha,-\alpha}((0,0),(t,x))
=
e^{4\alpha^2 g(\i t,x)},\\ 
G^{21}_\alpha(t,x)
&:=
C^{\mathrm L}_{-\alpha,\alpha}((t,x),(0,0))
= e^{4\alpha^2 g(-\i t,-x)}.
\end{align*}
To compare the two orderings directly, let
\[
  L(y):=\log(1-e^{-\pi \i y}),
\]
with the principal branch of the logarithm, and for
$y\notin 2\mathbb Z$ let $r(y)\in(0,2)$ be the representative of $y$ modulo
$2$. Since $1-e^{\pi\i y}$ and $1-e^{-\pi\i y}$ are both in the open right half-plane and $\log a - \log b = \log (a/b)$ for $a,b$ in the open right half-plane, and $r(y)-1\in (-1,1)$, we have 
\[
  L(-y)-L(y)
  =\log\biggl(\frac{1-e^{\pi\i y}}{1-e^{-\pi\i y}}\biggr)
  =\log(-e^{\pi\i r(y)}) = \log(e^{\pi \i (r(y)-1)}) 
  =\pi\i(r(y)-1).
\] 
Using this and equation \eqref{eq:gimag}, we get
\[
  g(\i t,x)-g(-\i t,-x)
  =-\pi\i t+\frac{\pi\i}{2}(r(t+x)+r(t-x)-2).
\]
If \(0<t<x<1\), then \(t+x\in(0,2)\) and \(t-x\in(-1,0)\), so
\(r(t+x)=t+x\) and \(r(t-x)=t-x+2\). Hence
\[
  g(\i t,x)-g(-\i t,-x)=0,
\]
which gives 
\[
G^{12}_\alpha(t,x)=G^{21}_\alpha(t,x).
\]
If instead $0<x<t<2-x$, then \(0<t-x<t+x<2\). So in this case \(r(t+x)=t+x\) and \(r(t-x)=t-x\), and therefore,
\[
g(\i t,x)-g(-\i t,-x)=-\pi \i.
\]
This gives
\begin{align}\label{eq:vac1}
G^{12}_\alpha(t,x)=e^{-4\pi \i\alpha^2}G^{21}_\alpha(t,x).
\end{align}
Now, for any \(y\in(0,2)\) one has
\[
1-e^{-\pi \i y}=2\sin\biggl(\frac{\pi y}{2}\biggr)e^{\frac{1}{2}\pi \i (1-y)}.
\]
Thus, the principal logarithm satisfies
\[
\log(1-e^{-\pi \i y})
=
\log\biggl(2\sin\frac{\pi y}{2}\biggr)+\frac{\pi \i}{2}(1-y).
\]
When $0<x<t<2-x$, we can apply this in equation \eqref{eq:gimag} with \(y=t+x\) and \(y=t-x\) (since \(0<t\pm x<2\)), and get 
\[
g(\i t,x)
=
-\frac{\pi \i}{2}
-\frac{1}{2}\log\biggl(4\sin\frac{\pi(t+x)}{2}\sin\frac{\pi(t-x)}{2}\biggr).
\]
Thus, we get 
\begin{align}\label{eq:vac2}
G^{12}_\alpha(t,x)
=
  e^{-2\pi \i\alpha^2}
  \biggl(4\sin\frac{\pi(t+x)}{2}\sin\frac{\pi(t-x)}{2}\biggr)^{-2\alpha^2},
  \qquad 0<x<t<2-x.
\end{align}
By equations \eqref{eq:vac1} and \eqref{eq:vac2}, the vacuum commutator matrix element is
\begin{align*}
&\inn{[V_\alpha(0,0), V_{-\alpha}(t,x)]} = G^{12}_\alpha(t,x)-G^{21}_\alpha(t,x) \\ 
&\qquad =
\begin{cases}
0, & 0<t<x,\\[4pt]
-2\i \sin(2\pi \alpha^2)(4\sin\frac{\pi(t+x)}{2}\sin\frac{\pi(t-x)}{2})^{-2\alpha^2},
& 0<x<t<2-x.
\end{cases}
\end{align*}
This gives an explicit Lorentzian exchange phase while preserving
locality in the spacelike region. 

\paragraph{Exact \(w=1\) cylinder tadpole.}
Assume \(b^2<\frac12\), so that the insertion \(V_{-b}\) is allowed.
Then \(w=1\), and translation invariance gives a constant one-point
function
\[
C^{\mathrm L}_{-b}(u)
=
-\mu I_b,
\qquad u\in M,
\qquad
I_b
:=
\int_{\mathbb R\times[-1,1]} e^{4b^2 g(t,x)}\,\dd t\,\dd x .
\]
Since
\[
  g(t,x)=-\frac{\pi|t|}{2}-\log|1-e^{-\pi|t|+\pi \i x}|,
\]
we have
\[
  e^{4b^2g(t,x)}=e^{-2\pi b^2|t|}|1-e^{-\pi|t|+\pi \i x}|^{-4b^2}.
\]
The integrand is even in \(t\), so with \(q=e^{-\pi t}\),
\[
I_b
=
2\int_0^\infty e^{-2\pi b^2 t}
\biggl(\int_{-1}^{1}
|1-qe^{\pi \i x}|^{-4b^2}\,\dd x\biggr)\,\dd t.
\]
To compute the inner integral, set \(\theta=\pi x\). Then
\[
\int_{-1}^{1}
|1-qe^{\pi \i x}|^{-4b^2}\,\dd x
=
\frac{1}{\pi}\int_{-\pi}^{\pi}(1-qe^{\i \theta })^{-2b^2}(1-qe^{-\i \theta })^{-2b^2}\,\dd \theta.
\]
For \(|q|<1\), the binomial series gives
\[
  (1-qe^{\pm \i\theta})^{-2b^2}
  =
  \sum_{n=0}^{\infty}\frac{(2b^2)_n}{n!}q^n e^{\pm \i n\theta},
\]
with absolute convergence. Here $(a)_n := a(a+1)\cdots (a+n-1)$ denotes the rising Pochhammer symbol. Multiplying the two series and integrating term by
term, only the diagonal terms survive. Thus
\[
\begin{aligned}
\int_{-1}^{1}
|1-qe^{\pi \i x}|^{-4b^2}\,\dd x
&=
\frac{1}{\pi}
\sum_{m,n\ge0}
\frac{(2b^2)_m(2b^2)_n}{m!n!}q^{m+n}
\int_{-\pi}^{\pi}e^{\i(m-n)\theta}\,\dd \theta \\
&=
2\sum_{n=0}^{\infty}\frac{(2b^2)_n^2}{(n!)^2}q^{2n}
=
2\,{}_2F_1(2b^2,2b^2;1;q^2),
\end{aligned}
\]
where the last equality is the defining series of the Gauss hypergeometric
function; see, for example, \cite[Ch.~2]{andrewsetal99} for the standard
notation. Therefore
\[
I_b
=
4\int_0^\infty
e^{-2\pi b^2 t}
\,{}_2F_1(2b^2,2b^2;1;e^{-2\pi t})\,\dd t .
\]
After the change of variables \(y=e^{-2\pi t}\), so that
\(\dd t=-(2\pi y)^{-1}\,\dd y\), this becomes
\[
I_b
=
\frac{2}{\pi}
\int_0^1
y^{b^2-1}\,
{}_2F_1(2b^2,2b^2;1;y)\,\dd y .
\]
Expanding \({}_2F_1\) once more and integrating term by term gives
\[
\begin{aligned}
I_b
&=
\frac{2}{\pi}
\sum_{n=0}^{\infty}
\frac{(2b^2)_n^2}{(n!)^2}
\int_0^1 y^{n+b^2-1}\,\dd y \\
&=
\frac{2}{\pi}
\sum_{n=0}^{\infty}
\frac{(2b^2)_n^2}{(n!)^2}\frac{1}{n+b^2}.
\end{aligned}
\]
Using
\[
  \frac{1}{n+b^2}
  =
  \frac{1}{b^2}\frac{(b^2)_n}{(b^2+1)_n},
\]
we can rewrite this as
\[
I_b
=
\frac{2}{\pi b^2}\,
{}_3F_2\biggl(
\begin{matrix}
2b^2,\;2b^2,\;b^2\\
1,\;b^2+1
\end{matrix}
;1
\biggr),
\]
where ${}_3F_2$ denotes the generalized hypergeometric series. This series converges
precisely in the range \(b^2<\frac{1}{2}\) (see \cite[Theorem 2.1.2]{andrewsetal99}). Therefore, for $b < \frac{1}{\sqrt{2}}$, we get the result
\[
C^{\mathrm L}_{-b}(u)
=
-\frac{2\mu}{\pi b^2}\,
{}_3F_2\biggl(
\begin{matrix}
2b^2,\;2b^2,\;b^2\\
1,\;b^2+1
\end{matrix}
;1
\biggr).
\]

\subsection{Smeared observables}\label{sec:smear}
We will now construct ``smeared'' versions of the observables that we have been considering until now. The reason for smearing is that unsmeared pointlike observables need not have finite Lorentzian vacuum expectation values. One could try to keep only those unsmeared observables whose expectations are finite, but that class would not be closed under multiplication: even if two pointlike observables separately have finite expectation values, their product may create additional coincident or light-cone-singular insertions and thereby lose finiteness. Smearing is what produces a collection of observables with well-defined expectations that is robust enough for the later algebraic construction. 

The smeared observables are defined as follows. Throughout this construction,
assume
\[
  0<b<\frac{1}{\sqrt{2}}.
\]
Let
\[
  \mathcal I_b:=\biggl\{n\in\Z: |n|<\frac{1}{\sqrt2 b}\biggr\}.
\]
The basic smeared observables are the integer-charge vertex operators
\[
  \mo_{f,n}(\phi):=\int_M f(u)V_{bn}(u)\,\dd u,
  \qquad f\in C_c^\infty(M,\C),\ n\in\mathcal I_b.
\]
For an integer tuple
$\boldsymbol n=(n_1,\ldots,n_k)\in\mathcal I_b^k$, write
\[
  C_{b\boldsymbol n}^{\mathrm L}(u_1,\ldots,u_k)
  :=C_{bn_1,\ldots,bn_k}^{\mathrm L}(u_1,\ldots,u_k),
\]
with the convention  that
this correlator is identically zero when
$w(\boldsymbol n):=-\sum_{j=1}^k n_j$ is a negative integer. In the nonnegative
sector the point correlator is understood on the non-light-cone domain and may
be assigned arbitrary values on the measure-zero exceptional set before
smearing. Naturally, for an
ordered product of basic smeared observables we define
\begin{align*}
\inn{\mo_{f_1,n_1}\cdots\mo_{f_k,n_k}}
&:=\int_{M^k}\prod_{j=1}^k f_j(u_j)\,
C_{b\boldsymbol n}^{\mathrm L}(u_1,\ldots,u_k)\,\dd u_1\cdots\dd u_k.
\end{align*}
The following theorem shows that these definitions make sense. The proof is in
\secref{sec:smearproof}.

\begin{thm}\label{thm:smear}
Assume $0<b<\frac{1}{\sqrt{2}}$. For any $k\ge1$, any
$n_1,\ldots,n_k\in\mathcal I_b$, and any
$f_1,\ldots,f_k\in C_c^\infty(M,\C)$, the integral defining $\inn{\mo_{f_1,n_1}\cdots\mo_{f_k,n_k}}$ 
is absolutely convergent. In particular, this vacuum expectation is finite.
\end{thm}

This absolute convergence statement provides the analytic foundation for the algebraic and representation theoretic constructions developed in the next subsections.

\subsection{AQFT construction without positivity}\label{sec:gnsquant}
We now extract from the Lorentzian vacuum expectation values an AQFT-type
structure on the cylinder. Because positivity is absent, the outcome is not a
Haag--Kastler theory in the usual positive sense: neither an abstract net of
local $C^*$-algebras nor, after Hilbert space representation, a net of local
von Neumann algebras. Instead, we obtain a
net of algebraic local cores with a Hermitian vacuum functional, together with
a complete locally convex state space carrying a continuous representation, the
operator-topological closures of the represented local observables, and the
cylinder translation action. The vertex operator analysis from
\secref{sec:vertexcomm} leads to a full locality statement at the represented
level: the locality theorem proved later in this section applies to
the represented local net itself. For comparison with the standard
AQFT framework, see~\secref{sec:app-aqft}.

\subsubsection{The global algebra and vacuum functional}
We begin by defining the ordered algebra (that is, an algebra in which the insertion order is part of the data) generated by the integer-charge smeared observables and the vacuum functional induced by the Lorentzian correlators.
Let $\mathcal F:=C_c^\infty(M,\C)$. For each $n\in\mathcal I_b$, let
$\mathcal G_n$ be a copy of $\mathcal F$, written in the form
$f\mapsto\mo_{f,n}$, modulo linearity in $f$, and set
\[
  \mathcal G:=\bigoplus_{n\in\mathcal I_b}\mathcal G_n.
\]
We define $\ma$ to be the tensor algebra $T(\mathcal G)$ with unit $1$. Thus
$\ma$ is the complex vector space spanned by $1$ and finite ordered monomials
\[
  \mo_{f_1,n_1}\cdots\mo_{f_k,n_k},
  \qquad k\ge1,
  \quad f_j\in\mathcal F,
  \quad n_j\in\mathcal I_b,
\]
modulo linearity in each $f_j$, and multiplication is concatenation of words.
For $f\in\mathcal F$, write $\bar f(u):=\overline{f(u)}$, and define the
involution on generators by
\[
  1^*=1,
  \qquad
  \mo_{f,n}^*:=\mo_{\bar f,n},
\]
extending conjugate-linearly and anti-multiplicatively to all of $\ma$.
Then $(A^*)^* = A$ and $(AB)^*=B^*A^*$.  Next, define the vacuum functional $\omega:\ma\to\C$ by
\[
  \omega(1):=1,
\]
and, on monomials of positive length, by
\[
  \omega(\mo_{f_1,n_1}\cdots\mo_{f_k,n_k})
  :=\inn{\mo_{f_1,n_1}\cdots\mo_{f_k,n_k}},
\]
and finally, extending linearly. This is well-defined by Theorem~\ref{thm:smear}. The
associated sesquilinear form~is
\[
  \langle A,B\rangle_{\ma}:=\omega(A^*B),\qquad A,B\in\ma.
\]
This form is conjugate linear in the first argument and linear in the second. We adopt this definition since it is conventional  in physics to have conjugate linearity in the first argument. It turns out that this form is Hermitian, but the argument is nontrivial. The following theorem is proved in \secref{sec:hermitproof}.
\begin{thm}[Hermiticity of the algebraic form]\label{thm:ma-hermitian}
The form $\langle\cdot,\cdot\rangle_{\ma}$ is Hermitian. That is,
\[
  \langle A,B\rangle_{\ma}=\overline{\langle B,A\rangle_{\ma}}
  \qquad\text{for all }A,B\in\ma.
\]
\end{thm}

This supplies the Hermitian vacuum form that underlies the later state space and representation constructions.

\subsubsection{Cylinder translations and local algebras}\label{sec:aqft}
We now record the natural cylinder-translation action on the global algebra and use it to define the corresponding local subalgebras.
For $(a,\theta), (t,x)\in\R\times \mathbb S^1$, let
\[
  T_{a,\theta}(t,x):=(t+a,x+\theta),
\]
where the second coordinate is understood modulo $2$. For $f\in\mathcal F$,
define
\[
  (\eta_{a,\theta}f)(u):=f\bigl(T_{-a,-\theta}u\bigr),
\]
and on generators set
\[
  \eta_{a,\theta}\mo_{f,n}:=\mo_{\eta_{a,\theta}f,n},
  \qquad
  \eta_{a,\theta}1:=1,
\]
extending multiplicatively and linearly to all of $\ma$.
We write
\[
  \tau_t:=\eta_{-t,0},\qquad t\in\R,
\]
for the time-translation subgroup. For each nonempty open set $O\subset M$,
let $\ma(O)$ be the unital $*$-subalgebra of $\ma$ generated by $1$ and all
symbols $\mo_{f,n}$ with $\supp(f)\subset O$. We say that $\ma(O)$ is the \emph{local algebra} at $O$.

In the following result, a $*$-automorphism means a unital algebra automorphism that preserves the involution, and a normalized Hermitian functional on a unital $*$-algebra means a linear functional $\varphi$ such that $\varphi(1)=1$ and $\varphi(A^*)=\overline{\varphi(A)}$ for all $A$.

\begin{prop}[AQFT-type net on the cylinder]\label{prop:aqft-net}
The family $(\eta_{a,\theta})_{(a,\theta)\in\R\times \mathbb S^1}$ is an
action of $\R\times \mathbb S^1$ on $\ma$ by $*$-automorphisms, and $\omega(\eta_{a,\theta}A)=\omega(A)$ for all $A\in\ma$. If $O_1\subset O_2$, then $\ma(O_1)\subset \ma(O_2)$. Lastly, we have that $\eta_{a,\theta}\ma(O)=\ma(T_{a,\theta}O)$  for all $(a,\theta)\in\R\times \mathbb S^1$, and for each open set $O$, the restriction $\omega_O:=\omega|_{\ma(O)}$ is a normalized Hermitian functional on~$\ma(O)$.
\end{prop}
The proof of this result is in \secref{sec:timeev-details}. This provides the translation-covariant isotone family of local algebras that serves as the algebraic starting point for the AQFT-type construction developed below.

\subsubsection{The state space and the represented net}
We now use the vacuum functional $\omega$ and its associated Hermitian form
$\langle\cdot,\cdot\rangle_{\ma}$ to build an explicit state space and the
corresponding family of local operator algebras. In the usual positive GNS
construction, one would quotient by null vectors (that is, $A$ such that
$\langle A,A\rangle_{\ma}=0$) and complete in the induced norm,
but here the form is not assumed positive, and is shown below to be indefinite for $b<8^{-1/2}$; it therefore does not furnish a canonical Hilbert norm. Instead, we realize
each observable $A$ through its vacuum matrix coefficients $(\omega(B^*A))_{B\in \ma}$, take the resulting
closure inside the product space $\C^{\ma}$, and use the algebra action to
obtain the representation. This yields a complete locally convex topological vector space with a
dense algebraic core, together with the represented local algebras that form
the concrete AQFT-type net.

To begin, consider $\C^{\ma}$ with the product topology, and define
\[
  \Phi:\ma\to\C^{\ma},
  \qquad
  \Phi(A):=(\omega(B^*A))_{B\in\ma}.
\]
The map $\Phi$ is linear. Set
\[
  \mh_0:=\Phi(\ma),
  \qquad
  \mh:=\overline{\mh_0}\subset\C^{\ma}.
\]
Since $\C^{\ma}$ is complete, $\mh$ is a complete Hausdorff topological vector
space. (See \secref{sec:app-tvs} for the basics of topological nets and completeness in topological vector spaces.) Define a pairing on $\mh_0\times\mh$ by
\[
  (\Phi(A),z):=z_A,
  \qquad
  A\in\ma,
  \ z=(z_B)_{B\in\ma}\in\mh.
\]
The next theorem records the basic structural properties of the state space:
it shows that $\mh_0$ is dense in $\mh$, that the above formula gives a
well-defined sesquilinear pairing, and that this pairing recovers the original
Hermitian form on algebraic vectors. Its proof is deferred to
\secref{sec:innlmmproof}.
\begin{thm}[Properties of the state space]\label{thm:innlmm}
The space $\mh_0$ is dense in $\mh$, and the above formula defines a
well-defined sesquilinear pairing on $\mh_0\times\mh$, conjugate linear in the first
argument and linear in the second. For $A,B\in\ma$, $(\Phi(A),\Phi(B))=\langle A,B\rangle_{\ma}$. 
Consequently, the restriction of the pairing to $\mh_0$ is a 
nondegenerate Hermitian form. More generally, if $(u,z)=0$ for all $u\in\mh_0$, then
$z=0$.
\end{thm}

The next proposition shows that the Hermitian form on $\mh_0$ is genuinely
indefinite when $b$ is small enough: it has vectors of positive and of negative squared norm. The proof is in  \secref{sec:innlmmproof}. It is possible that the result is true for any $b$, but we do not attempt to prove it here.
\begin{prop}[The form on $\mh_0$ is indefinite]\label{prop:mh0-indefinite}
Suppose that $b\in (0,8^{-1/2})$. Then there exist vectors $u_+,u_-\in\mh_0$ such that $(u_+,u_+)>0$  and $(u_-,u_-)<0$. In particular, the Hermitian form on $\mh_0$ is neither positive semidefinite nor negative
semidefinite.
\end{prop}

For each $C\in\ma$, define $\pi(C):\mh \to \mh$ as 
\[
  \pi(C)(z_A)_{A\in\ma}:=(z_{C^*A})_{A\in\ma}.
\]
We now verify that these coordinate-shift operators realize the observable
algebra concretely on the state space. More precisely, we will prove that
$C\mapsto\pi(C)$ is a unital algebra representation of $\ma$ on $\mh$, that
is, an algebra homomorphism
\[
  \pi:\ma\to\mathrm{End}_{\mathrm{cont}}(\mh),
\]
where $\mathrm{End}_{\mathrm{cont}}(\mh)$ denotes the algebra of continuous
linear endomorphisms of $\mh$, with $\pi(1)=\mathrm{id}_{\mh}$. We will also identify the distinguished
vector $\Omega:=\Phi(1)$ and show that it is cyclic, meaning that the linear
span of $\{\pi(A)\Omega:A\in\ma\}$ is dense in $\mh$; in fact, the theorem
below proves the stronger statement that this span is exactly $\mh_0$, which
is dense in $\mh$. Finally, the theorem shows that $\Omega$ reproduces the
vacuum functional through the represented observables. Its proof is given in
\secref{sec:innlmmproof}.
\begin{thm}[Representation and vacuum vector]\label{thm:obsrep}
The map $C\mapsto\pi(C)$ is a unital algebra representation of $\ma$ by
continuous linear maps on $\mh$. Moreover:
\begin{enumerate}
\item $\pi(C)$ preserves $\mh_0$, and for every $A\in\ma$, $\pi(C)\Phi(A)=\Phi(CA)$.
\item For every $u\in\mh_0$ and $z\in\mh$, $(\pi(C)u,z)=(u,\pi(C^*)z)$.
\item If we set $\Omega:=\Phi(1)\in\mh_0$, then $\operatorname{span}\{\pi(C)\Omega:C\in\ma\}=\mh_0$. Thus, $\Omega$ is cyclic.
\item For every $A\in\ma$, $\omega(A)=(\Omega,\pi(A)\Omega)$.
\end{enumerate}
\end{thm}

With the representation on $\mh$ now in hand, we can pass from the algebraic
local subalgebras to local operator algebras acting on the state
space. We do this by taking the represented images of $\ma(O)$ and closing
them in the natural operator topology. Recall first that a seminorm on $\mh$ is a map
\[
  p:\mh\to[0,\infty)
\]
such that $p(\lambda z)=|\lambda|p(z)$ and
$p(z+w)\le p(z)+p(w)$ for all $z,w\in\mh$ and $\lambda\in\C$. A subset
$B\subset\mh$ is called bounded if every continuous seminorm is bounded on
$B$, that is, if for every continuous seminorm $p$ on $\mh$ one has
\[
  \sup_{z\in B} p(z)<\infty.
\]
We denote by \(\mathcal T_b\) the topology of uniform convergence on bounded subsets
of \(\mh\) on \(\operatorname{End}_{\mathrm{cont}}(\mh)\). Since \(\mh\) is not a Hilbert space, \(\operatorname{End}_{\mathrm{cont}}(\mh)\) has no canonical
involution. We therefore use the involution supplied by the Hermitian pairing on the
algebraically adjointable operators. Define \(\operatorname{End}_{\dagger}(\mh)\) to be the set
of all \(T\in \operatorname{End}_{\mathrm{cont}}(\mh)\) for which there exists
\(S\in \operatorname{End}_{\mathrm{cont}}(\mh)\) such that
\[
T(\mh_0)\subseteq \mh_0,\qquad S(\mh_0)\subseteq \mh_0,
\]
and
\[
(Tu,z)=(u,Sz),\qquad (Su,z)=(u,Tz)
\]
for every \(u\in \mh_0\) and \(z\in \mh\). By the nondegeneracy statement in Theorem \ref{thm:innlmm},
such an \(S\), if it exists, is unique. We write \(T^\dagger:=S\). We equip \(\operatorname{End}_{\dagger}(\mh)\) with the graph-bounded topology
\(\mathcal T_\dagger\), defined by
\[
T_\nu\to T \text{ in } \mathcal T_\dagger
\quad\Longleftrightarrow\quad
T_\nu\to T \text{ and } T_\nu^\dagger\to T^\dagger
\text{ in } \mathcal T_b .
\]
With this involution and topology, \(\operatorname{End}_{\dagger}(\mh)\) is the natural
represented operator algebra in the present locally convex setting without positivity. The following lemma, proved in \secref{sec:innlmmproof}, lists some of its basic properties.
\begin{lmm}\label{lem:5x}
The space \(\operatorname{End}_{\dagger}(\mh)\) is a unital algebra with
involution \(T\mapsto T^\dagger\). Moreover, multiplication is separately continuous
for the graph-bounded topology \(\mathcal T_\dagger\), and for every \(C\in \mathcal A\), $\pi(C)\in \operatorname{End}_{\dagger}(\mh)$ and $\pi(C)^\dagger=\pi(C^*)$.
\end{lmm}
For each open set \(O\subseteq M\), define the represented local algebra by
\[
\mathfrak A_{\mathrm{loc}}(O)
:=
\overline{\pi(\mathcal A(O))}^{\,\mathcal T_\dagger}
\subseteq \operatorname{End}_{\dagger}(\mh).
\]
Equivalently, \(\mathfrak A_{\mathrm{loc}}(O)\) is the
\(\mathcal T_\dagger\)-closed unital \(\dagger\)-subalgebra of
\(\operatorname{End}_{\dagger}(H)\) generated by \(\pi(\mathcal A(O))\). In our locally convex setting without positivity,
this plays the role that the local von Neumann algebras play in standard AQFT
after passing to operator closures in a Hilbert space representation; the
remarks below explain this comparison in more detail.

\begin{remark}[Why pass to the closed represented local algebra?]
The passage from $\ma(O)$ to $\pi(\ma(O))$ puts the local observables on the
state space $\mh$, where covariance, dynamics, and pairings are actually
realized. Once this representation has been constructed, the algebraic image
$\pi(\ma(O))$ should then be viewed only as a local core: it need not
already contain limits, in the topology of uniform convergence on bounded
subsets of $\mh$, of nets of local observables. Passing to
$\mathfrak A_{\mathrm{loc}}(O)$ incorporates those limits while preserving the
local action on $\mh$, isotony, and covariance. This is the locally convex
analogue of the standard AQFT step, after constructing a Hilbert space
representation, of moving from the abstract local algebra to a closed local
operator algebra (usually, a von Neumann algebra) on the state space.
\end{remark}

\begin{remark}[How this completion differs from the AQFT completion]
In standard AQFT, one distinguishes two related levels. At the abstract level,
one works with local $C^*$-algebras. After choosing a Hilbert space
representation, one often passes to the associated local von Neumann algebras
by weak operator closure. The completion constructed here is of a different
kind. Because positivity is absent, the vacuum functional does not provide a
canonical $C^*$-seminorm on $\ma(O)$, so we do not complete the abstract local
algebra $\ma(O)$ in any $C^*$-norm. Instead, there are two distinct closure
operations. First, we complete the cyclic state space:
\[
  \mh=\overline{\Phi(\ma)}\subset\C^{\ma}
\]
with respect to the product topology. Second, after representing on $\mh$, we
pass from the algebraic image $\pi(\ma(O))$ to the closed unital $*$-subalgebra
$\mathfrak A_{\mathrm{loc}}(O)$ in the operator topology of uniform convergence
on bounded subsets of $\mh$. Thus the local algebras are closures of the
represented observables, but this closure is locally convex rather than
$C^*$- or von Neumann-based. In a positive theory, one would expect the GNS
Hilbert completion of states and then the usual von Neumann or $C^*$ closures
of the represented local algebras to recover the standard AQFT framework. Here
the corresponding closure lives naturally in the locally convex category rather
than the $C^*$ or von Neumann category.
\end{remark}

\subsubsection{Translation covariance on the state space}\label{sec:timeev}
Having constructed the state space \(\mh\), the representation \(\pi\), and the represented
local algebras \(\mathfrak A_{\mathrm{loc}}(O)\), we now implement cylinder translations on
the represented level. The same translation symmetry that acts on the algebraic observables
is realized on \(\mh\) by an algebraic action whose individual elements are continuous linear
homeomorphisms. We do not claim here that the orbit maps
\[
(a,\theta)\mapsto U_{\mathrm{cyl}}(a,\theta)z
\]
are continuous for every \(z\in \mh\), nor that the action is jointly continuous as a map
\((a,\theta,z)\mapsto U_{\mathrm{cyl}}(a,\theta)z\). 

For $(a,\theta)\in\R\times \mathbb S^1$ and $z\in\mh$, define
\[
  U_{\mathrm{cyl}}(a,\theta)z:=(z_{\eta_{-a,-\theta}A})_{A\in\ma}.
\]
We write
\[
  U(t):=U_{\mathrm{cyl}}(t,0),\qquad t\in\R,
\]
for the time-translation subgroup. The next theorem shows that these maps
implement the cylinder translations on the represented state space in a manner
compatible with the vacuum, the representation, and the represented local net.
In particular, the time subgroup provides the natural translation dynamics on
$\mh$. The proof is given in \secref{sec:timeev-details}.
\begin{thm}[Translation covariance]\label{thm:timeev-main-properties}
The family \((U_{\mathrm{cyl}}(a,\theta))_{(a,\theta)\in \mathbb R\times S^1}\) is an
algebraic action of \(\mathbb R\times S^1\) on \(\mh\), and each
\(U_{\mathrm{cyl}}(a,\theta)\) is a continuous linear homeomorphism of \(\mh\). Moreover:
\begin{enumerate}
\item For every $A\in\ma$, $U_{\mathrm{cyl}}(a,\theta)\Phi(A)=\Phi(\eta_{a,\theta}A)$. In particular, $\mh_0$ is invariant.
\item For every $u\in\mh_0$ and $z\in\mh$, $(U_{\mathrm{cyl}}(a,\theta)u,z) = (u,U_{\mathrm{cyl}}(-a,-\theta)z)$. As a consequence, $(U_{\mathrm{cyl}}(a,\theta)u,U_{\mathrm{cyl}}(a,\theta)v)=(u,v)$  for all $u,v\in\mh_0$.
\item The vacuum vector is invariant, i.e., $U_{\mathrm{cyl}}(a,\theta)\Omega=\Omega$. Moreover, for every $C\in\ma$,
\[
  U_{\mathrm{cyl}}(a,\theta)\pi(C)U_{\mathrm{cyl}}(a,\theta)^{-1}
  =
  \pi(\eta_{a,\theta}C).
\]
\item The represented local net is isotone and covariant, meaning that if $O_1\subset O_2$, then $\mathfrak A_{\mathrm{loc}}(O_1)\subset\mathfrak A_{\mathrm{loc}}(O_2)$, and
\[
  U_{\mathrm{cyl}}(a,\theta)\,\mathfrak A_{\mathrm{loc}}(O)\,U_{\mathrm{cyl}}(a,\theta)^{-1}
  =
  \mathfrak A_{\mathrm{loc}}(T_{a,\theta}O).
\]
\item The time subgroup satisfies \(U(t)\Phi(A)=\Phi(\tau_{-t}A)\) for all \(A\in\mathcal A\),
and \((U(t))_{t\in\mathbb R}\) is an algebraic one-parameter group whose individual
elements are continuous linear homeomorphisms of \(\mh\).
\end{enumerate}
\end{thm}

\begin{remark}[No strong-continuity assertion]
The continuity assertion in Theorem \ref{thm:timeev-main-properties}
concerns each operator \(U_{\mathrm{cyl}}(a,\theta)\) as a map \(\mh\to \mh\). We do not assert
that the orbit maps
\[
(a,\theta)\mapsto U_{\mathrm{cyl}}(a,\theta)z
\]
are continuous for arbitrary \(z\in \mh\), nor that the action is jointly continuous. On the
algebraic core \(\mh_0\), the matrix coefficients are continuous in the translation parameter:
for \(A,C\in\mathcal A\),
\[
\bigl(U_{\mathrm{cyl}}(a,\theta)\Phi(C)\bigr)_A
=
\omega(A^*\eta_{a,\theta}C),
\]
and this depends continuously on \((a,\theta)\) by the continuity of translations on test
functions and the dominated convergence estimates used for smeared correlators. Extending
this to all of \(\mh\) would require a stronger topology or an additional equicontinuity
argument, and is not used in the present construction.
\end{remark}

This furnishes the represented theory with its translation-covariant dynamical structure; we next turn to locality on the represented net.

\subsubsection{Locality for the represented net}\label{sec:einstein-causality}
As mentioned earlier, one of the central axioms of algebraic quantum field theory is
\emph{Einstein causality}, also commonly called \emph{microcausality} or \emph{locality}. The guiding
idea is that observables localized in spacelike separated regions should be
compatible: operations performed in one region should not influence
measurements in the other, and algebraically this is expressed by the
commutativity of the corresponding local observables. In a Haag--Kastler net
this requirement is formulated abstractly at the level of the local algebras
assigned to spacetime regions.

Recall the definition of spacelike separation on the cylinder from \secref{sec:vertexcomm}. For open sets \(O_1,O_2\subset M\), we say that they are spacelike separated
if every pair \(u\in O_1\), \(v\in O_2\) is spacelike separated in this
sense. The relevant local objects are the closed represented  local algebras $\mathfrak A_{\mathrm{loc}}(O)$ constructed above. In this
represented setting, \emph{locality} means that
whenever $O_1$ and $O_2$ are spacelike separated open sets, every operator in
$\mathfrak A_{\mathrm{loc}}(O_1)$ commutes with every operator in
$\mathfrak A_{\mathrm{loc}}(O_2)$. Because the Lorentzian correlators are
defined as ordered boundary values, commutativity is not automatic even in the
integer-charge sector: it must first be recovered from the spacelike exchange property proved for vertex operators in Theorem \ref{thm:microcausality}, then transferred to smeared observables, and finally
passed to the represented closures. The resulting statement is the following
theorem, whose proof is deferred to 
\secref{sec:locproof}.

\begin{thm}[Locality for the represented net]\label{thm:einstein-causality}
Let $O_1,O_2\subset M$ be spacelike separated open sets. Then $[S,T]=0$ for all $S\in\mathfrak A_{\mathrm{loc}}(O_1)$ and $T\in\mathfrak A_{\mathrm{loc}}(O_2)$.
\end{thm}

Next, we record the standard Haag--Kastler features that remain
unavailable or unproved in the present setting.

\subsubsection{What is still missing from the Haag--Kastler package}\label{sec:missing-haag-kastler}
The AQFT-type structure described above is assembled across several results
in this section. Theorem~\ref{thm:smear} makes the smeared vacuum functional well
defined; Theorem~\ref{thm:ma-hermitian} supplies the Hermitian algebraic form; Theorems~\ref{thm:innlmm}
and~\ref{thm:obsrep} construct the state space, pairing, representation, and cyclic vacuum
vector; Proposition~\ref{prop:aqft-net} together with
Theorem~\ref{thm:timeev-main-properties} establishes covariance
under cylinder translations; and Theorem~\ref{thm:einstein-causality} establishes locality 
for the represented local net. Beyond this
point, the following standard Haag--Kastler features remain unavailable or
unproved in the present locally convex setting without positivity:
\begin{enumerate}
\item \textit{Positivity and $C^*$-structure.} The local algebras here are
closed unital $*$-subalgebras of the continuous endomorphisms of a complete
locally convex space with a Hermitian pairing that is not assumed positive, obtained from the
algebraic images $\pi(\ma(O))$ by locally convex operator closure. Because the
vacuum functional is not positive, we do not obtain a net of $C^*$-algebras or
von Neumann algebras on a Hilbert space.
\item \textit{Spectrum condition and infinitesimal generators.} The time-translation subgroup acts on \(\mh\) by continuous linear homeomorphisms, but
strong continuity of the orbit maps on all of \(\mh \) is not asserted. Consequently, we do
not analyze an infinitesimal generator or prove any positive-energy statement. In
particular, we do not obtain a Hilbert space self-adjoint generator of time
translations.
\item \textit{Additivity, time-slice, and Reeh--Schlieder.} We have also not
addressed several further structural properties that, in standard AQFT,
describe how much of the theory is already encoded in sufficiently small
spacetime regions; see for example
\cite{haag96,brunettifredenhagenverch03,reehschlieder61}. Additivity says
roughly that if a region $O$ is covered by smaller open sets $O_\alpha\subset
O$, then the algebra assigned to $O$ is generated, in the relevant closure
sense, by the algebras of the $O_\alpha$. The time-slice axiom is a stronger
causal propagation statement: an arbitrarily thin neighborhood of a Cauchy
surface should already determine the algebra of its full domain of dependence.
The Reeh--Schlieder property says that the vacuum vector $\Omega$ is cyclic for
every nonempty local algebra, so that vectors of the form $A\Omega$ with $A$
localized in $O$ are dense already for each nonempty region $O$. None of these
statements is proved here for the represented net
$\mathfrak A_{\mathrm{loc}}(O)$ in our locally convex setting without positivity.
\end{enumerate}
For these reasons, the object constructed here is best regarded as an AQFT-type
net rather than a complete Haag--Kastler theory.

This completes the statement of the main results. The next section begins the proofs by constructing the Euclidean theory in finite and then infinite volume.

\section{Construction of the Euclidean theory}\label{sec:math}
In this section, we will carry out the mathematical steps in the construction of the Euclidean timelike Liouville field theory on the cylinder $M = \R \times [-1,1]$, culminating in the proof of Theorem \ref{thm:euc}. We will start by defining the Euclidean theory in the finite torus $M_T=[-T,T]\times [-1,1]$ with an ultraviolet cutoff and a regularization for the zero mode. We will then remove the cutoff and the regularization to get the Euclidean theory on the torus without regularization. Finally, we will send $T\to \infty$ to get the theory on the~cylinder.

\subsection{The regularized Euclidean theory in finite volume}\label{sec:regeuc}
In analogy with \cite[Section 1.4]{chatterjee25}, we define correlation functions for the Euclidean theory by first introducing an ultraviolet (spectral) cutoff and a regularization for the zero mode. Additionally, we work on the finite torus $M_T = [-T,T]\times [-1,1]$, with the intention of later sending $T\to \infty$. Let $\Delta$ be the Laplacian on $M_T$ and $G_T$ be the inverse of $-\frac{1}{2\pi}\Delta$ on smooth functions that integrate to zero. For details about these operators, see~\secref{sec:green}. 

Let $\{e_k\}_{k\ge 0}$ be any $L^2_{\R}(M_T)$-orthonormal basis of (smooth) real eigenfunctions of
$-\Delta$ with eigenvalues $\{\lambda_k\}_{k\ge 0}$, where $\lambda_0=0$ corresponds to the constant mode.
Then $G_T$ admits the (formal) eigenfunction expansion
\[
G_T(u,v)=2\pi \sum_{k=1}^{\infty} \frac{e_k(u)e_k(v)}{\lambda_k},
\qquad u,v\in M_T,
\]
with the constant mode omitted. We implement a frequency cutoff to define the regularized Green's function
\begin{align}\label{eq:gtnrep}
  G_{T,N}(u,v) := 2\pi  \sum_{k=1}^{L_N} \frac{e_k(u)e_k(v)}{\lambda_k},
\end{align}
where $L_N\to\infty$ as $N\to\infty$. A specific choice of basis and cutoff will be made later in this section.
Given $N$ and a function $\phi$ such that $\inn{\phi,e_k}=0$ for all $k>L_N$, define
\begin{align}
&\normord{e^{2b\phi(u)}}  \;:=\exp\biggl(2b\phi(u)+2b^2G_{T,N}(u,u)\biggr),\label{eq:normnedef}\\ 
&\normord{F(\phi)}\; 
:=\exp \biggl(2\sum_{j=1}^k \alpha_j \phi(q_j) +2\sum_{j=1}^k\alpha_j^2 G_{T,N}(q_j, q_j)\biggr),\label{eq:vndef}
\end{align}
Here $q_1,\ldots,q_k$ are distinct points in $M_T$ and $\alpha_1,\ldots,\alpha_k$ are complex numbers with real parts strictly bigger than $-\frac{1}{2b}$. These will remain fixed henceforth. We define the frequently occurring quantity
\[
  w := -\frac{\sum_{j=1}^k \alpha_j}{b},
\] 
which is assumed to be an integer. 
The regularized Euclidean correlation function is defined as 
\begin{equation}\label{eq:cutoff-correlator-with-line}
\langle F\rangle_{T,N,\epsilon}
:= \int \normord{F(\phi)}e^{-S_{N,\epsilon}(\phi)}\,\md_N\phi,
\end{equation}
where $\md_N\phi$ denotes Lebesgue measure on the finite-dimensional subspace
$\mathrm{span}\{e_0,\dots,e_{L_N}\}$, and
\[
S_{N,\epsilon}(\phi)
:=\frac{1}{4\pi}\int_{M_T}\biggl\{-|\nabla \phi(u)|^2-\frac{2\pi  \ep c(\phi)^2}{4T} +4\pi\mu\normord{e^{2b\phi(u)}} \biggr\}\,
\mathrm{d}u,
\]
where 
\[
c(\phi):= \frac{1}{4T}\int_{M_T} \phi(u) \, \dd u
\] 
is the zero mode of $\phi$. Note that this is a Gaussian integral where the quadratic term has the \emph{wrong sign}.

In \cite[Section~3]{chatterjee25}, the wrong-sign field is realized by expanding in an eigenbasis and taking wrong-sign Gaussian coefficients. In our setting, write the regularized centered wrong-sign field as
\[
X_{N}(u)=\sum_{k=1}^{L_N} \xi_k e_k(u),
\]
where $(\xi_k)_{1\le k\le L_N}$ are centered independent wrong-sign Gaussian random variables with variances
$\E[\xi_k^2]=-2\pi \lambda_k^{-1}$ for $1\le k\le L_N$, so that the covariance is 
\[
  \E(X_N(u)X_N(v)) = - G_{T,N}(u,v).
\] 
Additionally, we let $\xi_0$ be a wrong-sign centered Gaussian random variable with $\E[\xi_0^2] = -\ep^{-1}$.
In terms of this field, one may rewrite the expression in equation~\eqref{eq:cutoff-correlator-with-line} as the expectation of a
functional of $X_{N}$ and $\xi_0$:
\begin{equation}\label{eq:cutoff-correlator-wrong-sign-X}
\langle F\rangle_{T,N,\epsilon}
=\E\biggl[ \normord{F(\xi_0+X_{N})}\exp\biggl(-\mu e^{2b\xi_0}\int_{M_T} \normord{e^{2bX_{N}(u)}} \,\mathrm{d}u\biggr)\biggr].
\end{equation}
By the theory developed in \cite[Section 2.4]{chatterjee25}, one may replace each wrong-sign coefficient $\xi_k$ by $\mathrm{i}Z_k$, where $(Z_k)_{0\le k\le L_N}$ are
independent standard (right sign) Gaussians with $\E[Z_k^2]=-\E[\xi_k^2]$, and thereby
rewrite the quantity in equation~\eqref{eq:cutoff-correlator-wrong-sign-X} as 
\begin{equation}\label{eq:cutoff-correlator-regular-expectation}
\langle F\rangle_{T,N,\epsilon}
=\E\biggl[\normord{F(\i Z_0+ \mathrm{i}\varphi_{N})}
\exp\biggl(-\mu e^{2\i bZ_0}\int_{M_T} \normord{e^{2\i b\varphi_{N}(u)}}\,\mathrm{d}u\biggr)\biggr],
\end{equation}
where 
\[
  \varphi_{N}(u):=\sum_{k=1}^{L_N} Z_ke_k(u)
\] 
is the usual (right-sign) cutoff centered Gaussian free field on $M_T$. The following lemma evaluates this quantity.
\begin{lmm}\label{lem:denominator-evaluation}
We have 
\begin{align*}
\langle F\rangle_{T,N,\epsilon}
&=\sum_{n=0}^\infty \frac{(-\mu)^n}{n!}\exp\biggl(-\frac{2b^2 (n-w)^2}{\ep} - 4\sum_{1\le j<j'\le k}\alpha_j \alpha_{j'}G_{T,N}(q_j, q_{j'})\biggr) \\
&\quad  \cdot \int_{M_T^n} \exp\biggl(- 4b \sum_{j=1}^k \sum_{l=1}^n\alpha_j G_{T,N}(q_j,u_l) - 4b^2 \sum_{1\le l<l'\le n} G_{T,N}(u_l, u_{l'}) \biggr)\, \mathrm{d} u_1\cdots \mathrm{d} u_n.
\end{align*}
\end{lmm}

\begin{proof}
By expanding in a power series and applying Fubini's theorem (justified here since $G_{T,N}$ is bounded, which implies that the absolute value of the term inside the expectation below is growing at most exponentially in $n$),
\[
\langle F\rangle_{T,N,\epsilon}
=\sum_{n=0}^\infty \frac{(-\mu)^n}{n!}\int_{M_T^n}
\E\biggl[\normord{F(\i Z_0+ \mathrm{i}\varphi_N)} e^{2\i n b Z_0}\prod_{l=1}^n \normord{e^{2\i b\varphi_N(u_l)}}\biggr]\,
\mathrm{d}u_1\cdots \mathrm{d}u_n.
\]
Recalling the definitions \eqref{eq:normnedef} and \eqref{eq:vndef}, we get
\begin{align*}
&\normord{F(\i Z_0+ \mathrm{i}\varphi_N)} e^{2\i n b Z_0}\prod_{l=1}^n \normord{e^{2b\mathrm{i}\varphi_N(u_l)}} \\
&\qquad =\exp\biggl(2\i \biggl(\sum_{j=1}^k \alpha_j + nb\biggr) Z_0 + 2\i \sum_{j=1}^k \alpha_j \varphi_N(q_j) +2\sum_{j=1}^k\alpha_j^2 G_{T,N}(q_j, q_j)
\\ 
&\qquad \qquad+2\i b\sum_{l=1}^n \varphi_N(u_l) +2b^2\sum_{l=1}^n G_{T,N}(u_l,u_l)\biggr).
\end{align*}
Since $\varphi_N$ is a centered Gaussian field with covariance $G_{T,N}$ and $Z_0$ is an independent Gaussian random variable with mean zero and variance $\ep^{-1}$, the expected value of the above quantity is 
\begin{align*}
&\exp\biggl(-\frac{2b^2 (n-w)^2}{\ep} - 4\sum_{1\le j<j'\le k}\alpha_j \alpha_{j'}G_{T,N}(q_j, q_{j'})\\ 
&\qquad \qquad - 4b \sum_{j=1}^k \sum_{l=1}^n\alpha_j G_{T,N}(q_j,u_l) - 4b^2 \sum_{1\le l<l'\le n} G_{T,N}(u_l, u_{l'}) \biggr).
\end{align*}
Plugging this into the first display of the proof yields the desired formula.
\end{proof}
This identifies the regularized finite-volume correlator. We next remove the zero-mode and ultraviolet regularizations.

\subsection{Removing the regularization}\label{sec:zerodenom}
We will now remove the regularization from the finite-volume theory by first sending $\ep$ to zero and then sending $N$ to infinity.
We first record the negative screening consequence of the zero-mode selection rule.
\begin{proof}[Proof of Theorem \ref{thm:euc} in the case $w<0$]
Suppose that $w$ is a negative integer. In the expansion from
Lemma~\ref{lem:denominator-evaluation}, the zero-mode factor in the $n$th term
is
\[
  \exp\bigl(-2b^2(n-w)^2/\epsilon\bigr),
  \qquad n\in\Z_{\ge0}.
\]
Since $n-w\ne0$ for every $n\ge0$, each term tends to zero as
$\epsilon\downarrow0$. The estimate used in the proof of
Lemma~\ref{lem:denominator-evaluation-integer} dominates the whole series
uniformly for $0<\epsilon\le1$: after taking real parts, the insertion terms
are $O(n)$ for fixed $T,N$, while the positive-semidefinite estimate for
$G_{T,N}$ controls the screening-screening pair sum by $O(n)$. Thus the $n$th
integral is bounded by $C^n$ on the finite torus, and the factor $1/n!$ makes
the resulting majorant summable. Dominated convergence therefore
gives
\[
  \lim_{\epsilon\downarrow0}\langle F\rangle_{T,N,\epsilon}=0
\]
for every fixed $T$ and $N$. The subsequent limits $N\to\infty$ and
$T\to\infty$ are therefore also zero. This proves the
claim.
\end{proof}

Throughout the rest of this subsection, we work under the assumptions that $\Re(\alpha_j)>-\frac{1}{2b}$ for each $j$ and $w$ is a nonnegative integer.

\begin{lmm}\label{lem:denominator-evaluation-integer}
We have 
\begin{align*}
\inn{F}_{T,N}&:= \lim_{\ep\to 0}\inn{F}_{T,N,\ep} \\ 
&=\frac{(-\mu)^w}{w!}\exp\biggl( - 4\sum_{1\le j<j'\le k}\alpha_j \alpha_{j'}G_{T,N}(q_j, q_{j'})\biggr) \\
&\quad  \cdot \int_{M_T^w} \exp\biggl(- 4b \sum_{j=1}^k \sum_{l=1}^w\alpha_j G_{T,N}(q_j,u_l) - 4b^2 \sum_{1\le l<l'\le w} G_{T,N}(u_l, u_{l'}) \biggr)\, \mathrm{d} u_1\cdots \mathrm{d} u_w,
\end{align*}
where the integral is interpreted as $1$ if $w=0$.
\end{lmm}
The proof of this lemma is deferred to \secref{sec:aux-zero-mode-limit}. Having sent $\ep\to 0$ in Lemma \ref{lem:denominator-evaluation-integer}, we will now send $N\to\infty$ to complete the process of removing all regularizations in equation \eqref{eq:cutoff-correlator-regular-expectation}. But before this can be done, we need to make a specific choice for the basis $\{e_k\}_{k\ge 0}$. The final answer will not depend on this choice, but a basis needs to be fixed so that we can send $N\to\infty$ and obtain the correct limit.

We choose the standard Fourier basis on the flat torus $M_T$ with periodic boundary conditions. Let $\Z^2_*$ be obtained from $\Z^2$ by choosing exactly one element from each pair $(m,n), (-m,-n)\in \Z^2$. Then a complete orthonormal basis of $L_{\R}^2(M_T)$ consisting of eigenfunctions of $-\Delta$ is given by 
\begin{align*}
e_{m,n}^{\mathrm{c}}(t,x)&:=\frac{1}{\sqrt{2T}}\cos\biggl(\pi\frac{m}{T}t+\pi n x\biggr),\\
e_{m,n}^{\mathrm{s}}(t,x)&:=\frac{1}{\sqrt{2T}}\sin\biggl(\pi\frac{m}{T}t+\pi n x\biggr), 
\qquad (m,n)\in\Z_*^2,
\end{align*}
with eigenvalues
\[
\lambda_{m,n}=\pi^2\biggl(\frac{m^2}{T^2}+n^2\biggr).
\]
With this choice of basis, we define the cutoff Green's function to be 
\begin{align}\label{eq:gtnalt}
&G_{T,N}((t,x), (s,y)) \notag \\ 
&:= 2\pi \sum_{\substack{(m,n)\in \Z^2_*,\ |m|\le N,\ |n|\le N\\ (m,n)\neq(0,0)}}\frac{1}{\lambda_{m,n}}(e_{m,n}^{\mathrm{c}}(t,x) e_{m,n}^{\mathrm{c}}(s,y)+ e_{m,n}^{\mathrm{s}}(t,x) e_{m,n}^{\mathrm{s}}(s,y)).
\end{align}
Lemma \ref{lem:GN-form} of \secref{sec:green} identifies this $G_{T,N}$ with the Fourier-series cutoff used in the Appendix, Lemma \ref{lem:gN-to-g} shows pointwise convergence to a limit function $G_T$ away from the diagonal in a suitable sense as $N\to\infty$, and Proposition \ref{prop:green} shows that $G_T$ acts as a well-defined integral kernel on smooth functions and is the inverse of $-\frac{1}{2\pi}\Delta$ on mean-zero smooth functions (as expected). 
The following result gives the limit of $\inn{F}_{T,N}$ as $N\to\infty$ as a formula involving the kernel $G_T$. 
\begin{thm}\label{thm:denominator-limit-N}
We have 
\begin{align*}
\inn{F}_{T}&:= \lim_{N\to \infty}\inn{F}_{T,N} \\ 
&=\frac{(-\mu)^w}{w!}\exp\biggl( - 4\sum_{1\le j<j'\le k}\alpha_j \alpha_{j'}G_{T}(q_j, q_{j'})\biggr) \\
&\quad  \cdot \int_{M_T^w} \exp\biggl(- 4b \sum_{j=1}^k \sum_{l=1}^w\alpha_j G_{T}(q_j,u_l) - 4b^2 \sum_{1\le l<l'\le w} G_{T}(u_l, u_{l'}) \biggr)\, \mathrm{d} u_1\cdots \mathrm{d} u_w.
\end{align*}
\end{thm}
\begin{proof}
If $w=0$, the claim follows immediately from Lemma~\ref{lem:denominator-evaluation-integer}, Lemma~\ref{lem:gN-to-g}, and Lemma~\ref{lem:GN-form}.
So assume $w\ge 1$. Fix distinct points $u_1,\dots,u_w\in M_T$. Throughout this proof, $C,C_1,\ldots$ will denote positive constants that do not depend on $u_1,\ldots,u_w$. By
Lemma~\ref{lem:gN-to-g} and Lemma~\ref{lem:GN-form}, 
$G_{T,N}(u_l,u_{l'})\allowbreak\to G_T(u_l,u_{l'})$ for each $l\neq l'$. To justify dominated convergence, first note that by Lemma~\ref{lem:gN-lower},
\begin{align*}
&\Re\biggl\{- 4b \sum_{j=1}^k \sum_{l=1}^w\alpha_j G_{T,N}(q_j,u_l) - 4b^2 \sum_{1\le l<l'\le w} G_{T,N}(u_l, u_{l'})\biggr\} \\ 
&\le C_1 -4b \sum_{j=1}^k\sum_{l=1}^w\Re(\alpha_j)G_{T,N}(q_j,u_l).
\end{align*}
Now, if $\Re(\alpha_j) \ge 0$ for some $j$, then Lemma~\ref{lem:gN-lower} gives 
\[
  -4b\Re(\alpha_j) G_{T,N}(q_j, u_l)\le C_2.
\]
On the other hand, if $\Re(\alpha_j) < 0$, then Lemma \ref{lem:gN-log-upper} gives 
\[
  -4b\Re(\alpha_j) G_{T,N}(q_j, u_l)\le C_3 -4b\Re(\alpha_j) |\log d(q_j, u_l)|,
\]
where $d(q_j, u_l)$ denotes the geodesic distance between $q_j$ and $u_l$ on $M_T$. Since $\Re(\alpha_j)> -\frac{1}{2b}$ for each $j$, we conclude that there is some $\delta>0$ such that for any $u_1,\ldots,u_w$,
\begin{align*}
&\biggl|\exp\biggl(- 4b \sum_{j=1}^k \sum_{l=1}^w\alpha_j G_{T,N}(q_j,u_l) - 4b^2 \sum_{1\le l<l'\le w} G_{T,N}(u_l, u_{l'}) \biggr)\biggr| \\ 
&= \exp\biggl(\Re\biggl\{- 4b \sum_{j=1}^k \sum_{l=1}^w\alpha_j G_{T,N}(q_j,u_l) - 4b^2 \sum_{1\le l<l'\le w} G_{T,N}(u_l, u_{l'})\biggr\}\biggr) \\ 
&\le \exp\biggl(C_1 +(2-\delta) \sum_{j=1}^k\sum_{l=1}^w|\log d(q_j,u_l)|\biggr).
\end{align*}
Since $M_T$ is a two-dimensional manifold and $u_1,\ldots,u_w$ are distinct, it is straightforward to see that the above upper bound has a finite integral over $M_T^w$. This allows us to apply the dominated convergence theorem and complete the proof.
\end{proof}
At this point the finite-volume theory has been constructed without regularization. The remaining task is to send the torus size to infinity.

\subsection{Euclidean theory in infinite volume}
We will now construct the Euclidean theory in infinite volume by sending  $T\to\infty$. A key step is to identify the limiting behavior of $g_T$ as $T\to \infty$. In Lemmas \ref{lem:gtlimit} and \ref{lem:gtlimit2} from~\secref{sec:green}, we show that for any distinct $(t,x), (s,y)\in M$, 
\begin{align*}
\lim_{T\to \infty} \biggl(G_T((t,x), (s,y))-\frac{\pi T}{6}\biggr) = -\frac{\pi |t-s|}{2}-\log|1-e^{-\pi|t-s|+\pi \i (x-y)}|. 
\end{align*}
Motivated by this, we define the renormalized Green's function in finite volume as 
\[
  \widetilde G_T(u,v):=G_T(u,v)-\frac{\pi}{6}T,
\]
and its infinite-volume limit as 
\[
  \widetilde G((t,x),(s,y))
  :=
  -\frac{\pi |t-s|}{2}
  -\log|1-e^{-\pi |t-s|+\pi \i (x-y)}|.
\]
We also define certain functions for notational convenience. For $T>0$, define $\Phi_T:M^w \to \C$ as $\Phi_T(u_1,\ldots,u_w)=0$ if $(u_1,\ldots,u_w)\notin M_T^w$, and if $(u_1,\ldots,u_w)\in M_T^w$, we let
\[
  \Phi_T(u_1,\dots,u_w)
  :=
  \exp\biggl(
    -4b\sum_{j=1}^k\sum_{l=1}^w \alpha_j \widetilde G_T(q_j,u_l)
    -4b^2\sum_{1\le l<l'\le w}\widetilde G_T(u_l,u_{l'})
  \biggr).
\]
We define the $T=\infty$ version of $\Phi_T$, denoted $\Phi$, by replacing $\widetilde G_T$ by $\widetilde G$ everywhere in the above formula.  
The following theorem gives the limit of $\inn{F}_T$ as $T\to\infty$, after renormalizing by a suitable $T$-dependent factor. This is actually Theorem \ref{thm:euc} presented slightly differently; thus, the proof of this theorem, given below, completes the proof of Theorem \ref{thm:euc}.
\begin{thm}\label{thm:denominator-limit-T}
We have
\begin{align*}
&\lim_{T\to\infty}
\exp\biggl\{ 
\frac{\pi T}{3}\sum_{j=1}^k \alpha_j(b-\alpha_j)
\biggr\}\,\inn{F}_T \\
&\qquad=
\frac{(-\mu)^w}{w!}
\exp\biggl(
-4\sum_{1\le j<j'\le k}\alpha_j\alpha_{j'}\widetilde G(q_j,q_{j'})
\biggr) \int_{M^w}
\Phi(u_1,\ldots,u_w)\,\dd u_1\cdots \dd u_w,
\end{align*}
and the integral on the right is absolutely convergent.
\end{thm}

\begin{proof}
Recall the formula for $\inn{F}_T$ from Theorem~\ref{thm:denominator-limit-N}. Since $G_T=\widetilde G_T+\frac{\pi}{6}T$, we may factor out the constant contribution:
\begin{align*}
\inn{F}_T
&=
\frac{(-\mu)^w}{w!}
\exp\biggl(
-4\sum_{1\le j<j'\le k}\alpha_j\alpha_{j'} \widetilde G_T(q_j,q_{j'})
\biggr) \\
&\quad \cdot
\exp\biggl\{ 
-\frac{2\pi T}{3}
\biggl(
\sum_{1\le j<j'\le k}\alpha_j\alpha_{j'}
+bw\sum_{j=1}^k\alpha_j
+b^2\binom{w}{2}
\biggr)
\biggr\} \\ 
&\qquad \qquad \cdot \int_{M^w}\Phi_T(u_1,\dots,u_w)\,\dd u_1\cdots \dd u_w .
\end{align*}
Using $\sum_{j=1}^k \alpha_j=-bw$, we obtain
\begin{align*}
-\frac{2\pi T}{3}
\biggl(
\sum_{1\le j<j'\le k}\alpha_j\alpha_{j'}
+bw\sum_{j=1}^k\alpha_j
+b^2\binom{w}{2}
\biggr)
&=
\frac{\pi T}{3}\biggl(\sum_{j=1}^k\alpha_j^2+b^2w\biggr)\\ 
&= -\frac{\pi T}{3}\sum_{j=1}^k \alpha_j(b - \alpha_j).
\end{align*}
Therefore,
\begin{align*}
&\exp\biggl\{
\frac{\pi T}{3}\sum_{j=1}^k\alpha_j(b -\alpha_j)
\biggr\}\,\inn{F}_T \notag\\
&\qquad=
\frac{(-\mu)^w}{w!}
\exp\biggl(
-4\sum_{1\le j<j'\le k}\alpha_j\alpha_{j'} \widetilde G_T(q_j,q_{j'})
\biggr)  \int_{M^w}\Phi_T(u_1,\dots,u_w)\,\dd u_1\cdots \dd u_w.
\end{align*}
We already know that $\widetilde G_T(q_j,q_{j'}) \to \widetilde G(q_j,q_{j'})$ as $T\to \infty$. So the proof will be complete if we can show that 
\begin{align}\label{eq:phittoshow}
\lim_{T\to \infty}\int_{M^w}\Phi_T(u_1,\dots,u_w)\,\dd u_1\cdots \dd u_w = \int_{M^w}\Phi(u_1,\dots,u_w)\,\dd u_1\cdots \dd u_w.
\end{align}
Again, we know that the integrand on the left converges pointwise to the integrand on the right as $T\to \infty$. Thus, we only need to show that the limit on the left can be moved inside the integral. 
To this end, fix $R>0$, and let
\[
K_R:=([-R,R]\times[-1,1])^w.
\]
On $K_R$, the integrand is dominated by the integrable bound from
Lemma~\ref{lem:renorm-int-tail} of \secref{sec:bnkbounds}. So, 
\[
\lim_{T\to\infty}\int_{K_R}\Phi_T(u_1,\ldots,u_w)\,\dd u_1\cdots \dd u_w
=
\int_{K_R}\Phi(u_1,\ldots,u_w)\,\dd u_1\cdots \dd u_w.
\]
Next, again by Lemma~\ref{lem:renorm-int-tail}, there are positive constants  $T_0,C,\eta, \delta$, independent of $R$, such that 
\begin{align*}
&\sup_{T\ge T_0}\int_{M^w\setminus K_R} |\Phi_T(u_1,\ldots,u_w)|\,\dd u_1\cdots \dd u_w \\ 
&\le
\int_{M^w\setminus K_R}
C \exp\biggl(-\eta\sum_{l=1}^w \|u_l\|\biggr)
\prod_{j=1}^k\prod_{l=1}^w d(q_j,u_l)^{-(2-\delta)}
\,\dd u_1\cdots \dd u_w =: \ep(R),
\end{align*}
where $d$ denotes geodesic distance on $M$. 
From this, Fatou's lemma shows that 
\[
  \int_{M^w\setminus K_R} |\Phi(u_1,\ldots,u_w)|\,\dd u_1\cdots \dd u_w\le \ep(R).
\]
Combining the last three displays, we get 
\begin{align*}
\limsup_{T\to\infty} \biggl|\int_{M^w}\Phi_T(u_1,\ldots,u_w)\, \dd u_1 \cdots \dd u_w - \int_{M^w} \Phi(u_1,\ldots,u_w)\, \dd u_1\cdots \dd u_w\biggr|\le \ep(R).
\end{align*}
But the left side does not depend on $R$, and the dominated convergence theorem shows that $\ep(R)\to 0$ as $R\to \infty$. This proves equation \eqref{eq:phittoshow}, thereby completing the proof of the main claim of the theorem. For the finiteness of the integral appearing in the theorem statement, we simply invoke the bound from Lemma~\ref{lem:renorm-int-tail} and apply Fatou's lemma.
\end{proof}
This completes the Euclidean part of the construction and, with it, the proof of Theorem~\ref{thm:euc}. We now turn to analytic continuation and the passage to Lorentzian signature.

\section{Analytic continuation}\label{sec:anacont}
In this section, our main goal is to prove Theorem \ref{thm:main}. The case $w< 0$ is immediate, so we will assume throughout that $w$ is a nonnegative integer. We carry out the proof by first establishing a stronger analytic-continuation statement on the ordered tube and then identifying its Lorentzian boundary values. The task is divided into subsections. In the first subsection below, we inspect the analytic continuation of the function $g$ to the open right half-plane and establish some of its properties.
\subsection{Analytic continuation of the Green's function}\label{sec:anagreen}
Recall the function $g:\R^2\to\R$,
\[
  g(t,x) = -\frac{\pi |t|}{2}-\log|1-e^{-\pi |t|+\pi \i x}|,
\]
that gave us the mean-zero Green's function $G(u,v)=g(u-v)$ on the cylinder $M$. We now work out the analytic continuation of $g$ in the time coordinate to the open right half-plane. Observe that for $t>0$, 
\begin{align*}
g(t,x) &= -\frac{\pi t}{2}-\frac{1}{2}\log|1-e^{-\pi |t|+\pi \i x}|^2\\ 
&= -\frac{\pi t}{2}-\frac{1}{2}\log(1-e^{-\pi t+\pi \i x}) - \frac{1}{2}\log(1-e^{-\pi t-\pi \i x}),
\end{align*}
where the logarithms in the second line denote the analytic branch of the logarithm on $\C \setminus(-\infty,0]$ that is real-valued on the real line. Now suppose in the above formula that we replace $t$ by a complex variable $\tau = a+\i b$, with $a>0$. Note that 
\begin{align*}
\Re(1-e^{-\pi \tau+\pi \i x}) &= \Re(1-e^{-\pi a +\pi \i (x-b)})\\ 
&= 1- e^{-\pi a}\cos(\pi(x-b)) > 0,
\end{align*}
since $|e^{-\pi a}\cos(\pi(x-b))| < 1$. This shows that $g(t,x)$ has a straightforward analytic continuation in the time coordinate to the open right half-plane, given by 
\begin{align}\label{eq:greencont}
g(\tau,x) &= -\frac{\pi\tau}{2}-\frac{1}{2}\log(1-e^{-\pi \tau+\pi \i x}) - \frac{1}{2}\log(1-e^{-\pi \tau-\pi \i x}).
\end{align}
The function now extends by continuity to $(\i t , x)$ where $t,x\in \R$ and $t \pm x \notin 2\Z$. We will also use the following reflected boundary-value convention for arguments with negative real part: for $a<0$, set $g(a+\i b,x):=g(-a+\i b,x)$, away from the same singular points. This is a time-ordering convention, not an analytic continuation across the imaginary axis. The following estimate for $g$ will be useful.
\begin{lmm}\label{lem:gcontlmm}
There are positive universal constants $C_1,C_2,C_3$ such that for any $\tau$ in the open right half-plane and any $x\in\R$, if we define
\[
  d_\pm(\tau,x):=\min_{k\in\Z}|\tau\pm \i x-2\i k|,
  \qquad
  d(\tau,x):=\min\{d_+(\tau,x),d_-(\tau,x)\},
\]
then
\begin{align*}
  -C_1 \le \Re\biggl(g(\tau,x)+\frac{\pi\tau}{2}\biggr) \le C_2 + |\log d(\tau,x)|,
  \qquad
  \biggl|\Im\biggl(g(\tau,x)+\frac{\pi\tau}{2}\biggr)\biggr|\le C_3.
\end{align*}
\end{lmm}
\begin{proof}
Throughout this proof, $C_0, C_1,\ldots$ will denote positive universal constants whose values may change from line to line. 
Write $\tau=a+\i b$ with $a>0$, and define
\[
  z_\pm:=1-e^{-\pi\tau\pm \pi \i x}.
\]
From equation~\eqref{eq:greencont},
\[
  g(\tau,x)+\frac{\pi\tau}{2}= -\frac12\log z_+ -\frac12\log z_-.
\]
For either sign,
\[
\Re(z_\pm)=1-e^{-\pi a}\cos(\pi(x\mp b))\ge 1-e^{-\pi a}>0.
\]
So $z_\pm$ lies in the open right half-plane. Hence, for our branch of logarithm,
$|\arg(z_\pm)|<\frac\pi2$. Therefore
\[
\biggl|\Im\biggl(g(\tau,x)+\frac{\pi\tau}{2}\biggr)\biggr|
=\frac12|\arg z_++\arg z_-|
\le \frac12(|\arg z_+|+|\arg z_-|)
\le \frac\pi2.
\]
This proves the bound on the imaginary part claimed in the lemma. For the lower bound on the real part, note that since $|z_\pm|\le 1+|e^{-\pi\tau\pm \pi \i x}|\le 2$, we get
\[
\Re\biggl(g(\tau,x)+\frac{\pi\tau}{2}\biggr)
=-\frac12\log|z_+|-\frac12\log|z_-|
\ge -\log 2.
\]
Now we prove the upper bound on the real part. Let
\[
  w_\pm:=\pi(\tau\mp \i x)=\pi a+\i\pi(b\mp x).
\]
Choose $k_\pm\in\Z$ so that
\[
  |\Im(w_\pm)-2\pi k_\pm|=\dist(\Im(w_\pm),2\pi\Z),
\]
and define
\[
  \widetilde b_\pm:=\Im(w_\pm)-2\pi k_\pm\in[-\pi,\pi].
\]
Then
\[
  |z_\pm|
  =|1-e^{-w_\pm}|
  =|1-e^{-\pi a+\i\widetilde b_\pm}|.
\]
Applying Lemma~\ref{lem:one-minus-exp-lower} with $u=\pi a$ and $v=\widetilde b_\pm$, we get
\[
  |z_\pm|\ge C_0\min\biggl\{1,\sqrt{(\pi a)^2+\widetilde b_\pm^2}\biggr\}.
\]
But
\[
  \sqrt{(\pi a)^2+\widetilde b_\pm^2}=|w_\pm-2\pi\i k_\pm|=|\pi(\tau\mp\i x-2\i k_\pm)|\ge \pi d(\tau,x).
\]
Thus, we conclude that 
\[
  |z_\pm|\ge C_1\min\{1,d(\tau,x)\}.
\]
Hence,
\begin{align*}
\Re\biggl(g(\tau,x)+\frac{\pi\tau}{2}\biggr)
&=-\frac12\log|z_+|-\frac12\log|z_-|\\
&\le -\log(\min\{|z_+|,|z_-|\})\\
&\le -\log(C_1\min\{1,d(\tau,x)\})\\
&\le C_2+|\log d(\tau,x)|.
\end{align*}
This completes the proof.
\end{proof}
With the analytically continued Green's function under control, we can now reorganize the full correlation function into pieces suited to continuation.

\subsection{Decomposing the Euclidean correlation}
Recall that the Euclidean correlation is given by
\begin{align*}
C_{\alpha_1,\ldots,\alpha_k}(q_1;\ldots; q_k)&= \frac{(-\mu)^w}{w!} A_{\alpha_1,\ldots, \alpha_k}(q_1;\ldots;q_k)B_{\alpha_1,\ldots,\alpha_k}(q_1;\ldots;q_k)
\end{align*}
where 
\begin{align*}
A_{\alpha_1,\ldots, \alpha_k}(q_1;\ldots;q_k) &:= \exp\biggl(
-4\sum_{1\le j<j'\le k}\alpha_j\alpha_{j'} g(q_j-q_{j'})
\biggr),
\end{align*}
and 
\begin{align*}
B_{\alpha_1,\ldots,\alpha_k}(q_1;\ldots;q_k)&:= \int_{M^w} \exp\biggl(
    -4b\sum_{j=1}^k\sum_{l=1}^w \alpha_j g(q_j-u_l)
    \\ 
    &\qquad \qquad -4b^2\sum_{1\le l<l'\le w}g(u_l-u_{l'})
  \biggr)\,\dd u_1\cdots \dd u_w.
\end{align*}
Let us now denote $q_j = (t_j, x_j)$. Fixing $\alpha_1,\ldots,\alpha_k$ and $x_1,\ldots,x_k$, let us denote the above functions simply by $A(t_1,\ldots,t_k)$ and $B(t_1,\ldots,t_k)$. We will refer to these as the ``free part'' and the ``Coulomb part'' of the correlation. Our goal is to analytically continue these functions to the domain 
\[
  \Omega_k := \{(t_1,\ldots,t_k)\in \C^k: \Re(t_1)<\cdots < \Re(t_k) \},
\]
and then extend continuously to obtain the values at $(\i t_1,\ldots,\i t_k)$ for some given real numbers $t_1,\ldots,t_k$. 

\subsection{Continuing the free part}\label{sec:free}
The free part of the correlation is easier to continue analytically. Let $g$ denote the analytic continuation of the Green's function $g$, as displayed in equation \eqref{eq:greencont}. For $(t_1,\ldots,t_k)\in \Omega_k$, define 
\begin{align*}
A(t_1,\ldots,t_k) &:= \exp\biggl(
-4\sum_{1\le j<j'\le k}\alpha_j\alpha_{j'} g(t_{j'}-t_{j}, x_{j'} -x_{j})\biggr).
\end{align*}
Note that $\Re(t_{j'}-t_{j}) > 0$ for every $j' > j$. Since $g(t,x)$, as a function of $t$, is analytic on the open right half-plane, the above function is analytic on $\Omega_k$. Moreover, it coincides with $A$ for real $t_1,\ldots,t_k$, since $g(t,x) = g(-t,-x)$ for real $t,x$. This gives the required analytic continuation of the free part of the correlation.
Thus the free part poses no further difficulty; the main work lies in continuing the Coulomb integral.

\subsection{Continuing the Coulomb part}\label{sec:coulomb}
Consider $t_1<\cdots<t_k\in \R$. 
Let us first write the Coulomb part by expressing the variables $u_1,\ldots,u_w$ explicitly in coordinates:
\begin{align*}
B(t_1,\ldots,t_k) &= \int_{\R^w}\int_{[-1,1]^w} \exp\biggl(
    -4b\sum_{j=1}^k\sum_{l=1}^w \alpha_j g(t_j-a_l, x_j - b_l)
    \\ 
    &\qquad \qquad -4b^2\sum_{1\le l<l'\le w}g(a_l-a_{l'}, b_l - b_{l'})\biggr) \, \dd b_1\cdots \dd b_w\, \dd a_1 \cdots \dd a_w.
\end{align*}
First, using symmetry between $a_1,\ldots,a_w$ in the integral, and the fact that $g(t,x)=g(-t,-x)$, we can rewrite this as 
\begin{align*}
B(t_1,\ldots,t_k) &= w!\int_{a_1<\cdots< a_w}\int_{[-1,1]^w} \exp\biggl(
    -4b\sum_{j=1}^k\sum_{l=1}^w \alpha_j g(t_j-a_l, x_j - b_l)
    \\ 
    &\qquad \qquad -4b^2\sum_{1\le l<l'\le w}g(a_{l'}-a_{l}, b_{l'} - b_{l})\biggr) \, \dd b_1\cdots \dd b_w\, \dd a_1 \cdots \dd a_w.
\end{align*}
Next, let us write the above as a sum of functions 
\[
  B(t_1,\ldots,t_k) = \sum_{1\le n_1 \le \cdots \le n_k\le w+1} B_{n_1,\ldots,n_k}(t_1,\ldots,t_k),
\]
where $B_{n_1,\ldots,n_k}$ is obtained by restricting the integral to those $a_1<\cdots<a_w$ such that 
\begin{itemize}
\item $a_l < t_1$ for all $l< n_1$, 
\item $t_j < a_l < t_{j+1}$ for all $n_j \le l< n_{j+1}$, $1\le j< k$, and 
\item $a_l > t_k$ for all $l\ge n_k$. 
\end{itemize}
Let $\Gamma_{n_1,\ldots,n_k}$ denote the set of all such $(a_1,\ldots,a_w)$. (To justify this decomposition, fix any $a_1<\cdots<a_w$. The exceptional case where some $a_l=t_j$ lies in a finite union of codimension-one hyperplanes, hence has Lebesgue measure zero and does not affect the integral. So assume $a_l\neq t_j$ for all $l,j$. For each $1\le j\le k$, let $n_j$ be the smallest $l$ such that $a_l>t_j$ if such an $l$ exists, and let $n_j=w+1$ otherwise. Then $1\le n_1\le\cdots\le n_k\le w+1$, the point $(a_1,\ldots,a_w)$ belongs to $\Gamma_{n_1,\ldots,n_k}$, and this index tuple is unique.) Again using $g(t,x)=g(-t,-x)$, we can write
\begin{align*}
&B_{n_1,\ldots,n_k}(t_1,\ldots,t_k) \\ 
&= w!\int_{\Gamma_{n_1,\ldots, n_k}}\int_{[-1,1]^w} \exp\biggl\{ 
    -4b\sum_{l=1}^{n_1-1} \sum_{j=1}^k \alpha_j g(t_j -a_l, x_j-b_l) \\ 
    &\qquad \qquad -4b \sum_{r=1}^{k-1}\sum_{l=n_r}^{n_{r+1}-1}\biggl(\sum_{j=1}^r \alpha_jg(a_l - t_j, b_l -x_j) + \sum_{j=r+1}^k \alpha_j g(t_j - a_l, x_j -b_l)\biggr) \\ 
    &\qquad \qquad -4b \sum_{l=n_k}^w \sum_{j=1}^k \alpha_j g(a_l - t_j, b_l -x_j) \\ 
    &\qquad \qquad -4b^2\sum_{1\le l<l'\le w}g(a_{l'}-a_{l}, b_{l'} - b_{l})\biggr\} \, \dd b_1\cdots \dd b_w\, \dd a_1 \cdots \dd a_w.
\end{align*}
The purpose of writing things in the above manner is to ensure that the time arguments in the $g$'s all have positive real parts. The next step is to apply a change of variables to the above integral, by reparametrizing $a_l$ as 
\begin{align}\label{eq:alcl}
a_l = 
\begin{cases}
t_1 + c_l &\text{ if } l < n_1,\\ 
c_l t_{j+1} + (1-c_l)t_j &\text{ if } n_j\le l< n_{j+1},\, 1\le j< k,\\ 
t_k + c_l &\text{ if } l \ge n_k.
\end{cases}
\end{align}
Then note that 
\begin{align*}
\dd a_1 \cdots \dd a_w = \biggl(\prod_{j=1}^{k-1}(t_{j+1}-t_j)^{n_{j+1}-n_j}\biggr) \, \dd c_1\cdots \dd c_w,
\end{align*}
and the range of integration for $(c_1,\ldots,c_w)$ is the region
\begin{align*}
\Delta_{n_1,\ldots,n_k} &:= \biggl\{(c_1,\ldots,c_w)\in \R^w:c_1<\cdots<c_{n_1-1}<0,\\
&\qquad\qquad 0<c_{n_j}<\cdots<c_{n_{j+1}-1}<1\text{ for each }1\le j<k,\\
&\qquad\qquad 0<c_{n_k}<\cdots<c_w\biggr\},
\end{align*}
with the convention that constraints corresponding to empty index ranges are omitted. The crucial fact is that this range has no dependence on $t_1,\ldots,t_k$. With this reparametrization, we have 
\begin{align}\label{eq:bform}
&B_{n_1,\ldots,n_k}(t_1,\ldots,t_k) = w!\biggl(\prod_{j=1}^{k-1}(t_{j+1}-t_j)^{n_{j+1}-n_j}\biggr)\notag \\ 
&\qquad \cdot \int_{\Delta_{n_1,\ldots, n_k}}\int_{[-1,1]^w} \exp\biggl\{ 
    -4b\sum_{l=1}^{n_1-1} \sum_{j=1}^k \alpha_j g(t_j -a_l, x_j-b_l)\notag  \\ 
    &\qquad \qquad -4b \sum_{r=1}^{k-1}\sum_{l=n_r}^{n_{r+1}-1}\biggl(\sum_{j=1}^r \alpha_jg(a_l - t_j, b_l -x_j) + \sum_{j=r+1}^k \alpha_j g(t_j - a_l, x_j -b_l)\biggr)\notag  \\ 
    &\qquad \qquad -4b \sum_{l=n_k}^w \sum_{j=1}^k \alpha_j g(a_l - t_j, b_l -x_j)\notag  \\ 
    &\qquad \qquad -4b^2\sum_{1\le l<l'\le w}g(a_{l'}-a_{l}, b_{l'} - b_{l})\biggr\} \, \dd b_1\cdots \dd b_w\, \dd c_1 \cdots \dd c_w,
\end{align}
where now each $a_l$ is a function of $c_l$ given by the formula \eqref{eq:alcl}. The purpose of this reparametrization is made clear by the following lemma.
\begin{lmm}\label{lem:reparam}
If we explicitly write each $a_l$ in equation \eqref{eq:bform} as a function of $c_l$ as written in equation \eqref{eq:alcl}, and $(t_1,\ldots,t_k)\in \Omega_k$, then the time argument of each $g$ in equation \eqref{eq:bform} has positive real part.
\end{lmm}
\begin{proof}
Fix $t=(t_1,\ldots,t_k)\in\Omega_k$ and $(c_1,\ldots,c_w)\in\Delta_{n_1,\ldots,n_k}$. We prove, case by case, that every time argument in \eqref{eq:bform} has positive real part. Since $t\in\Omega_k$, we have
\[
\Re(t_{r+1}-t_r)>0,\qquad 1\le r<k.
\]
Hence, for any $p\ge q$,
\[
\Re(t_p-t_q)\ge 0,
\]
with strict inequality when $p>q$.
Also, by definition of $\Delta_{n_1,\ldots,n_k}$,
\[
c_l<0\ \text{for }l<n_1,\qquad 0<c_l<1\ \text{for }n_1\le l<n_k,\qquad c_l>0\ \text{for }l\ge n_k.
\]
For the first sum in \eqref{eq:bform}, note that if $l<n_1$, then $a_l=t_1+c_l$; so for any $1\le j\le k$,
\[
t_j-a_l=(t_j-t_1)-c_l.
\]
Here $\Re(t_j-t_1)\ge 0$ and $-c_l>0$, and hence
\[
\Re(t_j-a_l)>0.
\]
For the middle sum in equation~\eqref{eq:bform}, fix $1\le r<k$ and $n_r\le l<n_{r+1}$, so that $a_l=c_l t_{r+1}+(1-c_l)t_r$. If $1\le j\le r$, then
\[
a_l-t_j=c_l(t_{r+1}-t_r)+(t_r-t_j).
\]
Since $c_l>0$, we have $\Re(c_l(t_{r+1}-t_r))>0$, and $\Re(t_r-t_j)\ge 0$. Therefore,
\[
\Re(a_l-t_j)>0.
\]
If $r+1\le j\le k$, then
\[
t_j-a_l=(t_j-t_{r+1})+(1-c_l)(t_{r+1}-t_r).
\]
Since $1-c_l>0$, we have $\Re((1-c_l)(t_{r+1}-t_r))>0$, and $\Re(t_j-t_{r+1})\ge 0$. Hence,
\[
\Re(t_j-a_l)>0.
\]
For the third sum in equation~\eqref{eq:bform}, note that if $l\ge n_k$, then $a_l=t_k+c_l$; so for any $1\le j\le k$,
\[
a_l-t_j=(t_k-t_j)+c_l.
\]
Here $\Re(t_k-t_j)\ge 0$ and $c_l>0$, so
\[
\Re(a_l-t_j)>0.
\]
Finally, let us inspect the pair terms in equation \eqref{eq:bform}. For $1\le l<l'\le w$, we need to show that $\Re(a_{l'}-a_l)>0$. It is convenient to write each $a_l$ in the form
\[
  a_l=t_1+\xi_l+\sum_{r=1}^{k-1}\theta_{l,r}(t_{r+1}-t_r),
\]
where
\[
\xi_l:=
\begin{cases}
 c_l, & l<n_1\ \text{or}\ l\ge n_k,\\
 0, & n_1\le l<n_k,
\end{cases}
\]
and, for each $1\le r\le k-1$,
\[
\theta_{l,r}:=
\begin{cases}
0, & l<n_1,\\
1, & n_j\le l<n_{j+1}\ \text{for some }j\in\{1,\ldots,k-1\}\text{ with }r<j,\\
c_l, & n_j\le l<n_{j+1}\ \text{for some }j\in\{1,\ldots,k-1\}\text{ with }r=j,\\
0, & n_j\le l<n_{j+1}\ \text{for some }j\in\{1,\ldots,k-1\}\text{ with }r>j,\\
1, & l\ge n_k.
\end{cases}
\]
(For a fixed $l$ in the middle region $n_1\le l<n_k$, the index $j$ above is the unique one with $n_j\le l<n_{j+1}$.) Therefore,
\[
  a_{l'}-a_l
  =
  (\xi_{l'}-\xi_l)
  +\sum_{r=1}^{k-1}(\theta_{l',r}-\theta_{l,r})(t_{r+1}-t_r).
\]
From the explicit formulas and the ordering $c_1<\cdots<c_w$ on $\Delta_{n_1,\ldots,n_k}$, we have
\[
  \xi_{l'}-\xi_l\ge 0,
  \qquad
  \theta_{l',r}-\theta_{l,r}\ge 0\ \text{for all }r.
\]
Moreover, because $l<l'$, at least one of these inequalities is strict. Since each $\Re(t_{r+1}-t_r)>0$, it follows that
\[
  \Re(a_{l'}-a_l)>0.
\]
This completes the proof.
\end{proof}

Lemma \ref{lem:reparam} suggests the following candidate for the analytic continuation of $B_{n_1,\ldots,n_k}$ to $\Omega_k$. Take the same formula as in equation \eqref{eq:bform}, and replace each $g$ by its analytic continuation on the open right half-plane in the time coordinate. The following lemma shows that this works.
\begin{lmm}\label{lem:bnk}
If we define $B_{n_1,\ldots,n_k}$ on $\Omega_k$ as above, then it is an analytic function on this domain.
\end{lmm}
\begin{proof}
Let $\boldsymbol{n}:=(n_1,\ldots,n_k)$ and write
\[
  B_{\boldsymbol{n}}(t)=P_{\boldsymbol{n}}(t)\,I_{\boldsymbol{n}}(t),
\]
where $t=(t_1,\ldots,t_k)\in\Omega_k$,
\[
  P_{\boldsymbol{n}}(t):=w!\prod_{j=1}^{k-1}(t_{j+1}-t_j)^{n_{j+1}-n_j},
\]
and $I_{\boldsymbol{n}}(t)$ is the $\Delta_{\boldsymbol{n}}\times[-1,1]^w$ integral in equation \eqref{eq:bform}. Since each exponent $n_{j+1}-n_j$ is a nonnegative integer, $P_{\boldsymbol{n}}$ is entire. So it is enough to show that $I_{\boldsymbol{n}}$ is analytic on $\Omega_k$.

Fix $(c_1,\ldots,c_w)\in\Delta_{\boldsymbol{n}}$ and $(b_1,\ldots,b_w)\in[-1,1]^w$. By Lemma \ref{lem:reparam}, the integrand is holomorphic in $t$ pointwise. We need a suitable bound on the integrand to ensure that the holomorphy is retained after integration. 

For $\tau$ in the open right half-plane and $x\in\R$, let
\[
  R(\tau,x):=g(\tau,x)+\frac{\pi\tau}{2}.
\]
By Lemma~\ref{lem:gcontlmm}, there are universal constants $C_0,C_1,C_2>0$ such that
\begin{equation}\label{eq:bnk-rem-bounds}
  -C_0\le \Re R(\tau,x)\le C_1+|\log d(\tau,x)|,
  \qquad
  |\Im R(\tau,x)|\le C_2,
\end{equation}
where $d$ is the functional defined in Lemma~\ref{lem:gcontlmm}.  Let $E(t,c,b)$ be the exponent in the integrand, and decompose
\[
  E=E_{\mathrm{lin}}+E_{\mathrm{rem}},
\]
where $E_{\mathrm{lin}}$ is obtained by replacing each $g(\tau,x)$ by $-\pi\tau/2$, and $E_{\mathrm{rem}}$ is the sum of the remainder terms.

Fix $m\in\{1,\ldots,k\}$, fix $t_j$ for $j\neq m$, and let $\gamma$ be the boundary of a triangle in the corresponding one-variable slice of $\Omega_k$. We now bound the integrand uniformly for $t_m\in\gamma$.
Here and below, $O_\gamma(1)$ denotes any quantity whose absolute value is bounded by a constant depending on $\gamma$ (and the fixed data $b,w,\alpha_1,\ldots,\alpha_k,x_1,\ldots,x_k,\boldsymbol{n}$), but independent of $t_m\in\gamma$ and of the integration variables $(c_1,\ldots,c_w,b_1,\ldots,b_w)$. Write
\[
L:=\{1,\ldots,n_1-1\},\qquad M:=\{n_1,\ldots,n_k-1\},\qquad R:=\{n_k,\ldots,w\}.
\]
On $\Delta_{\boldsymbol{n}}$, we have $0<c_l<1$ for $l\in M$, so only the $c_l$'s with $l$ in $L$ or $R$ are unbounded. Take $l\in L$. In $E_{\mathrm{lin}}$, the insertion part involving this $l$ contributes
\[
2\pi b\sum_{j=1}^k\alpha_j(t_j-t_1-c_l)
=(-2\pi b^2w)(-c_l)+O_\gamma(1),
\]
using $\sum_{j=1}^k\alpha_j=-bw$. For pair terms, fix this $l$ and write the pair indices as $p<q$. Terms containing $c_l$ are of two kinds: $p=l<q$ and $p<q=l$. If $p=l<q$, then $2\pi b^2(a_q-a_l)$ contributes $+2\pi b^2(-c_l)$ (because $l<n_1$ implies $a_l=t_1+c_l$), and there are exactly $w-l$ such values of $q$. If $p<q=l$, then $2\pi b^2(a_l-a_p)$ contributes $-2\pi b^2(-c_l)$, and there are exactly $l-1$ such values of $p$. Therefore the net coefficient of $(-c_l)$ coming from pair linear terms is
\[
2\pi b^2((w-l)-(l-1))=2\pi b^2(w-2l+1).
\]
Thus, the net coefficient of $(-c_l)$ coming from all linear terms is 
\[
2\pi b^2(w-2l+1)-2\pi b^2w = -2\pi b^2(2l-1).
\]
Similarly, for $l\in R$, the linear parts of the insertion terms give
\[
2\pi b\sum_{j=1}^k\alpha_j(c_l+t_k-t_j)
=-2\pi b^2w c_l+O_\gamma(1),
\]
and the net coefficient of $c_l$ from the linear parts of the pair terms is 
\[
2\pi b^2((l-1)-(w-l)) =2\pi b^2(2l-w-1).
\]
Thus, the net coefficient of $c_l$ from all terms is
\[
-2\pi b^2w + 2\pi b^2(2l-w-1)= -2\pi b^2(2(w-l)+1).
\]
From the above discussion, we deduce that there exist constants $\sigma>0$ and $C_\gamma$, independent of the $b$'s and $c$'s, such that
\begin{equation}\label{eq:bnk-linear-tail}
\Re (E_{\mathrm{lin}})
\le C_\gamma
-\sigma\sum_{l\in L}(-c_l)
-\sigma\sum_{l\in R}c_l.
\end{equation}
Let us now inspect $E_{\mathrm{rem}}$ using equation \eqref{eq:bnk-rem-bounds}. Consider first one insertion remainder term. Write
\[
\alpha_j=\alpha_j^{\mathrm{R}}+\i\alpha_j^{\mathrm{I}},\qquad
R(\tau,x)=U(\tau,x)+\i V(\tau,x).
\]
Then
\[
\Re(-4b\alpha_jR(\tau,x))
=-4b(\alpha_j^{\mathrm{R}}U-\alpha_j^{\mathrm{I}}V)
\le -4b\alpha_j^{\mathrm{R}}U+4b|\alpha_j^{\mathrm{I}}|\,|V|.
\]
By equation~\eqref{eq:bnk-rem-bounds}, $|V|\le C_2$, so the second term is bounded by a constant. For the first term we split into two cases. If $\alpha_j^{\mathrm{R}}\ge 0$, then $U\ge -C_0$ gives 
\[
-4b\alpha_j^{\mathrm{R}}U\le 4b\alpha_j^{\mathrm{R}}C_0.
\]
If $\alpha_j^{\mathrm{R}}<0$, then $U\le C_1+|\log d(\tau,x)|$ gives 
\[
-4b\alpha_j^{\mathrm{R}}U
\le -4b\alpha_j^{\mathrm{R}}C_1 +(-4b\alpha_j^{\mathrm{R}})|\log d(\tau,x)|.
\]
Combining the two cases and absorbing all constant terms into $C_{j,\gamma}$ gives
\[
\Re(-4b\alpha_jR(\tau,x))
\le C_{j,\gamma}+p_j|\log d(\tau,x)|,
\qquad
p_j:=\max\{0,-4b\Re(\alpha_j)\}.
\]
Because $\Re(\alpha_j)>-\frac{1}{2b}$ by assumption, we get
\[
0\le p_j<2.
\]
Now consider one pair remainder term. Note that 
\[
\Re(-4b^2R(\tau,x))=-4b^2\Re R(\tau,x)\le 4b^2C_0,
\]
again by equation \eqref{eq:bnk-rem-bounds}. So pair remainder terms contribute only bounded constants and no logarithmic singular factors.

Summing all insertion and pair remainder terms, and absorbing all bounded contributions into a single constant, we obtain
\begin{equation}\label{eq:bnk-rem-est}
\Re (E_{\mathrm{rem}})
\le C_\gamma+\sum_{j=1}^k\sum_{l=1}^w p_j\,|\log d_{j,l}|,
\end{equation}
where $0\le p_j<2$ for each $j$ and $d_{j,l}$ denotes the corresponding $d$-quantity from Lemma~\ref{lem:gcontlmm}.
More explicitly, for each insertion factor we write it as $g(\tau_{j,l},\chi_{j,l})$ and set
\[
  d_{j,l}:=d(\tau_{j,l},\chi_{j,l})
  =\min\{d_+(\tau_{j,l},\chi_{j,l}),d_-(\tau_{j,l},\chi_{j,l})\},
\]
with
\[
  d_\pm(\tau,\chi):=\min_{q\in\Z}|\tau\pm \i\chi-2\i q|.
\]
Here the insertion arguments are exactly those from equation \eqref{eq:bform}:
\[
\tau_{j,l}=
\begin{cases}
t_j-a_l=t_j-t_1-c_l, & l<n_1,\\
a_l-t_j=c_l(t_{r+1}-t_r)+(t_r-t_j), & n_r\le l<n_{r+1},\ 1\le j\le r,\\
t_j-a_l=(t_j-t_{r+1})+(1-c_l)(t_{r+1}-t_r), & n_r\le l<n_{r+1},\ r+1\le j\le k,\\
a_l-t_j=t_k-t_j+c_l, & l\ge n_k,
\end{cases}
\]
and
\[
\chi_{j,l}=
\begin{cases}
x_j-b_l, & l<n_1,\\
b_l-x_j, & n_r\le l<n_{r+1},\ 1\le j\le r,\\
x_j-b_l, & n_r\le l<n_{r+1},\ r+1\le j\le k,\\
b_l-x_j, & l\ge n_k.
\end{cases}
\]
Exponentiating the inequality \eqref{eq:bnk-rem-est}, we get
\[
e^{\Re (E_{\mathrm{rem}})}
\le e^{C_\gamma}\prod_{j=1}^k\prod_{l=1}^w e^{p_j|\log d_{j,l}|}.
\]
For any $r>0$ and $p\ge 0$,
\begin{align*}
e^{p|\log r|}
&=
\begin{cases}
r^{-p}, & 0<r\le 1,\\
r^{p}, & r\ge 1,
\end{cases}\\ 
&= (r\wedge 1)^{-p}(r\vee 1)^p
\le (r\wedge 1)^{-p}(1+r)^p.
\end{align*}
Applying this with $r=d_{j,l}$ and absorbing $e^{C_\gamma}$ into $C_\gamma$, we get
\[
e^{\Re (E_{\mathrm{rem}})}
\le C_\gamma\prod_{j=1}^k\prod_{l=1}^w (d_{j,l}\wedge 1)^{-p_j}(1+d_{j,l})^{p_j}.
\]
The singular factor $(d_{j,l}\wedge 1)^{-p_j}$ depends only on the two real variables $(c_l,b_l)$ corresponding to that insertion. To show that the full product over $j$ is locally integrable for fixed $l$, we must rule out the possibility that several such factors become singular at the same point. Since $\gamma$ is a compact triangle contained in the one-variable slice of $\Omega_k$, there exists $\eta_\gamma>0$ such that
\[
  \Re(t_{r+1}-t_r)\ge \eta_\gamma
  \qquad\text{for all }1\le r<k\text{ and all }t_m\in\gamma.
\]
Consequently,
\[
  \Re(t_p-t_q)\ge \eta_\gamma\qquad\text{whenever }p>q.
\]
Fix $l$. We inspect the explicit formulas for $\tau_{j,l}$ case by case. If $l<n_1$, then $\tau_{j,l}=t_j-t_1-c_l$ and $c_l<0$. Hence for every $j>1$,
\[
  \Re\tau_{j,l}=\Re(t_j-t_1)-c_l\ge \eta_\gamma,
\]
so $d_{j,l}\ge \eta_\gamma$ and the corresponding factor is uniformly bounded. Thus only the $j=1$ factor can become singular. Its singular points satisfy $c_l=0$ and $b_l\equiv x_1\pmod 2$; because $b_l\in[-1,1]$, only finitely many such points occur.

If $n_r\le l<n_{r+1}$ for some $1\le r<k$, then $0<c_l<1$. For $1\le j\le r$ we have
\[
  \tau_{j,l}=c_l(t_{r+1}-t_r)+(t_r-t_j),
\]
while for $r+1\le j\le k$,
\[
  \tau_{j,l}=(t_j-t_{r+1})+(1-c_l)(t_{r+1}-t_r).
\]
If $j<r$, then $\Re(t_r-t_j)\ge \eta_\gamma$, so again $d_{j,l}\ge \eta_\gamma$. If $j>r+1$, then $\Re(t_j-t_{r+1})\ge \eta_\gamma$, so $d_{j,l}\ge \eta_\gamma$ as well. Therefore only the adjacent insertions $j=r$ and $j=r+1$ can be singular. Moreover,
\[
  \tau_{r,l}=c_l(t_{r+1}-t_r)
\]
can approach the singular set only when $c_l\to 0$, whereas
\[
  \tau_{r+1,l}=(1-c_l)(t_{r+1}-t_r)
\]
can approach it only when $c_l\to 1$. Hence these two possible singularities lie on different boundary faces of the strip $0<c_l<1$, so they cannot pile up at the same $(c_l,b_l)$. Again, for each face there are only finitely many relevant values of $b_l$ because $b_l\in[-1,1]$.

If $l\ge n_k$, then $\tau_{j,l}=t_k-t_j+c_l$ and $c_l>0$. For every $j<k$,
\[
  \Re\tau_{j,l}=\Re(t_k-t_j)+c_l\ge \eta_\gamma,
\]
so only the $j=k$ factor can become singular, necessarily near points with $c_l=0$ and $b_l\equiv x_k\pmod 2$.

We have therefore shown that for each fixed $l$, near any candidate singular point there is at most one unbounded factor among $\{(d_{j,l}\wedge 1)^{-p_j}:1\le j\le k\}$; all the other factors are uniformly bounded there. For that single potentially singular factor, the map
\[
  (c_l,b_l)\longmapsto \tau_{j,l}(c_l)\pm \i\chi_{j,l}(b_l)-2\i q
\]
is affine with invertible linear part (uniformly for $t_m\in\gamma$), so $d_{j,l}$ is comparable to Euclidean distance from the relevant singular point. Since $p_j<2$, this shows that the whole product
\[
  \prod_{j=1}^k (d_{j,l}\wedge 1)^{-p_j}
\]
is locally integrable in $(c_l,b_l)$. Because different values of $l$ involve disjoint variable pairs $(c_l,b_l)$, the full singular product $\prod_{j=1}^k\prod_{l=1}^w(d_{j,l}\wedge1)^{-p_j}$ is locally integrable on $\Delta_{\boldsymbol n}\times[-1,1]^w$.

For the growth factor, for $t_m\in\gamma$ every relevant $\tau$ is affine in $(c_l,t_m)$ with bounded coefficients, while $b_l\in[-1,1]$ and the $x_j$ are fixed. Choosing $k=0$ in the definition of $d_\pm$ gives
\[
d_{j,l}\le |\tau\pm \i x|\le C_\gamma(1+|c_l|),
\]
and so,
\[
(1+d_{j,l})^{p_j}\le C_\gamma'(1+|c_l|)^{p_j}.
\]
Thus the product of all growth factors has at most polynomial growth in the unbounded $c_l$ variables. Combining this with the linear tail bound \eqref{eq:bnk-linear-tail} yields
\[
|\text{integrand}(t,c,b)|\le \Psi_\gamma(t,c,b),
\]
where $\text{integrand}(t,c,b)$ denotes the integrand in equation \eqref{eq:bform}, and $\Psi_\gamma$ is an explicit product of:
\begin{itemize}
\item locally integrable singular factors $(d_{j,l}\wedge 1)^{-p_j}$ (with exponents $<2$),
\item polynomial factors in $|c_l|$, and
\item exponential tails in the unbounded directions $c_l\to -\infty$ for $l\in L$ and $c_l\to +\infty$ for~$l\in R$.
\end{itemize}
Now view $\Psi_\gamma$ as a function of $(t_m,c,b)$, with $(t_j)_{j\ne m}$ fixed. The hidden constants in the above bounds depend continuously on $t$ on $\Omega_k$; restricting to the compact slice obtained by varying $t_m\in\gamma$, they are uniformly bounded. Therefore the same singular-exponent and tail estimates imply
\[
\Psi_\gamma\in L^1(\gamma\times\Delta_{\boldsymbol n}\times[-1,1]^w),
\]
(with arc-length measure on $\gamma$), and
\[
|\text{integrand}(t,c,b)|\le \Psi_\gamma(t,c,b).
\]
Hence Fubini's theorem applies, giving:
\[
\int_{\gamma} I_{\boldsymbol{n}}(t)\,\dd t_m
=\int_{\Delta_{\boldsymbol{n}}}\int_{[-1,1]^w}
\biggl(\int_{\gamma}\text{integrand}(t,c,b)\,\dd t_m\biggr)
\dd b\,\dd c.
\]
The innermost integral on the right vanishes, since $\text{integrand}(t,c,b)$ is analytic in $t_m$.  By Morera's theorem, this shows that $I_{\boldsymbol n}$ is holomorphic in $t_m$. Since $m$ was arbitrary, $I_{\boldsymbol n}$ is separately holomorphic on $\Omega_k$. Thus, Hartogs' theorem~\cite[Theorem~2.2.8]{hormander73} implies that $I_{\boldsymbol n}$ is jointly holomorphic on $\Omega_k$. This completes the proof.
\end{proof}
We have now constructed the holomorphic continuation on the ordered domain. It remains to identify its Lorentzian boundary values and match them with the contour formula stated earlier.

\subsection{Calculating the Lorentzian correlators}\label{sec:mainproof}
The point of this final subsection is to convert the holomorphic continuation obtained above into an explicit contour formula on the ordered tube and then identify its boundary values with the time-ordered contour formula from equation~\eqref{eq:clorentz}. Recall the set $\Omega_k$ defined in equation~\eqref{eq:omegakplus}. Take any $(t_1,\ldots,t_k)\in \Omega_k$. Let $\gamma$ be the contour that comes in horizontally from $t_1-\infty$ to $t_1$, then goes via straight lines from $t_1$ to $t_2$, then from $t_2$ to $t_3$, and so on, until it reaches $t_k$. Then, finally, it goes horizontally from $t_k$ to $t_k+\infty$. A schematic is shown in Figure~\ref{fig:gamma-contour-general}.

\begin{figure}[t!]
\centering
\begin{tikzpicture}[x=1.2cm,y=1.0cm,>=stealth, scale = 1.2]
  \draw[->] (-3.2,0) -- (3.2,0) node[below right] {$\Re z$};
  \draw[->] (0,-2.4) -- (0,2.4) node[above left] {$\Im z$};

  \coordinate (T1) at (-1.85,-0.5);
  \coordinate (T2) at (-1.2,-0.98);
  \coordinate (T3) at (-0.45,0.28);
  \coordinate (T4) at (0.18,-0.06);
  \coordinate (T5) at (1.22,1.25);

  \draw[very thick,blue,->] (-3.0,-0.5) -- (T1);
  \draw[very thick,blue,->] (T1) -- (T2);
  \draw[very thick,blue,->] (T2) -- (T3);
  \draw[very thick,blue,->] (T3) -- (T4);
  \draw[very thick,blue,->] (T4) -- (T5);
  \draw[very thick,blue,->] (T5) -- (3.0,1.25);

  \fill (T1) circle (1.9pt) node[below=1pt] {$t_1$};
  \fill (T2) circle (1.9pt) node[above=1pt] {$t_2$};
  \fill (T3) circle (1.9pt) node[below = 1pt] {$t_3$};
  \fill (T4) circle (1.9pt) node[below = 1pt] {$t_4$};
  \fill (T5) circle (1.9pt) node[above = 1pt] {$t_5$};
  \node at (0.5,0.53) {$\gamma$};
\end{tikzpicture}
\caption{Contour $\gamma$ for analytic continuation in the case $k=5$, with each increment $t_{j+1}-t_j$ in the open first quadrant.}
\label{fig:gamma-contour-general}
\end{figure}
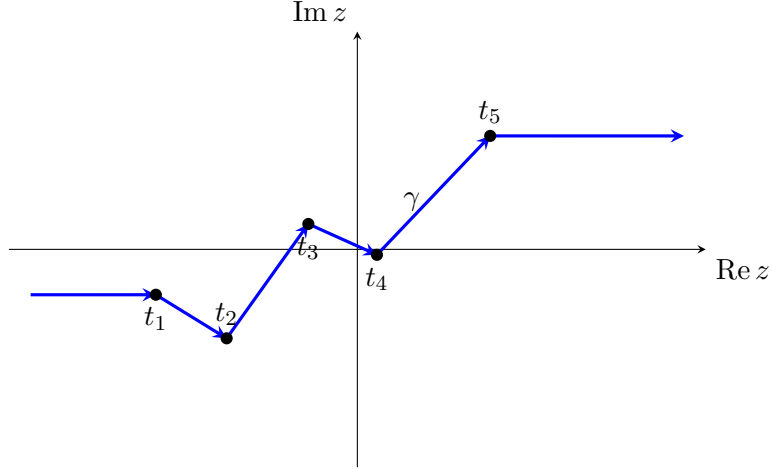

Fix a piecewise $C^1$ parametrization of $\gamma$ whose parameter $s$
increases from $-\infty$ to $\infty$ along the orientation of $\gamma$, and set
\[
  M_\gamma:=\gamma\times[-1,1].
\]
For $u=(\gamma(s),x)$ and $u'=(\gamma(s'),x')$ in $M_\gamma$, define the
\emph{contour-ordered difference}
\[
  \mathcal T_\gamma(u-u')
  :=
  \begin{cases}
    u-u' &\text{if } s> s',\\
    u'-u &\text{if } s< s'.
  \end{cases}
\]
When two screening variables have the same contour parameter, which occurs only
on a null set in the integrations, the value may be chosen arbitrarily. The
external points $v_j:=(t_j,x_j)\in M_\gamma$ are marked visits ordered by their
labels; hence, even if two visits have the same contour parameter, for $j<j'$
we set
\[
  \mathcal T_\gamma(v_j-v_{j'})=v_{j'}-v_j=(t_{j'}-t_j,x_{j'}-x_j).
\]
The opposite ordering is used for $j>j'$, and the same marked convention is used for mixed external-screening differences, with arbitrary choices on the measure-zero set of ties. We define contour integration on $M_\gamma^w$ by
\begin{align*}
\int_{M_\gamma^w} f(u_1,\ldots,u_w)\, \dd u_1 \cdots \dd u_w
&:= \int_{[-1,1]^w}\int_{\R^w}
    f((\gamma(s_1),b_1),\ldots,(\gamma(s_w),b_w)) \\
&\qquad \cdot \gamma'(s_1)\cdots \gamma'(s_w)\,
   \dd s_1\cdots \dd s_w\, \dd b_1 \cdots \dd b_w.
\end{align*}
where the finitely many corner points of the piecewise $C^1$ contour do not
affect the value of the integral. With these notations, we define
\begin{align}\label{eq:canacont}
&C_{\alpha_1,\ldots,\alpha_k}(t_1,x_1;\ldots; t_k,x_k)\notag \\
&:= \frac{(-\mu)^w}{w!}
\exp\biggl(
-4\sum_{1\le j<j'\le k}\alpha_{j}\alpha_{j'}
  g(\mathcal T_\gamma(v_j-v_{j'}))
\biggr) \notag \\
&\qquad \qquad \cdot \int_{M_\gamma^w} \exp\biggl(
    -4b\sum_{j=1}^k\sum_{l=1}^w \alpha_j g(\mathcal T_\gamma(v_j-u_l)) \notag \\
&\qquad \qquad \qquad \qquad   -4b^2\sum_{1\le l<l'\le w}g(\mathcal T_\gamma(u_l-u_{l'}))
  \biggr)\,\dd u_1\cdots \dd u_w,
\end{align}
where $u_1,\ldots,u_w\in M_\gamma$. This is the direct complex-time analogue
of the time-ordered Lorentzian formula \eqref{eq:clorentz}. The sector
decomposition below simply unwraps the $\mathcal T_\gamma$-notation and
recovers, sector by sector, the analytically continued formulas from
\secref{sec:coulomb}.
We will now prove the following analytic-continuation theorem on the ordered tube, which implies Theorem \ref{thm:main}.

\begin{thm}[Analytic continuation on the ordered tube]\label{thm:main2}
The function $C_{\alpha_1,\ldots,\alpha_k}$ defined in equation \eqref{eq:canacont} is analytic in the region $\Omega_k$. Moreover, it approaches the Euclidean correlation function $C_{\alpha_1,\ldots,\alpha_k}$ from equation~\eqref{eq:ceuc} as $(t_1,\ldots,t_k)$ approaches strictly increasing real $k$-tuples (and is the unique analytic function on $\Omega_k$ to do so), and the Lorentzian correlation function $C^{\mathrm{L}}_{\alpha_1,\ldots,\alpha_k}$ from equation~\eqref{eq:clorentz} as $(t_1,\ldots,t_k)$ approaches boundary points $(\i s_1,\ldots,\i s_k)$ with $s_1,\ldots,s_k\in\R$ satisfying $(s_j-s_{j'})\pm (x_j -x_{j'})\notin 2\Z$ for all $1\le j<j'\le k$. Lastly, this extension of $C_{\alpha_1,\ldots,\alpha_k}$ to this expanded domain is continuous, and the integral defining it is absolutely convergent everywhere on the domain.
\end{thm}
\begin{proof}
Fix $\alpha_1,\ldots,\alpha_k$ and $x_1,\ldots,x_k$, and write
\[
  C(t):=C_{\alpha_1,\ldots,\alpha_k}(t_1,x_1;\ldots;t_k,x_k),
\]
where $C_{\alpha_1,\ldots,\alpha_k}$ is the function defined in equation \eqref{eq:canacont}. 
Let $A(t)$ denote the free factor from equation \eqref{eq:canacont}. Since the
external points $v_1,\ldots,v_k$ occur along $\gamma$ in index order, this is
explicitly
\[
  A(t)=
  \exp\biggl(
    -4\sum_{1\le j<j'\le k}\alpha_j\alpha_{j'} g(t_{j'}-t_j,x_{j'}-x_j)
  \biggr).
\]
Let
\begin{align*}
  B^\gamma(t)
  &:=
  \int_{M_\gamma^w}
  \exp\biggl(
    -4b\sum_{j=1}^k\sum_{l=1}^w \alpha_j g(\mathcal T_\gamma(v_j-u_l)) \\
  &\qquad \qquad   -4b^2\sum_{1\le l<l'\le w}g(\mathcal T_\gamma(u_l-u_{l'}))
  \biggr)
  \,\dd u_1\cdots \dd u_w,
\end{align*}
so that
\[
  C(t)=\frac{(-\mu)^w}{w!}A(t)B^\gamma(t).
\]
Take $t$ in the domain $\Omega_k$ from the theorem (i.e., each $t_{j+1}-t_j$
lies in the open right half-plane), and consider the contour $\gamma$ from equation~\eqref{eq:canacont}. Write
\[
\gamma=\Gamma_0\cup\Gamma_1\cup\cdots\cup\Gamma_{k-1}\cup\Gamma_k,
\]
where
\begin{align*}
&\Gamma_0:=\{t_1+c:\ c<0\},\\
&\Gamma_r:=\{(1-c)t_r+ct_{r+1}:\ 0<c<1\}\quad (1\le r<k),\\
&\Gamma_k:=\{t_k+c:\ c>0\}.
\end{align*}
Because $\mathcal T_\gamma(u-u')=\mathcal T_\gamma(u'-u)$, the screening-screening
part of the integrand is symmetric under permutations of the screening
variables, and the mixed part is obviously symmetric as well. Hence we may
restrict to screening points $u_1,\ldots,u_w$ ordered along the orientation of
$\gamma$ and multiply by $w!$. We then partition according to how many ordered
points lie in each segment: for $1\le n_1\le\cdots\le n_k\le w+1$, the indices
\[
  l<n_1,\quad n_r\le l<n_{r+1}\ (1\le r<k),\quad l\ge n_k
\]
correspond exactly to
\[
  u_l\in\Gamma_0\times[-1,1],\qquad
  u_l\in\Gamma_r\times[-1,1],\qquad
  u_l\in\Gamma_k\times[-1,1],
\]
respectively. Writing $u_l=(a_l,b_l)$ and parameterizing the time coordinates by
\[
a_l=
\begin{cases}
t_1+c_l,& l<n_1,\\
(1-c_l)t_r+c_l t_{r+1},& n_r\le l<n_{r+1},\ 1\le r<k,\\
t_k+c_l,& l\ge n_k,
\end{cases}
\]
with $(c_1,\ldots,c_w)\in\Delta_{n_1,\ldots,n_k}$ gives exactly the change of
variables \eqref{eq:alcl} and Jacobian factor
\[
\dd a_1\cdots\dd a_w=
\biggl(\prod_{r=1}^{k-1}(t_{r+1}-t_r)^{n_{r+1}-n_r}\biggr)
\dd c_1\cdots\dd c_w.
\]
On such an ordered sector, the time-ordering is explicit:
\[
  \mathcal T_\gamma(u_l-u_{l'})=u_{l'}-u_l,
  \qquad 1\le l<l'\le w.
\]
Moreover,
\[
  \mathcal T_\gamma(v_j-u_l)=
  \begin{cases}
    v_j-u_l, & l<n_1,\\
    u_l-v_j, & n_r\le l<n_{r+1},\ 1\le j\le r,\\
    v_j-u_l, & n_r\le l<n_{r+1},\ r+1\le j\le k,\\
    u_l-v_j, & l\ge n_k,
  \end{cases}
\]
with $1\le r<k$ in the middle cases. Since
\[
  v_j-u_l=(t_j-a_l,x_j-b_l),\qquad
  u_l-v_j=(a_l-t_j,b_l-x_j),
\]
and
\[
  u_{l'}-u_l=(a_{l'}-a_l,b_{l'}-b_l),
\]
substituting into the time-ordered contour integral yields, term by term in
$(n_1,\ldots,n_k)$, precisely the formula \eqref{eq:bform} used in
Lemma~\ref{lem:bnk} (with $g$ interpreted via its analytic continuation on
$\Re\tau>0$). Therefore
\[
  B^\gamma(t)
  =
  \sum_{1\le n_1\le\cdots\le n_k\le w+1}B_{n_1,\ldots,n_k}(t),
\]
where the right side is exactly the analytically continued Coulomb decomposition
introduced in \secref{sec:coulomb}.

Now, the free part $A$ is analytic on $\Omega_k$ by the discussion in \secref{sec:free}. Also, by Lemma~\ref{lem:bnk}, each $B_{n_1,\ldots,n_k}$ is analytic on that same region. Since there are only finitely many index tuples $(n_1,\ldots,n_k)$, $B^\gamma$ is analytic there. Therefore $C$ is analytic on $\Omega_k$.

Next, let $s=(s_1,\ldots,s_k)$ with $s_1<\cdots<s_k$ real, and let $t^{(m)}\in\Omega_k$ satisfy $t^{(m)}\to s$. For each index tuple $\boldsymbol n=(n_1,\ldots,n_k)$, write
\[
  B_{\boldsymbol n}(t)=P_{\boldsymbol n}(t)I_{\boldsymbol n}(t)
\]
as in the proof of Lemma~\ref{lem:bnk}. The factor $P_{\boldsymbol n}$ is polynomial, hence continuous at $s$. For $I_{\boldsymbol n}$, the estimates proved in Lemma~\ref{lem:bnk} give an $L^1$ majorant uniform for $t$ in a small compact neighborhood of $s$ inside $\Omega_k$: the singular exponents are still $<2$, and the same exponential tails in the unbounded $c_l$-directions hold. Since each time argument in equation \eqref{eq:bform} stays in the open right half-plane for $t\in\Omega_k$ and converges to a positive real value at $t=s$, and $g(\tau,x)$ is continuous up to that boundary, the integrand converges pointwise. Therefore, by dominated convergence,
\[
  I_{\boldsymbol n}(t^{(m)})\to I_{\boldsymbol n}(s),\qquad
  B_{\boldsymbol n}(t^{(m)})\to B_{\boldsymbol n}(s).
\]
Summing over finitely many $\boldsymbol n$ gives
\[
  B^\gamma(t^{(m)})\to \sum_{\boldsymbol n} B_{\boldsymbol n}(s).
\]
At real $s$, this finite sum is exactly the Euclidean Coulomb factor (the same partition and change of variables that produced equation \eqref{eq:bform} from the Euclidean integral), i.e., the $B$-term in equation \eqref{eq:ceuc}. Also $A(t^{(m)})\to A(s)$. Hence
\[
  C(t^{(m)})\to C_{\alpha_1,\ldots,\alpha_k}(s_1,x_1;\ldots;s_k,x_k)
\]
with the right-hand side given by equation \eqref{eq:ceuc}. Since the sequence
$t^{(m)}\to s$ was arbitrary in $\Omega_k$, this proves the Euclidean boundary
limit.

For uniqueness on $\Omega_k$, let $\widetilde C$ be any analytic function on
$\Omega_k$ with the same Euclidean boundary values, and set
\[
  H:=\widetilde C-C.
\]
Then $H$ is analytic on $\Omega_k$ and has boundary limit $0$ at every point of
\[
  \mathcal E_k:=\{(r_1,\ldots,r_k)\in\R^k:r_1<\cdots<r_k\}.
\]
Introduce coordinates
\[
  u_1:=t_1,\qquad u_j:=t_j-t_{j-1}\quad (2\le j\le k),
\]
which biholomorphically map $\Omega_k$ onto
\[
  D_k:=\C\times {\mathbb H}_+^{k-1},\qquad \mathbb{H}_+:=\{z\in\C:\Re z>0\},
\]
and map $\mathcal E_k$ onto
\[
  B_k:=\R\times(0,\infty)^{k-1}.
\]
Define
\[
  \widehat H(u_1,\ldots,u_k):=H(u_1,u_1+u_2, \ldots, u_1+\cdots+u_k).
\]
Then $\widehat H$ is analytic on $D_k$ and is zero on $B_k$. Therefore,
Lemma~\ref{lem:uniq-analytic-cont} from~\secref{sec:unique} gives
$\widehat H\equiv 0$ on $D_k$, and hence $H\equiv 0$ on $\Omega_k$. So the
analytic continuation to $\Omega_k$ is unique.

Finally, let $s = (s_1,\ldots,s_k)$ be a real $k$-tuple satisfying the condition that $(s_j-s_{j'})\pm (x_j - x_{j'})\notin 2\Z$ for all $1\le j<j'\le k$, and let $t^{(m)}\in\Omega_k$ satisfy $t^{(m)}\to \i s$ as $m\to\infty$. Set
\[
  v_j:=(\i s_j,x_j),\qquad 1\le j\le k,
\]
and let $\widetilde\gamma$ be the contour from the definition \eqref{eq:clorentz} associated with $v_1,\ldots,v_k$. For each index tuple $\boldsymbol n=(n_1,\ldots,n_k)$, we again write
\[
  B_{\boldsymbol n}(t)=P_{\boldsymbol n}(t)I_{\boldsymbol n}(t).
\]
As above, $P_{\boldsymbol n}(t^{(m)})\to P_{\boldsymbol n}(\i s)$. For $I_{\boldsymbol n}$, one must refine the estimate from Lemma~\ref{lem:bnk} slightly and use the non-light-cone hypothesis on $s$. The exponential tail bounds in the unbounded $c_l$-directions are unchanged. On the middle part of the limiting contour, after writing $a=\i u$, the possible insertion singularities are governed by the affine lines
\[
u\mp b=s_j\mp x_j+2m,\qquad m\in\Z.
\]
Because $(s_j-s_{j'})\pm(x_j-x_{j'})\notin 2\Z$, no two parallel lines coming from different insertions coincide; any remaining intersections are transverse. Therefore, after an affine change of coordinates near any candidate singular point, the local singular factor is of the form $|u|^{-q_r}$ or $|u|^{-q_r}|v|^{-q_{r'}}$, where
\[
q_j:=\max\{0,-2b\Re(\alpha_j)\}<1.
\]
These model singularities are locally integrable, and by compactness they give an $L^1$ majorant uniform for $t$ in a small compact neighborhood of $\i s$ inside $\overline \Omega_k$. Moreover, by the boundary-value formula \eqref{eq:gimag}, each $g$-factor in \eqref{eq:bform} converges pointwise as $t^{(m)}\to \i s$. 

The limiting argument is precisely the time-ordered difference determined by $\widetilde\gamma$: if an ordered screening point lies before $v_1$ on $\widetilde\gamma$, then every mixed factor is of the form $g(v_j-u_l)$; if it lies on the segment from $v_r$ to $v_{r+1}$, then the mixed factors are $g(u_l-v_j)$ for $j\le r$ and $g(v_j-u_l)$ for $j\ge r+1$; and if it lies after $v_k$, then every mixed factor is of the form $g(u_l-v_j)$. (When $v_r=v_{r+1}$, that segment is empty and there is nothing to check.) 

Likewise, for screening points ordered along the contour one has $\mathcal T(u_{l'}-u_l)=u_{l'}-u_l$ whenever $l<l'$. Therefore, by dominated convergence, each sector integral converges to the corresponding ordered-sector contribution in the contour integral from \eqref{eq:clorentz}. Performing on that contour integral the same ordered-sector decomposition as above shows that the sum of these limits is exactly its screening part. Also, since $v_{j'}$ occurs later than $v_j$ on $\widetilde\gamma$ whenever $j<j'$, we have $\mathcal T(v_j-v_{j'})=v_{j'}-v_j$, and so the continuity of \eqref{eq:gimag} gives convergence of the free factor to the free factor in \eqref{eq:clorentz}. Summing over finitely many $\boldsymbol n$, and combining with the convergence of the free factor $A(t^{(m)})$, we obtain
\[
  C(t^{(m)})\longrightarrow C^{\mathrm L}_{\alpha_1,\ldots,\alpha_k}(s_1,x_1;\ldots;s_k,x_k),
\]
which is exactly equation \eqref{eq:clorentz}. This proves the Lorentzian limit.

It remains to justify absolute convergence of the integral in \eqref{eq:clorentz}, and continuity of the resulting boundary function on its domain. For any target point satisfying the same non-light-cone condition, the very same sectorwise argument applies in a sufficiently small neighborhood: the tail bounds are unchanged, parallel singular lines stay separated by the non-light-cone hypothesis, and the only remaining intersections are transverse, hence locally integrable. Evaluating the resulting majorant at the target point gives absolute convergence of the contour integral, and dominated convergence for nearby boundary points shows that the value of the contour integral in \eqref{eq:clorentz} depends continuously on $(t_1,x_1),\ldots,(t_k,x_k)$. Hence the analytic continuation extends continuously to the full domain of $C^{\mathrm L}_{\alpha_1,\ldots,\alpha_k}$.
\end{proof}
This completes the ordered-tube analytic-continuation theorem and, with it, the proof of Theorem~\ref{thm:main}.

\section{Quantization}\label{sec:quantization}
In this section, we carry out the quantization program outlined in \secrefrange{sec:smear}{sec:gnsquant}. The first step is to prove Theorem \ref{thm:smear}.
We first establish convergence for smeared observables, then construct the state space and pairing, next analyze translation covariance on the represented state space, and finally prove locality for the represented local net.

\subsection{Expectations of smeared observables}\label{sec:smearproof}
We prove Theorem~\ref{thm:smear}. Fix $n_1,\ldots,n_k\in\mathcal I_b$ and
$f_1,\ldots,f_k\in C_c^\infty(M,\C)$. Let
$\boldsymbol n=(n_1,\ldots,n_k)$ and set $\alpha_j:=bn_j$. Since
$n_j\in\mathcal I_b$, we have
\[
  |\alpha_j|<\frac1{\sqrt2},\qquad 1\le j\le k.
\]
In particular, because $b<\frac{1}{\sqrt2}$,
\[
  \alpha_j>-\frac1{\sqrt2}>-\frac1{2b}.
\]
Let
\[
  w:=-\sum_{j=1}^k n_j.
\]
If $w<0$, then by definition $C^{\mathrm L}_{b\boldsymbol n}\equiv 0$, and
there is nothing to prove. We therefore assume throughout the proof that
$w\ge0$. The Lorentzian contour formula of Theorem~\ref{thm:main} is then
available for the charges $\alpha_1,\ldots,\alpha_k$. Choose numbers
\[
  T_-<S_-<S_+<T_+
\]
such that
\[
  \supp(f_j)\subset [S_-,S_+]\times[-1,1],
  \qquad 1\le j\le k,
\]
and put
\[
  K:=[S_-,S_+]\times[-1,1].
\]
Since $\prod_{j=1}^k f_j(u_j)$ is bounded and supported in $K^k$, it is enough
to prove that
\begin{align}\label{tag51}
  \int_{K^k}
  |C^{\mathrm L}_{b\boldsymbol n}(u_1,\ldots,u_k)|
  \,\dd u_1\cdots \dd u_k<\infty .
\end{align}
Changing the values of the point correlator on the light-cone exceptional set
does not affect the integral, because that exceptional set has Lebesgue measure
zero. The set $K^k$ is the union, up to a null set, of the finitely many ordered
sectors
\[
  D_\sigma
  :=
  \{((t_1,x_1),\ldots,(t_k,x_k))\in K^k:
      t_{\sigma(1)}<\cdots<t_{\sigma(k)}\},
  \qquad \sigma\in S_k.
\]
It suffices to prove the claim \eqref{tag51} on each $D_\sigma$. We first treat the
standard sector
\[
  D:=\{((t_1,x_1),\ldots,(t_k,x_k))\in K^k:
      t_1<\cdots<t_k\}.
\]
Fix $\boldsymbol u=(u_1,\ldots,u_k)\in D$, with $u_j=(t_j,x_j)$, and write
$v_j=(\i t_j,x_j)$. On this sector, the contour in equation \eqref{eq:clorentz} is
\[
  \gamma=\Gamma_-\cup \Gamma_0\cup \Gamma_+,
\]
where
\[
  \Gamma_-:=\{\i t_1-r:r>0\},\qquad
  \Gamma_0:=\{\i s:t_1<s<t_k\},\qquad
  \Gamma_+:=\{\i t_k+r:r>0\}.
\]
Theorem~\ref{thm:main} gives
\[
  C^{\mathrm L}_{b\boldsymbol n}(u_1,\ldots,u_k)
  =
  \frac{(-\mu)^w}{w!}A(\boldsymbol u)
  \int_{M_\gamma^w}
  B_{\mathcal T}(\boldsymbol u;U_1,\ldots,U_w)
  \,\dd U_1\cdots \dd U_w,
\]
where
\[
  A(\boldsymbol u)
  :=
  \exp\biggl(
    -4\sum_{1\le i<j\le k}
    \alpha_i\alpha_j g(\i(t_j-t_i),x_j-x_i)
  \biggr),
\]
and
\begin{align*}
  B_{\mathcal T}(\boldsymbol u;U_1,\ldots,U_w)
  &:=
  \exp\biggl(
    -4b\sum_{i=1}^k\sum_{m=1}^w
      \alpha_i g(\mathcal T(v_i-U_m)) \\ 
&\qquad \qquad     -4b^2\sum_{1\le m<m'\le w}
      g(\mathcal T(U_m-U_{m'}))
  \biggr).
\end{align*}
Here $U_m=(a_m,b_m)\in M_\gamma$, and $\mathcal T$ denotes the
contour-ordering in equation \eqref{eq:clorentz}. The screening integrand is symmetric
in $U_1,\ldots,U_w$, because for screening-screening factors the ordered
difference is always the later point minus the earlier point along the
contour. Thus, after canceling the harmless factor $w!$, it is enough to
estimate the integral over sectors in which the screening variables are ordered
along $\gamma$.

We decompose such an ordered screening sector according to how many screening
variables lie on $\Gamma_-$, how many lie on the bounded vertical part, and how
many lie on $\Gamma_+$. Fix one such sector. Suppose that the first $p$
screening variables lie on $\Gamma_-$, the next $q$ lie on the bounded vertical
part, and the remaining $w-p-q$ lie on $\Gamma_+$. For each middle screening
variable choose a fixed $\rho_m\in\{1,\ldots,k-1\}$ such that this variable
lies on the segment from $\i t_{\rho_m}$ to $\i t_{\rho_m+1}$. We write
\begin{align*}
  &a_l=\i t_1-r_l,\qquad 1\le l\le p,\\ 
  &a_{p+m}=\i s_m,\qquad
  t_{\rho_m}<s_m<t_{\rho_m+1},\qquad 1\le m\le q,\\ 
  &a_l=\i t_k+r'_l,\qquad p+q+1\le l\le w.
\end{align*}
The ordered-sector condition imposes only ordering, or simplex, inequalities
on the coordinates along each contour piece. More explicitly, the orientation
on the left tail gives
\[
  r_1>\cdots>r_p>0,
\]
the orientation on the right tail gives
\[
  0<r'_{p+q+1}<\cdots<r'_w,
\]
and the vertical part orders the variables $s_m$ within each chosen segment,
in addition to the restrictions $t_{\rho_m}<s_m<t_{\rho_m+1}$. For upper
bounds we may enlarge this ordered sector to the corresponding product region.
In particular, we simplify the integration domains to 
\[
  r_l>0,\qquad r'_l>0,\qquad s_m\in(T_-,T_+).
\]
We first isolate the decay of the horizontal variables. Write
\[
  g(\tau,x)=-\frac{\pi\tau}{2}+R(\tau,x),
  \qquad \Re\tau>0.
\]
Let $E_{\rm lin}$ be the part of the exponent in
$A(\boldsymbol u)B_{\mathcal T}(\boldsymbol u;U_1,\ldots,U_w)$ obtained by
replacing every occurrence of $g(\tau,x)$ involving at least one screening
variable by $-\pi\tau/2$. The free factor $A(\boldsymbol u)$ contains no
horizontal screening variable and will be treated later in the compact
light-cone estimate.

Consider first a left-tail variable $U_l$, $1\le l\le p$. The mixed
external-screening linear terms involving $U_l$ contribute, in real part,
\[
  2\pi b\sum_{i=1}^k \alpha_i r_l
  =
  -2\pi b^2 w\,r_l,
\]
because $\sum_i\alpha_i=-bw$. The screening-screening linear terms involving
$U_l$ contribute
\[
  2\pi b^2((w-l)-(l-1))r_l
  =
  2\pi b^2(w-2l+1)r_l.
\]
Hence the total coefficient of $r_l$ in $\Re E_{\rm lin}$ is
\[
  -2\pi b^2(2l-1).
\]
Similarly, if $p+q+1\le l\le w$ is a right-tail variable, then the mixed
external-screening linear terms give $-2\pi b^2w\,r'_l$, while the
screening-screening linear terms give
\[
  2\pi b^2((l-1)-(w-l))r'_l
  =
  2\pi b^2(2l-w-1)r'_l.
\]
Thus the total coefficient of $r'_l$ in $\Re E_{\rm lin}$ is
\[
  -2\pi b^2(2(w-l)+1).
\]
It follows that there is a constant $C$ such that, uniformly for
$\boldsymbol u\in D$ and for all variables in the fixed sector,
\[
  \Re E_{\rm lin}
  \le
  C
  -2\pi b^2\sum_{l=1}^p(2l-1)r_l
  -2\pi b^2\sum_{l=p+q+1}^w(2(w-l)+1)r'_l .
\]
The remainder terms are controlled by Lemma~\ref{lem:gcontlmm}. Pair
screening-screening remainder terms with coefficient $-4b^2$ are bounded above
because $\Re R$ is bounded below. Mixed external-screening remainder terms are
bounded above by a constant times a finite product of logarithmic factors.
Consequently there are constants $C,N>0$ and positive constants $c_l,c'_l$
such that all terms involving large horizontal variables are bounded by
exponential decay times a polynomial factor.

Now split each left and right horizontal variable into a small and a large
part:
\[
  0\le r_l\le 1 \quad\text{or}\quad r_l>1,
  \qquad
  0\le r'_l\le 1 \quad\text{or}\quad r'_l>1.
\]
There are only finitely many such choices. Fix one. Let $L_0$ and $R_0$ be the
sets of left and right horizontal variables whose corresponding parameters are
at most $1$, and let $L_\infty$ and $R_\infty$ be the complementary sets.

If $l\in L_\infty$ or $l\in R_\infty$, then every mixed
external-horizontal argument involving that variable has real part at least
$1$. Hence such a variable cannot produce a light-cone singularity with any
external insertion. Screening-screening factors involving such a variable have
positive exponent $2b^2$ in each light-cone coordinate at possible
coincidences, and therefore give no singular upper bound. Combining this
observation with the preceding linear estimate and Lemma~\ref{lem:gcontlmm},
we obtain, uniformly in all remaining variables,
\[
  |\hbox{large-horizontal part of the integrand}|
  \le
  C
  \prod_{l\in L_\infty} e^{-c_l r_l}(1+r_l)^N
  \prod_{l\in R_\infty} e^{-c'_l r'_l}(1+r'_l)^N,
\]
for some $C, c_l, c_l', N > 0$. 
The right side is integrable over the large horizontal variables. After
integrating those variables out, it remains to prove integrability of the
compact part consisting of the external variables, the middle screening
variables, and the small horizontal variables.

We use the following terminology for the compact power-counting argument. For a
vertical point $(\i s,y)$, its $\sigma$ light-cone coordinate, with
$\sigma\in\{+,-\}$ identified with $\pm1$, is $s+\sigma y$. For an external
point $(t_i,x_i)$ this means $t_i+\sigma x_i$. For a horizontal screening point
$(\i s+h,y)$, we keep the same light-cone coordinate $s+\sigma y$ and call
$h$ its height coordinate; thus vertical points have height $0$, left-tail
points have negative height, and right-tail points have positive height. A
light-cone singularity is a possible singularity caused by equality, modulo the
period $2$, of such light-cone coordinates, together with equality of the
height coordinates when horizontal screening points are involved.

We now describe a majorant for the compact integral, and then use power counting to show that the majorant has a finite integral. Write the middle
screening variables as
\[
  W_m=(\i s_m,y_m),\qquad 1\le m\le q,
\]
with $s_m\in(T_-,T_+)$ and $y_m\in[-1,1]$. For $l\in L_0$ write
\[
  L_l=(\i t_1-r_l,y_l^-),
  \qquad 0\le r_l\le1,\quad y_l^-\in[-1,1],
\]
and for $l\in R_0$ write
\[
  R_l=(\i t_k+r_l^+,y_l^+),
  \qquad 0\le r_l^+\le1,\quad y_l^+\in[-1,1].
\]
For each $\sigma\in\{+,-\}$, identified with $\pm1$, set
\[
  p_i^\sigma:=t_i+\sigma x_i,\qquad
  q_m^\sigma:=s_m+\sigma y_m.
\]
For $\theta \in \R$, let 
\[
  \delta(\theta):=\dist(\theta,2\Z),
\]
and for $(a,h),(a',h')\in\R^2$ define
\[
  \Delta((a,h),(a',h'))
  :=
  (\delta(a-a')^2+|h-h'|^2)^{1/2}.
\]
The external light-cone points are
\[
  P_i^\sigma:=(p_i^\sigma,0),\qquad 1\le i\le k.
\]
The screening light-cone points left after the large variables have been
integrated out are
\[
  \mathscr S_\sigma
  :=
  \{(q_m^\sigma,0):1\le m\le q\}
  \cup
  \{(t_1+\sigma y_l^-,-r_l):l\in L_0\}
  \cup
  \{(t_k+\sigma y_l^+,r_l^+):l\in R_0\}.
\]
Thus each point in $\{P_i^\sigma\}\cup\mathscr S_\sigma$ is written as
$(a,h)$, where $a$ is the relevant $\sigma$ light-cone coordinate and $h$ is
the height coordinate.
The sign convention for the height of a left-tail point is only notational:
with this convention, left-left and right-right screening differences have
heights $|r_l-r_{l'}|$ and $|r_l^+-r_{l'}^+|$, while a left-right difference
has height $r_l+r_{l'}^+$.

For real $\beta,\rho,\xi$, the boundary formula \eqref{eq:gimag} gives
\begin{align}\label{eq:comp1}
  |e^{-4\beta g(\i\rho,\xi)}|
  =
  |1-e^{-\pi\i(\rho+\xi)}|^{2\beta}
  |1-e^{-\pi\i(\rho-\xi)}|^{2\beta}.
\end{align}
Also, there is a universal constant $C\ge1$ such that for all $\theta\in\R$,
\begin{align}\label{eq:comp2}
  C^{-1}\delta(\theta)
  \le
  |1-e^{-\pi\i\theta}|
  \le
  C\delta(\theta).
\end{align}
Finally, choose $m\in\Z$ such that $|\theta-2m|=\delta(\theta)$. Since
\[ 
  1-e^{-\pi h-\pi\i\theta}=1-e^{-\pi h-\pi\i(\theta-2m)},
\] 
applying
Lemma~\ref{lem:one-minus-exp-lower} with $u=\pi h$ and
$v=-\pi(\theta-2m)\in[-\pi,\pi]$ gives, uniformly for $0\le h\le2$ and
$\theta\in\R$,
\[
  |1-e^{-\pi h-\pi\i\theta}|
  \asymp
  (h^2+\delta(\theta)^2)^{1/2}.
\]
The exponential factor $e^{2\pi\beta h}$ coming from the linear term of
$g(h+\i\rho,\xi)$ is bounded above and below by constants on $0\le h\le2$.

Applying these comparisons to the free external-external factors, the
external-screening factors, and the screening-screening factors, and using
$\delta(\theta)=\delta(-\theta)$ to remove the contour-order signs, the
remaining compact absolute integrand is bounded by a constant times
$G_+G_-$, where
\[
  G_\sigma
  :=
  \prod_{1\le i<j\le k}
    \Delta(P_i^\sigma,P_j^\sigma)^{2\alpha_i\alpha_j}
  \prod_{i=1}^k\prod_{Z\in\mathscr S_\sigma}
    \Delta(P_i^\sigma,Z)^{2b\alpha_i}
  \prod_{\substack{Z,Z'\in\mathscr S_\sigma\\ Z<Z'}}
    \Delta(Z,Z')^{2b^2}.
\]
Here $Z<Z'$ denotes any fixed ordering of the finite set
$\mathscr S_\sigma$. The order is irrelevant for integrability.

It remains to prove that $G_+G_-$ is locally integrable on the compact domain
of the variables
\[
  (t_i,x_i)_{i=1}^k,\qquad
  (s_m,y_m)_{m=1}^q,\qquad
  (r_l,y_l^-)_{l\in L_0},\qquad
  (r_l^+,y_l^+)_{l\in R_0}.
\]
Since the domain is compact, local integrability suffices. Fix a point of this
domain. After shrinking to a sufficiently small coordinate neighborhood, every
factor $\delta(a-a')$ is either bounded above and below by positive constants
or is comparable to $|a-a'-2m|$ for a uniquely determined $m\in\Z$. Thus every
singular factor in $G_+G_-$ is locally comparable to a Euclidean norm of an
affine function of the integration variables. Hence the affine-distance form of
Lemma~\ref{lem:1d-power-counting} applies directly. It remains only to verify,
for each flat in this finite arrangement, that the sum of the exponents of the
factors vanishing on that flat is greater than the negative of its codimension.

Let us spell out what these flats are in the present situation. For a fixed
sign $\sigma\in\{+,-\}$, the factors in $G_\sigma$ are indexed by pairs of
points in the finite set consisting of the external light-cone points
$P_i^\sigma$ and the screening light-cone points in $\mathscr S_\sigma$. After
the localization just described, the vanishing of one such pair factor is the
affine condition that the two corresponding light-cone coordinates agree
modulo the locally fixed element of $2\Z$, together with equality of the height
coordinates appearing in $\Delta$. Thus a flat for one fixed $\sigma$ is
obtained by imposing a finite collection of such pairwise coincidences. If one
draws a graph whose vertices are the points $P_i^\sigma$ and
$Z\in\mathscr S_\sigma$, and whose edges are the pair factors required to
vanish, then the nontrivial connected components of this graph are exactly the
collapsing clusters: all points in such a component are forced to coincide in
the corresponding one-dimensional light-cone coordinate, with the additional
height equalities when horizontal screening points are present. Isolated
vertices impose no condition. Conversely, every one-family collapsing cluster,
or finite collection of disjoint such clusters, gives one of these flats.

For the product $G_+G_-$, let $F$ be an arbitrary flat in the combined
arrangement. For each sign $\sigma\in\{+,-\}$, define a graph
$\Gamma_\sigma(F)$ on the vertices $P_i^\sigma$ and $Z\in\mathscr S_\sigma$ by
joining two vertices exactly when the corresponding pair factor in $G_\sigma$
vanishes identically on $F$. Equivalently, two vertices are joined when, on
$F$, their $\sigma$ light-cone coordinates are forced to agree up to the
locally fixed element of $2\Z$, and their height coordinates are forced to
agree whenever the height condition is nontrivial. Thus the edge relation is
the pairwise part of the equivalence relation ``being forced to coincide in
the $\sigma$ family on $F$.'' Passing from the originally imposed edges to all
edges whose zero set contains $F$ simply saturates the graph under all
coincidences that follow from the definition of $F$. Hence the connected
components of $\Gamma_\sigma(F)$ are the equivalence classes of this forced
coincidence relation. The nontrivial components are precisely the maximal
$\sigma$ light-cone collapsing clusters of $F$. For such a component $C$, let
$A(C)$ be the set of external labels in $C$, let $n(C)$ be the number of
screening vertices in $C$, and set

\[
  r(C):=|A(C)|+n(C)-1.
\]
For any set $A$ of external labels and any $n\ge0$, define
\[
  \Theta(A,n)
  :=
  2\sum_{\substack{i,j\in A\\ i<j}}\alpha_i\alpha_j
  +
  2bn\sum_{i\in A}\alpha_i
  +
  b^2 n(n-1).
\]
The contribution of all factors whose zero sets contain $F$ is therefore
\[
  E(F)
  =
  \sum_{\sigma\in\{+,-\}}
  \sum_{C\in\mathcal C_\sigma(F)}
  \Theta(A(C),n(C)),
\]
where $\mathcal C_\sigma(F)$ denotes the set of nontrivial components of
$\Gamma_\sigma(F)$. Define also
\[
  R(F)
  :=
  \sum_{\sigma\in\{+,-\}}
  \sum_{C\in\mathcal C_\sigma(F)} r(C).
\]
We claim that $\operatorname{codim}F\ge R(F).$ To see this, write $a_Z^\tau$ and
$h_Z$ for the light-cone and height coordinates of a vertex $Z$ in the
$\tau$-family. Choose a root $\rho_C\in C$ in every component
$C\in\mathcal C_\tau(F)$. The light-cone part of the collapse equations
contains
\[
  \mathcal E(F):\qquad
  a_Z^\tau-a_{\rho_C}^\tau=c_{Z,C},
  \qquad
  C\in\mathcal C_\tau(F),\quad Z\in C\setminus\{\rho_C\},
  \quad \tau\in\{+,-\},
\]
for suitable locally fixed constants $c_{Z,C}$. The number of equations in
$\mathcal E(F)$ is exactly $R(F)$. Let $\mathcal H(F)$ denote the height
equalities forced by the same vanishing $\Delta$-factors, namely
$h_Z=h_{\rho_C}$ whenever one of $Z$ and $\rho_C$ is horizontal.

If all light-cone coordinates were independent, the equations in
$\mathcal E(F)$ would already have rank $R(F)$. In the actual variables, the
only possible loss of rank comes from horizontal screening vertices attached to
the same endpoint of the contour. If $e(l)=1$ for a left-tail point and
$e(l)=k$ for a right-tail point, then
\[
  a_l^\tau-p_{e(l)}^\tau=\tau(y_l-x_{e(l)}),
  \qquad \tau\in\{+,-\}.
\]
Thus a $+$ relative equation and a $-$ relative equation can duplicate each
other only through relations of the form $y_l-y_{l'}=\mathrm{const}$, or
$y_l-x_{e(l)}=\mathrm{const}$, among horizontal points with the same endpoint
anchor. But in exactly these cases the corresponding $\Delta$-vanishing also
imposes the independent height relation $h_l-h_{l'}=0$, or $h_l=0$. Hence the
height equations compensate any such loss of light-cone rank:
\[
  \operatorname{rank}\mathcal E(F)+\operatorname{rank}\mathcal H(F)
  \ge R(F).
\]
Since the height variables are independent of the light-cone variables, and
since $F$ is contained in the common affine zero set of these equations,
\[
  \operatorname{codim}F
  \ge
  \operatorname{rank}\mathcal E(F)+\operatorname{rank}\mathcal H(F)
  \ge R(F).
\]
Thus it remains to prove the one-family cluster inequality
$\Theta(A,n)>-(|A|+n-1)$, since summing that inequality over the components of
an arbitrary $F$ will give 
\[ 
  E(F)>-R(F)\ge-\operatorname{codim}F.
\]
We now check these inequalities explicitly. Consider first a cluster in one
fixed light-cone family $\sigma\in\{+,-\}$. Suppose that the cluster contains
$m$ external points, indexed by a set $A\subset\{1,\ldots,k\}$, and $n$
screening points from $\mathscr S_\sigma$. By the definition above, the exponent
contributed by the pairwise factors internal to this cluster is $\Theta(A,n)$.
If $m=0$, then $\Theta(A,n)=b^2n(n-1)\ge0$, so this cluster causes no
singularity. Assume now that $m\ge1$. Set
\[
  S:=\sum_{i\in A}\alpha_i,\qquad
  Q:=\sum_{i\in A}\alpha_i^2.
\]
Then
\[
  \Theta(A,n)
  =
  S^2-Q+2bnS+b^2n(n-1),
\]
and therefore
\begin{align*}
  \Theta(A,n)+m+n-1
  &=
  S^2-Q+2bnS+b^2n(n-1)+m+n-1  \\
  &=
  (S+bn)^2-Q-b^2n+m+n-1.
\end{align*}
Since $|\alpha_i|<\frac{1}{\sqrt2}$ and $b<\frac{1}{\sqrt2}$, we have
\[
  Q<\frac m2,\qquad b^2<\frac12.
\]
Thus
\[
  \Theta(A,n)+m+n-1
  >
  (S+bn)^2+\frac{m+n}{2}-1.
\]
If $m+n\ge3$, the right-hand side is strictly positive. The remaining cases
with $m+n=2$ are immediate:
\[
  2\alpha_i\alpha_j+1>0
\]
for two external points, because $|\alpha_i|,|\alpha_j|<\frac{1}{\sqrt2}$;
\[
  2b\alpha_i+1>0
\]
for one external point and one screening point, because
$\alpha_i>-\frac{1}{\sqrt2}$ and $b<\frac{1}{\sqrt2}$; and
\[
  2b^2+1>0
\]
for two screening points. Hence, for every one-family cluster containing
$m+n\ge2$ points,
\[
  \Theta(A,n)>-(m+n-1).
\]
This is precisely the one-family cluster inequality used above. Summing it
over the nontrivial components of $\Gamma_\sigma(F)$ for both signs gives
\[
  E(F)>-R(F)\ge-\operatorname{codim}F.
\]
These are exactly the inequalities required by Lemma~\ref{lem:1d-power-counting}. This proves that 
$G_+G_-$ is locally integrable.

We have proved that, after the large horizontal variables are integrated out,
the compact integral over the external variables, middle screening variables,
and small horizontal endpoint variables is finite. Since there are only
finitely many choices of the large/small decomposition and finitely many
ordered screening sectors, the contribution from the standard ordered sector
$D$ is finite.

Now fix an arbitrary $\sigma\in S_k$. On $D_\sigma$, the contour in
\eqref{eq:clorentz} still enters horizontally at $\i t_1$, visits
$\i t_2,\ldots,\i t_k$ in label order, and exits horizontally from $\i t_k$.
The bounded part of the contour is now a finite union of oriented vertical
segments, some of which may be traversed downward and some of which may overlap
as subsets of the imaginary axis. We decompose the ordered screening region
according to the contour segment occupied by each screening variable and the
order along that oriented segment. All middle screening times still lie in the
fixed compact interval $(T_-,T_+)$.

The horizontal-tail estimate is unchanged: for a left-tail variable the real
part of every mixed external-horizontal argument is $r_l$, and for a
right-tail variable it is $r'_l$, independent of the relative order of the
external times. The same linear computation therefore gives the same
exponential decay at infinity. After taking absolute values, reversing the
orientation of a bounded vertical segment or reversing a contour-ordered
difference only changes the sign of the corresponding light-cone difference,
and this has no effect because $\delta(\theta)=\delta(-\theta)$. Thus every sector over $D_\sigma$ is dominated by the same type of compact
majorant $G_+G_-$, with the same power-counting verification. Since there are
only finitely many permutations, contour-segment decompositions, and
large/small horizontal decompositions, the integral over every $D_\sigma$ is
finite. This proves equation \eqref{tag51}, and hence Theorem~\ref{thm:smear}.

\subsection{Hermiticity of the algebraic form}\label{sec:hermitproof}
We now prove Theorem~\ref{thm:ma-hermitian}.  The first step is the following.
\begin{lmm}\label{lem:invo}
For all $A,B\in \ma$, $(AB)^* = B^*A^*$ and $(A^*)^* = A$.
\end{lmm}
\begin{proof}
On a generator, $(\mo_{f,n}^*)^*=\mo_{\bar{\bar f},n}=\mo_{f,n}$.
Since the involution was extended conjugate-linearly and anti-multiplicatively
to the tensor algebra, the identities $(AB)^*=B^*A^*$ and $(A^*)^*=A$ follow
for all $A,B\in\ma$.
\end{proof}

The next lemma is the key ingredient in the proof of Hermiticity.
\begin{lmm}\label{lem:herm}
For any $A\in \ma$, $\omega(A^*) = \overline{\omega(A)}$.
\end{lmm}
\begin{proof}
By linearity, it is enough to consider a monomial
\[
  A=\mo_{f_1,n_1}\cdots\mo_{f_k,n_k}.
\]
Set $\boldsymbol n=(n_1,\ldots,n_k)$, write
$\boldsymbol n^\#=(n_k,\ldots,n_1)$, set $\alpha_j:=bn_j$, and
$w:=-\sum_{j=1}^k n_j$. If $w<0$, then both $\omega(A)$ and
$\omega(A^*)$ vanish,
because reversing the labels leaves $w$ unchanged. Thus assume $w\ge0$. We
first claim that whenever $(u_1,\ldots,u_k)$ satisfies the non-light-cone
condition \eqref{eq:lightcond},
\begin{align}\label{eq:conjclaim}
  C_{b\boldsymbol n}^{\mathrm L}(u_1,\ldots,u_k)
  =
  \overline{C_{b\boldsymbol n^\#}^{\mathrm L}(u_k,\ldots,u_1)}.
\end{align}
To prove this, fix $x_1,\ldots,x_k$, and for  $(t_1,\ldots,t_k)\in \Omega_k$, define
\[
  F(t_1,\ldots,t_k)
  :=
  C_{\alpha_1,\ldots,\alpha_k}(t_1,x_1;\ldots;t_k,x_k),
\]
and
\[
  G(t_1,\ldots,t_k)
  :=
  \overline{C_{\bar\alpha_k,\ldots,\bar\alpha_1}(-\bar {t_k},x_k;\ldots; -\bar {t_1},x_1)}.
\]
If $(t_1,\ldots,t_k)\in\Omega_k$, then
$(-\bar t_k,\ldots,-\bar t_1)\in\Omega_k$ as well, and so $G$ is well-defined. Moreover,~if 
\[
\sum_{n_1,\ldots,n_k\in \Z_{\ge0}} a_{n_1,\ldots,n_k} (\tau_1-t_1^0)^{n_1}\cdots(\tau_k - t_k^0)^{n_k}
\]
is the power series expansion of the analytic function 
\[
(\tau_1,\ldots,\tau_k)\mapsto  C_{\bar\alpha_k,\ldots,\bar\alpha_1}(-\tau_k,x_k;\ldots; -\tau_1,x_1) 
\]
in a neighborhood of a point $(t_1^0,\ldots,t_k^0)\in \Omega_k$, then for any $(t_1,\ldots,t_k)$ in the corresponding neighborhood of $(\bar{t_1^0},\ldots,\bar{t_k^0})$, 
\begin{align*}
G(t_1,\ldots,t_k)
  &=
  \overline{C_{\bar\alpha_k,\ldots,\bar\alpha_1}(-\bar {t_k},x_k;\ldots; -\bar {t_1},x_1)}\\
  &=\overline{\sum_{n_1,\ldots,n_k\in \Z_{\ge0}} a_{n_1,\ldots,n_k} (\bar{t_1}-t_1^0)^{n_1}\cdots(\bar{t_k} - t_k^0)^{n_k}}\\
  &= \sum_{n_1,\ldots,n_k\in \Z_{\ge0}} \bar{a_{n_1,\ldots,n_k}} (t_1-\bar{t_1^0})^{n_1}\cdots(t_k - \bar{t_k^0})^{n_k}.
\end{align*}
This shows that $G$ is analytic on
$\Omega_k$. Next, equation~\eqref{eq:ceuc}
and Theorem~\ref{thm:euc} show that complex conjugation of the Euclidean correlator  simply replaces each
charge by its conjugate, because all coefficients and the kernel $g$ are real
on real arguments. Therefore, for real numbers $r_1<\cdots<r_k$, 
\[
  G(r_1,\ldots,r_k)
  =
  C_{\alpha_k,\ldots,\alpha_1}(-r_k,x_k;\ldots; -r_1,x_1).
\]
Because the Euclidean correlator is invariant under simultaneous permutation of
the labeled triples $(\alpha_j,t_j,x_j)$, we get
\[
  G(r_1,\ldots,r_k)
  =
  C_{\alpha_1,\ldots,\alpha_k}(-r_1,x_1;\ldots; -r_k,x_k).
\]
Changing variables $a_l\mapsto -a_l$ in equation~\eqref{eq:ceuc} and using
the real-axis symmetry $g(t,x)=g(-t,x)$, we get
\[
  C_{\alpha_1,\ldots,\alpha_k}(-r_1,x_1;\ldots; -r_k,x_k)
  =
  C_{\alpha_1,\ldots,\alpha_k}(r_1,x_1;\ldots; r_k,x_k)
  =
  F(r_1,\ldots,r_k).
\]
Thus $F$ and $G$ are analytic on $\Omega_k$ and coincide there on every
strictly increasing real $k$-tuple. Therefore, by the uniqueness statement in
Theorem~\ref{thm:main2}, they agree on all of $\Omega_k$. In other words, 
\[
  F(t_1,\ldots,t_k)
  =
  \overline{C_{\bar\alpha_k,\ldots,\bar\alpha_1}(-\bar {t_k},x_k;\ldots; -\bar {t_1},x_1)}.
\]
Now let $(s_1,\ldots,s_k)\in\R^k$ satisfy condition \eqref{eq:lightcond} for our chosen $x_1,\ldots,x_k$, and choose
$t^{(m)}\in\Omega_k$ with $t^{(m)}\to(\i s_1,\ldots,\i s_k)$. Then
$(-\overline{t_k^{(m)}},\ldots,-\overline{t_1^{(m)}})\in\Omega_k$ and
converges to $(\i s_k,\ldots,\i s_1)$. Applying the Lorentzian boundary value
statement from Theorem~\ref{thm:main2} and letting $m\to\infty$ proves the
claim \eqref{eq:conjclaim}. Since the light cone boundary set has measure zero, this identity holds
almost everywhere on $M^k$. Since
\[
  A^*=\mo_{\bar f_k,n_k}\cdots\mo_{\bar f_1,n_1},
\]
we have
\begin{align*}
\omega(A^*)
&=\int_{M^k}\prod_{j=1}^k\overline{f_{k+1-j}(u_j)}\,
C_{b\boldsymbol n^\#}^{\mathrm L}(u_1,\ldots,u_k)\,\dd u_1\cdots\dd u_k\\
&=\overline{\int_{M^k}\prod_{j=1}^k f_{k+1-j}(u_j)\,
\overline{C_{b\boldsymbol n^\#}^{\mathrm L}(u_1,\ldots,u_k)}\,\dd u_1\cdots\dd u_k}\\
&=\overline{\int_{M^k}\prod_{j=1}^k f_{k+1-j}(u_j)\,
C_{b\boldsymbol n}^{\mathrm L}(u_k,\ldots,u_1)\,\dd u_1\cdots\dd u_k}\\
&=\overline{\int_{M^k}\prod_{j=1}^k f_j(u_j)\,
C_{b\boldsymbol n}^{\mathrm L}(u_1,\ldots,u_k)\,\dd u_1\cdots\dd u_k}
=\overline{\omega(A)}.
\end{align*}
This completes the proof of the lemma.
\end{proof}

We are now ready to complete the proof of Theorem \ref{thm:ma-hermitian}. 
\begin{proof}[Proof of Theorem~\ref{thm:ma-hermitian}]
By Lemma \ref{lem:invo} and Lemma \ref{lem:herm}, 
\[
\inn{B,A}_\ma = \omega(B^*A) = \omega((A^*B)^*) = \bar{\omega(A^*B)} = \bar{\inn{A,B}_\ma}.
\]
This proves the claim.
\end{proof}

\subsection{Construction of the state space}\label{sec:innlmmproof}
In this section, we prove Theorem~\ref{thm:innlmm}, Proposition~\ref{prop:mh0-indefinite},  Theorem~\ref{thm:obsrep}, and Lemma \ref{lem:5x}. 

\begin{proof}[Proof of Theorem \ref{thm:innlmm}] 
By definition,
\[
  \mh=\overline{\mh_0},
\]
so $\mh_0$ is dense in $\mh$. For each $A\in\ma$, let $p_A:\mh\to\C$
denote the coordinate projection
\[
  p_A((z_B)_{B\in\ma}) := z_A.
\]
By the definition of the product topology, each $p_A$ is continuous. We first prove that the pairing is well-defined. Suppose $\Phi(A_1)=\Phi(A_2)$. Then, by definition of $\Phi$,
\[
  \inn{C^*A_1}=\inn{C^*A_2} \qquad\text{for all } C\in\ma.
\]
Hence, for every $C\in\ma$,
\begin{align*}
  p_{A_1}(\Phi(C))
  &=\inn{A_1^*C} = \bar{\inn{C^*A_1}} = \bar{\inn{C^*A_2}} \\
  &=\inn{A_2^*C} = p_{A_2}(\Phi(C)),
\end{align*}
where we used Theorem~\ref{thm:ma-hermitian}. Therefore
$p_{A_1}-p_{A_2}$ vanishes on the dense subset $\Phi(\ma)\subseteq\mh$.
Since $p_{A_1}-p_{A_2}$ is continuous, it vanishes on all of
$\mh=\overline{\Phi(\ma)}$. Therefore $z_{A_1}=z_{A_2}$ for all
$z\in\mh$, proving that
\[
  (\Phi(A),z):=z_A
\]
is well-defined.

Next, fix $z\in\mh$. Choose a net $(C_\nu)$ in $\ma$ such that $\Phi(C_\nu)\to z$ in the product topology; see~\secref{sec:app-tvs} for background on nets and their relation to closure. For any $A,B\in\ma$ and $\lambda,\mu\in\C$,
\begin{align*}
z_{\lambda A+\mu B}
&= \lim_\nu \Phi(C_\nu)_{\lambda A+\mu B}
= \lim_\nu \inn{(\lambda A+\mu B)^* C_\nu} \\
&= \bar{\lambda}\lim_\nu \inn{A^*C_\nu}+\bar{\mu}\lim_\nu \inn{B^*C_\nu}
= \bar{\lambda} z_A+\bar{\mu} z_B.
\end{align*}
So $A\mapsto z_A$ is conjugate linear on $\ma$. Also, $\Phi$ is  linear, since
\[
  \Phi(\lambda A+\mu B)=\lambda\Phi(A)+\mu\Phi(B).
\]
Consequently, for $u_j=\Phi(A_j)\in\mh_0$,
\[
  (\lambda u_1+\mu u_2,z)
  =(\Phi(\lambda A_1+\mu A_2),z)
  =z_{\lambda A_1+\mu A_2}
  =\bar\lambda z_{A_1}+\bar \mu z_{A_2}
  =\bar\lambda(u_1,z)+\bar \mu(u_2,z).
\]
Hence the pairing is conjugate linear in the first argument.

For linearity in the second argument, fix $u=\Phi(A)\in\mh_0$ and
take any $z_1,z_2\in\mh$ and $\lambda,\mu\in\C$. Then
\begin{align*}
(u,\lambda z_1+\mu z_2)
&=(\lambda z_1+\mu z_2)_A\\
&=\lambda (z_1)_A+\mu (z_2)_A=\lambda(u,z_1)+\mu(u,z_2).
\end{align*}
So $z\mapsto (u,z)$ is linear, and the pairing is sesquilinear.

Now let $u=\Phi(A)$ and $v=\Phi(B)$ be in $\mh_0$. Then
\[
  (u,v)
  =\Phi(B)_A
  =\omega(A^*B)
  =\langle A,B\rangle_{\ma}.
\]
By Theorem~\ref{thm:ma-hermitian}, this shows that the restriction of the pairing to $\mh_0\times\mh_0$ is Hermitian.

Next, suppose that $(u,z)=0$ for all $u\in\mh_0$. Taking $u=\Phi(A)$ gives
\[
  0=(\Phi(A),z)=z_A \qquad\text{for all } A\in\ma.
\]
Hence every coordinate of $z=(z_A)_{A\in\ma}$ is zero, so $z=0$ in $\C^{\ma}$ (and therefore in $\mh$). The converse implication is immediate.
In particular, if $v\in\mh_0$ and $(u,v)=0$ for all $u\in\mh_0$, then $v=0$.
Thus the Hermitian form on $\mh_0$ is nondegenerate. 
\end{proof}

\begin{proof}[Proof of Proposition~\ref{prop:mh0-indefinite}]
Let
\[
  u_+:=\Omega=\Phi(1).
\]
Then, by Theorem~\ref{thm:innlmm},
\[
  (u_+,u_+)=(\Phi(1),\Phi(1))=\omega(1)=1>0.
\]
To construct a negative vector, fix a nonnegative function
$\rho\in C_c^\infty((-1,1)^2,\R)$ with
\[
  \int_{\R\times[-1,1]} \rho(t,x)\,\dd t\,\dd x=1.
\]
For $0<\ep<\frac12$, define
\[
  f_\ep(t,x):=\ep^{-2}\rho\biggl(\frac{t}{\ep},\frac{x}{\ep}\biggr),
\]
viewed as a smooth compactly supported function on the cylinder $M$. Then
$f_\ep\ge0$, $f_\ep\in C_c^\infty(M,\R)$,
\[
  \int_M f_\ep(u)\,\dd u=1,
\]
and the support of $f_\ep$ shrinks to $(0,0)$ as $\ep\downarrow0$.
Set
\[
  A_\ep:=\mo_{f_\ep,-b}\in\ma,
  \qquad
  u_\ep:=\Phi(A_\ep)-\omega(A_\ep)\Omega\in\mh_0.
\]
We first compute the one-point function. Taking $k=1$ and $\alpha_1=-b$ in
Theorem~\ref{thm:euc} (which is allowed since $-b > -\frac{1}{2b}$, because $b <\frac{1}{\sqrt{2}}$), we obtain a Euclidean one-point correlator that is
independent of the insertion point and equal to
\[
  c_b
  :=
  -\mu\int_M e^{4b^2 g(u)}\,\dd u.
\]
The integral is finite by Theorem~\ref{thm:euc}, and the integrand is strictly positive. Thus,  $c_b<0$. Since a constant analytic continuation remains the same constant,
Theorem~\ref{thm:main} gives the same value for the Lorentzian one-point
function. Therefore
\[
  \omega(A_\ep)=c_b\int_M f_\ep(u)\,\dd u=c_b<0
\]
for every $\ep$.

Next, we prove the integrated two-point estimate
\begin{align}\label{eq:5x}
\lim_{\ep \to 0} \omega(A_\epsilon^*A_\epsilon)= 0.
\end{align}
Note that $A_\ep^* = A_\ep$, and thus,
\begin{align*}
\omega(A_\ep^* A_\ep) &= \int_{M^2} f_\ep(u)f_\ep(v) C_{-b,-b}^{\mathrm{L}}(u,v) \, \dd u\, \dd v.
\end{align*}
Since the right side is a weighted average of $C_{-b,-b}^{\mathrm{L}}(u,v)$ in a neighborhood of radius $O(\ep)$ around zero in $M^2$, it suffices to show that  
\begin{align}\label{eq:5x0}
\lim_{u,v\to 0} C_{-b,-b}^{\mathrm{L}}(u,v) = 0.
\end{align}
Now, with $u=(t,x)$ and $v= (s,y)$, we have
\begin{align*}
C_{-b,-b}^{\mathrm{L}}(u,v) &= \frac{\mu^2}{2} e^{-4b^2 g(\i (t-s), x-y)}I(u,v),
\end{align*}
where $I(u,v)$ is the integral appearing in the Lorentzian correlator formula \eqref{eq:clorentz}. By equation \eqref{eq:comp1}, 
\begin{align*}
|e^{-4b^2 g(\i (t-s), x-y)}| &\le |1-e^{-\pi \i (t-s+x-y)}|^{2b^2} |1-e^{-\pi \i (t-s-x+y)}|^{2b^2}. 
\end{align*}
Clearly, the upper bound tends to zero as $u,v\to 0$. Thus, to prove the claim \eqref{eq:5x0}, it suffices to show that $I(u,v)$ is bounded in a neighborhood of zero.

Let $\gamma$ be the contour for the Lorentzian correlator in this problem, so that 
\begin{align*}
  I(u,v) &=
\int_{\gamma}
\exp\biggl(
4b^2\sum_{l=1}^2 g(\mathcal T(u-z_l))
\notag\\
&\qquad\qquad
+4b^2\sum_{\ell=1}^2 g(\mathcal T(v-z_l))
-4b^2 g(\mathcal T(z_1-z_2))
\biggr)
\,\dd z_1\,\dd z_2.
\end{align*}
Applying the Cauchy--Schwarz inequality, we get 
\[
  |I(u,v)|\le \sqrt{J(u)J(v)},
\]
where
\begin{align*}
J(u) &:= \int_{\gamma}
\exp\biggl(
8b^2\sum_{l=1}^2 \Re(g(\mathcal T(u-z_l)))-4b^2 \Re(g(\mathcal T(z_1-z_2)))
\biggr)
\,|\dd z_1|\,|\dd z_2|, 
\end{align*}
and $J(v)$ is defined similarly, replacing $u$ by $v$. We will now show that $J(u)$ is uniformly bounded when $u$ and $v$ lie in a fixed compact neighborhood of zero. (Note that $J(u)$ is not just a function of $u$; it also depends on $v$ through the definition of the contour $\gamma$.) The uniform boundedness of $J(v)$ will follow similarly. 

Suppose that $t\le s$. Then the  contour $\gamma$ can be written as $\Gamma_-\cup \Gamma_0\cup \Gamma_+$, where 
\[
\Gamma_-:=\{\i t-r:r>0\},\quad
\Gamma_0:=\{\i \tau:t<\tau<s\},\quad
\Gamma_+:=\{\i s+r:r>0\}.
\]
(If $t>s$, we just have to redefine $\Gamma_0 := \{\i \tau:t>\tau>s\}$. The proof will then proceed exactly the same as follows.) 
The integral defining $J(u)$ splits up into parts depending on which of the above segments contain $z_1$ and which contains $z_2$. For example, if $J_{++}(u)$ denotes the part where $z_1$ and $z_2$ are both in $\Gamma_+$, then 
\begin{align*}
J_{++}(u) &= 2\int_{0< r_1< r_2<\infty}\int_{a_1, a_2\in [-1,1]}\exp\biggl(
8b^2\sum_{l=1}^2 \Re(g(\i (s - t) + r_l, a_l-x))\\ 
&\qquad -4b^2 \Re(g(r_2-r_1, a_2 - a_1))
\biggr)
\,\dd a_1 \,\dd a_2\, \dd r_2\, \dd r_1. 
\end{align*}
From the formula \eqref{eq:gdefine0} for $g$, it follows that the integrand is bounded above by a constant times 
\begin{align*}
e^{-6\pi b^2r_1-2\pi b^2r_2}\prod_{l=1}^2 \prod_{\sigma\in \{-1,1\}}|1-e^{-\pi (\i (s - t) + r_l+ \i \sigma(a_l-x))}|^{-4b^2}.
\end{align*}
Now, if $r,\theta \in \R$, then
\begin{align}\label{eq:rtheta}
  |1-re^{\i \theta}|^2 = 1-2r\cos \theta + r^2 \ge 1-2r+r^2=(1-r)^2.
\end{align}
This shows that the quantity from the previous display is bounded above by
\begin{align*}
e^{-6\pi b^2r_1-2\pi b^2r_2}\prod_{l=1}^2(1-e^{-\pi r_l})^{-8b^2}.
\end{align*}
Note that this bound has no dependence on $u$ or $v$. Moreover, since $8b^2 < 1$, its integral over $0<r_1<r_2<\infty$ is finite. Thus, $J_{++}(u)$ is bounded above by a quantity that does not depend on $u,v$.

Next, let $J_{0+}(u)$ denote the part where $z_1\in \Gamma_0$ and $z_2\in \Gamma_+$. Then
\begin{align*}
J_{0+}(u)
&= 2\int_t^s\int_0^\infty\int_{a_1,a_2\in[-1,1]}
\exp\biggl(
8b^2\Re(g(\i (\tau-t),a_1-x))
\notag\\
&\qquad\qquad\qquad
+8b^2\Re(g(\i (s-t)+r,a_2-x))\\ 
&\qquad \qquad \qquad -4b^2\Re(g(r+\i(s-\tau),a_2-a_1))
\biggr)
\,\dd a_1\,\dd a_2\,\dd r\,\dd \tau.
\end{align*}
By equations \eqref{eq:gdefine0}, \eqref{eq:gimag}, and \eqref{eq:rtheta}, the integrand is bounded above by a constant times
\begin{align*}
&e^{-2\pi b^2r}(1-e^{-\pi r})^{-8b^2}\prod_{\sigma\in\{-1,1\}}|1-e^{-\pi\i((\tau-t)+\sigma(a_1-x))}|^{-4b^2}.
\end{align*}
The integral of the above with respect to $(r,\tau,a_1)$ factorizes into the product of an integral over $r$ and an integral over $(\tau, a_1)$. The integral over $r$ is finite and has no dependence on $(u,v)$. For the integral over $(\tau, a_1)$, a simple change of variable gives
\begin{align*}
&\int_{t}^s \int_{-1}^1 \prod_{\sigma\in\{-1,1\}}|1-e^{-\pi\i((\tau-t)+\sigma(a_1-x))}|^{-4b^2} \, \dd a_1 \, \dd \tau\\ 
&= \int_0^{s-t} \int_{-1}^1 \prod_{\sigma\in\{-1,1\}}|1-e^{-\pi\i(\xi+\sigma c)}|^{-4b^2} \, \dd c \, \dd \xi\\ 
&\le \int_0^{1} \int_{-1}^1 \prod_{\sigma\in\{-1,1\}}|1-e^{-\pi\i(\xi+\sigma c)}|^{-4b^2} \, \dd c \, \dd \xi,
\end{align*}
assuming that $s-t\le 1$. Again, since $8b^2 <1$, the above integral is finite; and it has no dependence on $(u,v)$. This shows that $J_{0+}(u)$ is uniformly bounded in the region where $0\le s-t\le 1$.

Next, let $J_{00}(u)$ denote the part where $z_1,z_2\in \Gamma_0$. Then
\begin{align*}
J_{00}(u)
&= 2\int_{t<\tau_1<\tau_2<s}\int_{a_1,a_2\in[-1,1]}
\exp\biggl(
8b^2\sum_{l=1}^2 \Re(g(\i (\tau_l-t),a_l-x))
\notag\\
&\qquad\qquad\qquad
-4b^2\Re(g(\i(\tau_2-\tau_1),a_2-a_1))
\biggr)
\,\dd a_1\,\dd a_2\,\dd \tau_2\,\dd \tau_1.
\end{align*}
By equation \eqref{eq:gimag}, the interaction factor is bounded above by a constant, because each term
\[
  |1-e^{-\pi\i((\tau_2-\tau_1)\pm(a_2-a_1))}|^{2b^2}
\]
is at most $2^{2b^2}$. Therefore the integrand is bounded by a constant times 
\[
  \prod_{l=1}^2\prod_{\sigma\in\{-1,1\}} |1 - e^{-\pi \i((\tau_l-t)+\sigma(a_l-x))}|^{-4b^2}.
\]
Dropping the ordering constraint $\tau_1<\tau_2$ and using the same change of variables as above for each pair $(\tau_l,a_l)$, we get the uniform boundedness of $J_{00}(u)$.

The remaining pieces are estimated similarly. The term $J_{--}(u)$ is the same as $J_{++}(u)$ by symmetry, $J_{-0}(u)$ is treated exactly like $J_{0+}(u)$, and $J_{-+}(u)$ is even easier because it has exponential decay in both tail variables. Therefore every piece in the decomposition of $J(u)$ is uniformly bounded on a compact neighborhood of zero, and hence so is $J(u)$ itself. The same argument applies to $J(v)$. Consequently $I(u,v)$ is bounded in a neighborhood of zero. Combined with the vanishing of the prefactor $e^{-4b^2g(\i(t-s),x-y)}$, this proves \eqref{eq:5x0}, and hence also \eqref{eq:5x}.

Now note that since \(A_\epsilon^*=A_\epsilon\), Theorem \ref{thm:innlmm} gives
\[
(u_\epsilon,u_\epsilon)
=
\omega((A_\epsilon-\omega(A_\epsilon)1)^*
(A_\epsilon-\omega(A_\epsilon)1)).
\]
Using \(\omega(A_\epsilon)=c_b\in\mathbb R\), this becomes
\[
(u_\epsilon,u_\epsilon)
=
\omega(A_\epsilon^*A_\epsilon)-|c_b|^2 .
\]
By the estimate just proved, \(\omega(A_\epsilon^*A_\epsilon)\to 0\). Therefore, for all
sufficiently small \(\epsilon\),
\[
(u_\epsilon,u_\epsilon)<0.
\]
Taking \(u_-:=u_\epsilon\) for such an \(\epsilon\) completes the proof.
\end{proof}

\begin{proof}[Proof of Theorem~\ref{thm:obsrep}]
It is easy to see that $\pi(C)$ is linear and continuous, as a map from $\mh_0$ into $\C^\ma$ (since we have not yet shown that it maps into $\mh$). If $A,B,C\in\ma$, then
\[
  (\pi(C)\Phi(A))_B
  =\Phi(A)_{C^*B}
  =\inn{(C^*B)^*A}
  =\inn{B^*CA}
  =\Phi(CA)_B.
\]
Thus, 
\[
  \pi(C)\Phi(A)=\Phi(CA)\in\mh_0.
\]
In particular, this shows that $\pi(C)(\mh_0)\subseteq\mh_0$. Continuity implies
\[
  \pi(C)(\mh)=\pi(C)(\overline{\mh_0})\subseteq\overline{\mh_0}=\mh.
\]
Thus $\pi(C)$ is a well-defined continuous linear map
on $\mh$, and item~1 follows. Next, for $C,D\in\ma$ and $z\in\mh$, we have
\[
  (\pi(C+D)z)_A=z_{(C+D)^*A}=z_{C^*A}+z_{D^*A}=((\pi(C)+\pi(D))z)_A,
\]
and similarly $\pi(\lambda C)=\lambda\pi(C)$ for every $\lambda\in\C$. Also,
\[
  (\pi(C)\pi(D)z)_A=(\pi(D)z)_{C^*A}=z_{D^*C^*A}=z_{(CD)^*A}=(\pi(CD)z)_A,
\]
while clearly $(\pi(1)z)_A=z_A$. Therefore $C\mapsto\pi(C)$ is a unital
algebra representation.

Now let $u=\Phi(A)\in\mh_0$ and $z\in\mh$. Then by item 1, 
\[
  (\pi(C)u,z)
  =(\Phi(CA),z)
  =z_{CA}
  =(\pi(C^*)z)_A
  =(u,\pi(C^*)z).
\]
This proves item~2. Next, set $\Omega:=\Phi(1)$. Then for every $C\in\ma$, item 1 gives 
\[
  \pi(C)\Omega=\Phi(C).
\]
Hence
\[
  \operatorname{span}\{\pi(C)\Omega:C\in\ma\}=\Phi(\ma)=\mh_0.
\]
Thus $\Omega$ is cyclic, proving item~3. Finally,
\[
  (\Omega,\pi(A)\Omega)
  =(\Phi(1),\Phi(A))
  =\langle 1,A\rangle_{\ma}
  =\omega(A), 
\]
which proves item~4.
\end{proof}

\begin{proof}[Proof of Lemma \ref{lem:5x}]
Uniqueness of the adjoint follows from Theorem \ref{thm:innlmm}: if \(S_1,S_2\) both satisfy
\((Tu,z)=(u,S_i z)\) for all \(u\in \mh_0,z\in \mh\), then
\((u,(S_1-S_2)z)=0\) for all \(u\in \mh_0\), hence \((S_1-S_2)z=0\) for every \(z\in \mh\).

It is immediate from the definition that \((T^\dagger)^\dagger=T\), and that the identity
operator belongs to \(\operatorname{End}_{\dagger}(\mh)\). If \(T,R\in
\operatorname{End}_{\dagger}(\mh)\), then \(TR\) preserves \(\mh_0\), and for
\(u\in \mh_0,z\in \mh\),
\[
(TRu,z)=(Ru,T^\dagger z)=(u,R^\dagger T^\dagger z).
\]
Similarly,
\[
(R^\dagger T^\dagger u,z)=(u,TRz).
\]
Thus \(TR\in\operatorname{End}_{\dagger}(\mh)\) and
\[
(TR)^\dagger=R^\dagger T^\dagger .
\]
Linearity is checked in the same way. Hence \(\operatorname{End}_{\dagger}(\mh)\) is a
unital \(*\)-algebra, with \(*\) denoted by \(\dagger\).

The definition of \(\mathcal T_\dagger\) makes \(T\mapsto T^\dagger\) continuous.
Composition with a fixed continuous linear map is continuous for \(\mathcal T_b\): on the
right because continuous linear maps send bounded sets to bounded sets, and on the left
because the pullback of a continuous seminorm by a continuous linear map is again a
continuous seminorm. Applying this both to operators and to their adjoints proves separate
continuity of multiplication for \(\mathcal T_\dagger\).

Finally, Theorem \ref{thm:obsrep} gives
\[
(\pi(C)u,z)=(u,\pi(C^*)z)
\]
for all \(u\in \mh_0,z\in \mh\). Applying the same identity with \(C^*\) in place of \(C\)
gives
\[
(\pi(C^*)u,z)=(u,\pi(C)z).
\]
Also \(\pi(C)\) and \(\pi(C^*)\) preserve \(\mh_0\). Hence
\(\pi(C)\in\operatorname{End}_{\dagger}(\mh)\) and
\(\pi(C)^\dagger=\pi(C^*)\).
\end{proof}

\subsection{Cylinder translations and represented covariance}\label{sec:timeev-details}
We now prove Proposition~\ref{prop:aqft-net} and
Theorem~\ref{thm:timeev-main-properties}.

\begin{proof}[Proof of Proposition~\ref{prop:aqft-net}]
Cylinder translations preserve smooth compact support, and
\[
  \eta_{a,\theta}\bar f=\overline{\eta_{a,\theta}f},
  \qquad f\in\mathcal F.
\]
Since $\eta_{a,\theta}$ was extended multiplicatively from the generators,
each $\eta_{a,\theta}$ is a $*$-automorphism of $\ma$, and the pullback
definition immediately gives
\[
  \eta_{a,\theta}\eta_{a',\theta'}=\eta_{a+a',\,\theta+\theta'},
  \qquad
  \eta_{0,0}=\mathrm{Id}_{\ma},
\]
with the second coordinate understood modulo $2$.

To prove invariance of the point correlator, fix
$\boldsymbol n=(n_1,\ldots,n_k)\in\mathcal I_b^k$. If
$-\sum_{j=1}^k n_j<0$, then $C_{b\boldsymbol n}^{\mathrm L}\equiv0$ and
invariance is immediate. Otherwise, let $u_j=(t_j,x_j)$ and write
$v_j=(\i t_j,x_j)$. Let $\gamma$ be the contour used in the definition
\eqref{eq:clorentz} of $C_{b\boldsymbol n}^{\mathrm L}(u_1,\ldots,u_k)$, and let
\[
  \widetilde v_j:=(\i(t_j+a),x_j+\theta),
  \qquad
  \widetilde\gamma:=\gamma+\i a.
\]
Then $\widetilde\gamma$ is exactly the contour appearing in the definition of
$C_{b\boldsymbol n}^{\mathrm L}(T_{a,\theta}u_1,\ldots,T_{a,\theta}u_k)$. If
$\mathcal T_\gamma$ and $\mathcal T_{\widetilde\gamma}$ denote the corresponding
time-ordering operations, then shifting every screening variable by
$(\i a,\theta)$ preserves the order along the contour and leaves all
time-ordered differences unchanged:
\[
\begin{aligned}
  \mathcal T_{\widetilde\gamma}(\widetilde v_j-\widetilde v_{j'})
  &=\mathcal T_\gamma(v_j-v_{j'}),\\
  \mathcal T_{\widetilde\gamma}(\widetilde v_j-\widetilde u_l)
  &=\mathcal T_\gamma(v_j-u_l),\\
  \mathcal T_{\widetilde\gamma}(\widetilde u_l-\widetilde u_{l'})
  &=\mathcal T_\gamma(u_l-u_{l'}),
\end{aligned}
\]
where $u_l=(\gamma(s_l),b_l)$ and
$\widetilde u_l=(\gamma(s_l)+\i a,b_l+\theta)$. Since every $g$-factor in
equation \eqref{eq:clorentz} depends only on these differences, the integrand is
unchanged. Also, $\widetilde\gamma'(s)=\gamma'(s)$, and the spatial integral is
unchanged by translation modulo the $2$-periodicity in the second coordinate.
Therefore the change of variables $u_l\mapsto \widetilde u_l$ in
equation \eqref{eq:clorentz} gives
\[
C_{b\boldsymbol n}^{\mathrm{L}}(T_{a,\theta}u_1,\ldots,T_{a,\theta}u_k)
=
C_{b\boldsymbol n}^{\mathrm{L}}(u_1,\ldots,u_k)
\]
for every admissible integer charge tuple \(\boldsymbol n\) and every configuration in the
non-light-cone domain.

If $C=\mo_{f_1,n_1}\cdots\mo_{f_k,n_k}$, then changing variables
$u_j\mapsto T_{a,\theta}u_j$ in the defining integral gives
\[
  \omega(\eta_{a,\theta}C)=\omega(C).
\]
By linearity this holds for all $C\in\ma$. If $O_1\subset O_2$, then every
generator of $\ma(O_1)$ is also a generator of $\ma(O_2)$, so
$\ma(O_1)\subset \ma(O_2)$. Moreover, if $\supp(f)\subset O$, then
$\supp(\eta_{a,\theta}f)\subset T_{a,\theta}O$, hence
$\eta_{a,\theta}\ma(O)\subset \ma(T_{a,\theta}O)$. Applying this inclusion to
$(-a,-\theta)$ and then composing with $\eta_{a,\theta}$ gives the reverse
inclusion, so equality holds. Finally, $\omega_O(1)=1$, and for every
$A\in\ma(O)$,
\[
  \omega_O(A^*)=\omega(A^*)=\overline{\omega(A)}=\overline{\omega_O(A)},
\]
where the middle identity follows from Theorem~\ref{thm:ma-hermitian} with
$B=1$. Thus $\omega_O$ is a normalized Hermitian functional on $\ma(O)$, and
Proposition~\ref{prop:aqft-net} follows.
\end{proof}

\begin{proof}[Proof of Theorem~\ref{thm:timeev-main-properties}]
A priori, it is easy to see that $U_{\mathrm{cyl}}(a,\theta)$ is a linear and
continuous map from $\mh$ into $\C^\ma$, although we do not know that it maps into $\mh$. For
$A,B\in\ma$,
\begin{align*}
  (U_{\mathrm{cyl}}(a,\theta)\Phi(B))_A
  &=\Phi(B)_{\eta_{-a,-\theta}A}
   =\omega((\eta_{-a,-\theta}A)^*B) \\
  &=\omega(\eta_{-a,-\theta}(A^*\eta_{a,\theta}B)))
   =\omega(A^*\eta_{a,\theta}B)
   =\Phi(\eta_{a,\theta}B)_A,
\end{align*}
where we used Proposition~\ref{prop:aqft-net}. Hence
$U_{\mathrm{cyl}}(a,\theta)(\mh_0)\subseteq\mh_0$, and therefore
\[
  U_{\mathrm{cyl}}(a,\theta)(\mh)
  =U_{\mathrm{cyl}}(a,\theta)(\overline{\mh_0})
  \subseteq\overline{\mh_0}=\mh.
\]
This proves that $U_{\mathrm{cyl}}(a,\theta)$ is a well-defined continuous linear map from $\mh$ into $\mh$, and the previous display proves item~1 of Theorem~\ref{thm:timeev-main-properties}.

Next, for $z\in\mh$ and $A\in\ma$,
\begin{align*}
  (U_{\mathrm{cyl}}(a,\theta)U_{\mathrm{cyl}}(a',\theta')z)_A
  &=(U_{\mathrm{cyl}}(a',\theta')z)_{\eta_{-a,-\theta}A}
   =z_{\eta_{-a',-\theta'}\eta_{-a,-\theta}A} \\
  &=z_{\eta_{-(a+a'),-(\theta+\theta')}A}
   =(U_{\mathrm{cyl}}(a+a',\theta+\theta')z)_A.
\end{align*}
Hence $(U_{\mathrm{cyl}}(a,\theta))_{(a,\theta)}$ is an action of
$\R\times\mathbb S^1$ on $\mh$, and
\[
  U_{\mathrm{cyl}}(a,\theta)^{-1}=U_{\mathrm{cyl}}(-a,-\theta).
\]
Since both maps are continuous, each $U_{\mathrm{cyl}}(a,\theta)$ is a
homeomorphism.

Now let $u=\Phi(A)\in\mh_0$ and $z\in\mh$. Then by item 1, 
\[
  (U_{\mathrm{cyl}}(a,\theta)u,z)
  =(\Phi(\eta_{a,\theta}A),z)
  =z_{\eta_{a,\theta}A}
  =(U_{\mathrm{cyl}}(-a,-\theta)z)_A
  =(u,U_{\mathrm{cyl}}(-a,-\theta)z).
\]
Applying this with $z=U_{\mathrm{cyl}}(a,\theta)v$ and
using the inverse relation gives
\[
  (U_{\mathrm{cyl}}(a,\theta)u,U_{\mathrm{cyl}}(a,\theta)v)=(u,v),
  \qquad u,v\in\mh_0.
\]
This proves item~2. Since $\Omega=\Phi(1)$, item~1 gives
\[
  U_{\mathrm{cyl}}(a,\theta)\Omega=\Phi(\eta_{a,\theta}1)=\Phi(1)=\Omega.
\]
Now let $C\in\ma$, $z\in\mh$, and $A\in\ma$. Then
\begin{align*}
  (U_{\mathrm{cyl}}(a,\theta)\pi(C)U_{\mathrm{cyl}}(a,\theta)^{-1}z)_A
  &=(\pi(C)U_{\mathrm{cyl}}(-a,-\theta)z)_{\eta_{-a,-\theta}A} \\
  &=(U_{\mathrm{cyl}}(-a,-\theta)z)_{C^*\eta_{-a,-\theta}A}
   =z_{\eta_{a,\theta}(C^*\eta_{-a,-\theta}A)} \\
  &=z_{\eta_{a,\theta}(C)^*A}
   =(\pi(\eta_{a,\theta}C)z)_A,
\end{align*}
because $\eta_{a,\theta}$ is a $*$-automorphism. This proves item~3. 

To prove item 4, first note that by item 2, \(U_{\mathrm{cyl}}(a,\theta)\) is adjointable
and
\[
U_{\mathrm{cyl}}(a,\theta)^\dagger=U_{\mathrm{cyl}}(a,\theta)^{-1}
=U_{\mathrm{cyl}}(-a,-\theta).
\]
Thus conjugation by \(U_{\mathrm{cyl}}(a,\theta)\),
\[
\Gamma_{a,\theta}(T)
:=
U_{\mathrm{cyl}}(a,\theta)T U_{\mathrm{cyl}}(a,\theta)^{-1},
\]
maps \(\operatorname{End}_{\dagger}(\mh)\) onto itself and satisfies
\[
\Gamma_{a,\theta}(T)^\dagger
=
\Gamma_{a,\theta}(T^\dagger).
\]
Since \(U_{\mathrm{cyl}}(a,\theta)\) and its inverse send bounded sets to bounded sets,
\(\Gamma_{a,\theta}\) is a homeomorphism of
\(\operatorname{End}_{\dagger}(\mh)\) for the graph-bounded topology
\(\mathcal T_\dagger\).

If $O_1\subset O_2$, then Proposition~\ref{prop:aqft-net} gives
$\ma(O_1)\subset\ma(O_2)$, hence
$\pi(\ma(O_1))\subset\pi(\ma(O_2))$. Therefore the closed unital $*$-subalgebra
generated by $\pi(\ma(O_1))$ is contained in the corresponding closed unital
$*$-subalgebra generated by $\pi(\ma(O_2))$, that is,
\[
  \mathfrak A_{\mathrm{loc}}(O_1)\subset\mathfrak A_{\mathrm{loc}}(O_2).
\]
Likewise, item~3 gives
\[
  \Gamma_{U_{\mathrm{cyl}}(a,\theta)}\bigl(\pi(\ma(O))\bigr)
  =\pi\bigl(\eta_{a,\theta}\ma(O)\bigr)
  =\pi\bigl(\ma(T_{a,\theta}O)\bigr),
\]
again by Proposition~\ref{prop:aqft-net}. Since
$\Gamma_{U_{\mathrm{cyl}}(a,\theta)}$ is a homeomorphism and an algebra
automorphism of $\mathrm{End}_\dagger(\mh)$, it carries the closed
unital $\dagger$-subalgebra generated by $\pi(\ma(O))$ onto the closed unital
$\dagger$-subalgebra generated by $\pi(\ma(T_{a,\theta}O))$. Hence
\[
  U_{\mathrm{cyl}}(a,\theta)\,\mathfrak A_{\mathrm{loc}}(O)
  \,U_{\mathrm{cyl}}(a,\theta)^{-1}
  =\mathfrak A_{\mathrm{loc}}(T_{a,\theta}O).
\]
Finally, by definition $U(t)=U_{\mathrm{cyl}}(t,0)$ and
$\tau_t=\eta_{-t,0}$, so item~1 gives
\[
  U(t)\Phi(A)=\Phi(\eta_{t,0}A)=\Phi(\tau_{-t}A),\qquad A\in\ma,
\]
and the one-parameter-group statement in item 5 is just the \(\theta=0\) special case of
the algebraic action already proved. No continuity in the group parameter is asserted here. This completes the proof of
Theorem~\ref{thm:timeev-main-properties}.
\end{proof}

\subsection{Locality for the operator algebra}\label{sec:locproof}
In this section we prove Theorem \ref{thm:microcausality} and Theorem \ref{thm:einstein-causality}. A key technical input is a
simple branch-symmetry identity for the function $g$, recorded in
Lemma~\ref{lem:spacelike-bridge-branch}. We use it first to prove Theorem \ref{thm:microcausality} using a combinatorial argument. 

We then prove Theorem~\ref{thm:einstein-causality}.  With the pointwise exchange statement in hand, we
reformulate it in the notation of Lorentzian correlators, derive exchange of
adjacent spacelike separated blocks in Lorentzian correlators, use that
symmetry to obtain an algebraic commutativity identity for smeared
observables, and finally pass from the algebraic represented cores to the
closed represented local algebras.

\subsubsection{Locality of the Green's function}
Recall from \secref{sec:einstein-causality} the function 
\(d_{1}\) on \(\R^2\) and the corresponding notion of
spacelike separation on the Lorentzian cylinder. The following lemma is a key ingredient in the proof of locality.

\begin{lmm}\label{lem:spacelike-bridge-branch}
Let \(t,x\in\mathbb R\) and suppose that $|t|<d_{1}(x)$. 
Then $g(\i t,x)=g(-\i t,-x)$.
\end{lmm}

\begin{proof}
Choose a lift of \(x\) with \(|t|<|x|\le 1\); the general case follows by the
\(2\)-periodicity in \(x\). It is enough to consider \(0<x\le 1\), since the
case \(x<0\) is obtained by replacing \((t,x)\) by \((-t,-x)\). Let
\[
  L(y):=\log(1-e^{-\pi \i y}),
\]
with the principal branch of $\log$. For \(y\notin2\mathbb Z\), let
\(r(y)\in(0,2)\) be the representative of \(y\) modulo \(2\). Write
\(y=r(y)+2m\) for some \(m\in\mathbb Z\).  Since $1-e^{\pi\i y}$ and $1-e^{-\pi\i y}$ are both in the open right half-plane and $\log a - \log b = \log (a/b)$ for $a,b$ in the open right half-plane,
\[
\begin{aligned}
  L(-y)-L(y)
  &= \log\biggl(\frac{1-e^{\pi \i y}}{1-e^{-\pi \i y}}\biggr) \\
  &= \log(-e^{\pi \i y}) = \log(-e^{\pi \i r(y)}) = \log(e^{\pi \i (r(y)-1)}).
\end{aligned}
\]
Since \(r(y)-1\in(-1,1)\), the number \(e^{\pi \i (r(y)-1)}\) has principal
argument \(\pi(r(y)-1)\in(-\pi,\pi)\), so the principal logarithm gives
\begin{equation}\label{eq:log-branch-difference}
  L(-y)-L(y)=\i\pi(r(y)-1).
\end{equation}
Since \(|t|<x\le 1\), we have
\[
  t+x\in(0,2),\qquad t-x\in(-2,0).
\]
Thus
\[
  r(t+x)=t+x,\qquad r(t-x)=t-x+2.
\]
Using the definition of $g$, equation~\eqref{eq:log-branch-difference},
and the last two identities gives
\[
\begin{aligned}
  g(\i t,x)-g(-\i t,-x)
  &=
  -\pi \i t
  +\frac{\i\pi}{2}\{r(t+x)+r(t-x)-2\}  \\
  &=
  -\pi \i t+\frac{\i\pi}{2}(2t)=0.
\end{aligned}
\]
This proves the claim.
\end{proof}

\subsubsection{Phantom charges}
The next lemma allows us to insert ``phantom charges'' without altering the Lorentzian correlator. This is later used to deform the contour $\gamma$ in Theorem \ref{thm:main} without altering the outcome. 

\begin{lmm}[Zero-charge insertions are invisible]
\label{lem:zero-charge-invisible}
Let \(\alpha_1,\ldots,\alpha_k\) satisfy the hypotheses of Theorem \ref{thm:main}, and let $u_1,\ldots,u_k\in M$ satisfy the non-light-cone condition. Take any additional point $u_*$ so that the expanded list $u_1,\ldots,u_k,u_*$ also satisfies the non-light-cone condition. Then for any $0\le r\le k$, 
\begin{align*}
&C^{\mathrm{L}}_{\alpha_1,\ldots,\alpha_r,0,\alpha_{r+1},\ldots,\alpha_k}
(u_1,\ldots,u_r,u_*,u_{r+1},\ldots,u_k)  \\
&\qquad =
C^{\mathrm{L}}_{\alpha_1,\ldots,\alpha_k}(u_1,\ldots,u_k).
\end{align*}
\end{lmm}

\begin{proof}
The screening number $w$ remains unchanged after adding the extra charge. In the Euclidean formula, every free factor involving the inserted charge has exponent
zero, and every screening-external factor involving the inserted charge also has exponent
zero. Thus the Euclidean correlator with the inserted zero-charge point is identically the
original Euclidean correlator:
\begin{align*}
&
C_{\alpha_1,\ldots,\alpha_r,0,\alpha_{r+1},\ldots,\alpha_k}
(q_1,\ldots,q_r,q_*,q_{r+1},\ldots,q_k)  \\
&\qquad =
C_{\alpha_1,\ldots,\alpha_k}(q_1,\ldots,q_k), \qquad q_1,\ldots,q_k,q_*\in M.
\end{align*}
Let us write $u_j = (t_j,x_j)$ and $u_* = (t_*,x_*)$. Then by the above observation, we see that  on the enlarged ordered tube $\Omega_{k+1}$, the function
\[
(\tau_1,\ldots,\tau_r,\tau_*,\tau_{r+1},\ldots,\tau_k)
\longmapsto
C_{\alpha_1,\ldots,\alpha_k}(\tau_1,x_1;\ldots;\tau_k,x_k)
\]
is analytic and has the same Euclidean boundary values as the analytic continuation of
the enlarged zero-charge correlator
\begin{align*}
&(\tau_1,\ldots,\tau_r,\tau_*,\tau_{r+1},\ldots,\tau_k)\\
&\longmapsto
C_{\alpha_1,\ldots,\alpha_r,0,\alpha_{r+1},\ldots,\alpha_k}(\tau_1,x_1;\ldots;\tau_r,x_r; \tau_*,x_*; \tau_{r+1},x_{r+1};\ldots;\tau_k,x_k).
\end{align*}
By uniqueness of analytic continuation to the ordered
tube (Lemma \ref{lem:uniq-analytic-cont}), the two analytic functions coincide. Taking Lorentzian boundary values gives the
claim.
\end{proof}

\subsubsection{Backtrack erasure}
The insertion of phantom charges allows us to deform the contour $\gamma$ in Theorem \ref{thm:main} in certain ways, without affecting the result.
\begin{lmm}[Backtrack erasure]\label{lem:be}
The contour $\gamma$ in Theorem \ref{thm:main} can be deformed in the following way, without altering the result. For any $1\le j\le k-1$, the part of the contour from $\i t_j$ to $\i t_{j+1}$ can be replaced by any number of piecewise linear movements in either direction on the imaginary axis, joined end to end, as long as it starts from $\i t_j$ and ends at~$\i t_{j+1}$. 
\end{lmm}
\begin{proof}
For any $r_1,\ldots, r_n \in \R$ and $y_1,\ldots, y_n\in [-1,1]$, let us insert phantom charges at $(r_1,y_1),\ldots,(r_n, y_n)$ in between $(t_j,x_j)$ and $(t_{j+1}, x_{j+1})$. Given $r_1,\ldots,r_n$, we choose $y_1,\ldots,y_n$ such that these new insertions preserve the no-light-cone condition. By Lemma~\ref{lem:zero-charge-invisible}, the insertion of these phantom charges does not alter the value of the correlation. Moreover, the only effect of inserting these charges in the formula for the correlation is that it alters the contour $\gamma$, which now goes from $\i t_j$ to $\i r_1$, then to $\i r_2,\ldots,\i r_n$, and finally to $\i t_{j+1}$. This proves the claim.
\end{proof}

\subsubsection{An algebraic lemma}
Let $n$ be a positive integer. Consider a collection of symbols $a_1^\pm,\ldots,a_n^\pm$, $b_1^\pm,\ldots, b_n^\pm$, and $S_{ij}^\pm$, $1\le i<j\le n$. Let $\mc$ be the free commutative algebra over $\C$ generated by these symbols. We will now define certain elements of $\mc$. First, define 
\[
        \delta_i^A:=a_i^+-a_i^-,
        \quad
        \delta_i^B:=b_i^+-b_i^- ,
        \quad 
        \Delta_{ij}:=S_{ij}^+-S_{ij}^-.
\]
Next, define 
\[
        U_i:=a_i^+b_i^-,
        \quad 
        W_{1,i}:=a_i^-b_i^-,
        \quad
        W_{2,i}:=a_i^-b_i^+,
        \quad
        W_{3,i}:=a_i^+b_i^+.
\]
Next, for $c,d\in \{1,2,3\}$, let
\[
S_{ij}^{cd} := 
\begin{cases}
S_{ij}^+ &\text{ if } c =1 \text{ or } d =3,\\
S_{ij}^- &\text{ otherwise.}
\end{cases}
\]
Finally, let $\sigma_1 = 1$, $\sigma_2=-1$, and $\sigma_3 = 1$. We will prove the following lemma.
\begin{lmm}\label{lem:alglmm}
Let $D\in \mc$ be defined as 
\begin{align*}
        D
        &:=
        \sum_{c_1,\ldots,c_n\in\{1,2,3\}}
        \biggl(\prod_{i=1}^n \sigma_{c_i}\biggr)
        \biggl[
        \prod_{i=1}^n U_i
        -
        \prod_{i=1}^n W_{c_i,i}
        \biggr]
        \prod_{1\le i<j\le n} S_{ij}^{c_i c_j}.
\end{align*}
Then $D$ belongs to the ideal $\mathcal{I}$ of the algebra $\mc$ generated by elements of the form 
\begin{align}\label{eq:chain}
        \delta_{i_0}^A
        \Delta_{i_0 i_1}
        \Delta_{i_1 i_2}
        \cdots
        \Delta_{i_{l-1}i_l}
        \delta_{i_l}^B,
        \qquad
        1\le i_0<i_1<\cdots<i_l\le n,
\end{align}
where \(l\ge0\). For \(l=0\), the generator is simply $\delta_i^A\delta_i^B$. 
\end{lmm}
\begin{proof}
First, we substitute
\[
        a_i^+=a_i^-+\delta_i^A,\quad
        b_i^+=b_i^-+\delta_i^B,\quad
        S_{ij}^+=S_{ij}^-+\Delta_{ij},
\]
and write $D$ as a polynomial in $a_i^-$, $b_i^-$, $S_{ij}^+$, $\delta_i^A$, $\delta_i^B$, and $\Delta_{ij}$, for $1\le i<j\le n$. To write this out carefully, we proceed as follows. Introduce, for \(c\in\{1,2,3\}\),
\[
        r_c:=1_{\{c=3\}},
        \quad
        \ell_c:=1_{\{c=1\}}.
\]
Then we have 
\[
        W_{c,i}
        =
        (a_i^-+r_c\delta_i^A)
        (b_i^-+(1-\ell_c)\delta_i^B),\quad S_{ij}^{cd} = S_{ij}^+ -(1-\ell_c)(1-r_d)\Delta_{ij}. 
\]
With these substitutions, we have
\begin{align*}
        D_2
        &:=
        \sum_{c_1,\ldots,c_n\in\{1,2,3\}}
        \prod_{i=1}^n \sigma_{c_i}
        \prod_{i=1}^n W_{c_i,i}
        \prod_{1\le i<j\le n} S_{ij}^{c_i c_j}\\
        &= \sum_{c_1,\ldots,c_n\in\{1,2,3\}} \sum_{P,Q,E} Y_{P,Q,E}
        \prod_{i=1}^n \sigma_{c_i}
        \prod_{i\in P} r_{c_i} \delta_i^A \prod_{i\in Q} (1-\ell_{c_i}) \delta_i^B\\
        &\hskip1in \cdot \prod_{(i,j)\in E}(1-\ell_{c_i})(1-r_{c_j}) \Delta_{ij},
\end{align*}
where $Y_{P,Q,E}$ is a polynomial in $a_i^-$, $b_i^-$, and $S_{ij}^+$, for $1\le i<j\le n$; $P,Q$ are subsets of $\{1,\ldots,n\}$; and $E$ is a subset of $\{(i,j):1\le i<j\le n\}$.

Since $Y_{P,Q,E}$ has no dependence on $c_1,\ldots,c_n$, we can interchange the order of summation and write
\[
D_2 = \sum_{P,Q,E} Y_{P,Q,E}Z_{P,Q,E},
\]
where 
\begin{align*}
Z_{P,Q,E} &:= \sum_{c_1,\ldots,c_n\in\{1,2,3\}} \prod_{i=1}^n \sigma_{c_i}
        \prod_{i\in P} r_{c_i} \delta_i^A \prod_{i\in Q} (1-\ell_{c_i}) \delta_i^B\\
        &\hskip1in \cdot \prod_{(i,j)\in E}(1-\ell_{c_i})(1-r_{c_j}) \Delta_{ij}\\
        &= C_{P,Q,E} \prod_{i\in P} \delta_i^A \prod_{i\in Q} \delta_i^B\prod_{(i,j)\in E} \Delta_{ij},
\end{align*}
with 
\[
C_{P,Q,E} := \sum_{c_1,\ldots,c_n\in\{1,2,3\}} \prod_{i=1}^n \sigma_{c_i}
        \prod_{i\in P} r_{c_i} \prod_{i\in Q} (1-\ell_{c_i})  \prod_{(i,j)\in E} (1-\ell_{c_i})(1-r_{c_j}).
\]
Suppose that $Q\ne \varnothing$ and that the coefficient $C_{P,Q,E}$ is nonzero. Then, we claim there is some $i_0\in P$ and $i_l\in Q$ such that either $i_0=i_l$, or $i_0<i_l$ and there is a directed path from $i_0$ to $i_l$ through the directed graph defined by $E$.

To see this, first note that since $r_c$ takes values in $\{0,1\}$, we have $r_c^m = r_c$ for any positive integer $m$. The same is true for $\ell_c$ and $1-\ell_c$. Thus,
\begin{align*}
\prod_{i\in P} r_{c_i} \prod_{i\in Q} (1-\ell_{c_i})  \prod_{(i,j)\in E}(1-\ell_{c_i})(1-r_{c_j}) &= \prod_{i=1}^n r_{c_i}^{p_i} (1-\ell_{c_i})^{L_i} (1-r_{c_i})^{R_i},
\end{align*}
where 
\begin{align*}
p_i &:= 1_{\{i\in P\}}, \\
L_i &:= 1_{\{i\in Q \text{ or } i \text{ has an outgoing edge in } E\}},\\
R_i &:= 1_{\{i \text{ has an incoming edge in } E\}}.
\end{align*}
From this, we get
\begin{align*}
C_{P,Q,E} &= \sum_{c_1,\ldots,c_n\in\{1,2,3\}} \prod_{i=1}^n \sigma_{c_i}r_{c_i}^{p_i} (1-\ell_{c_i})^{L_i} (1-r_{c_i})^{R_i}= \prod_{i=1}^n M(p_i,L_i,R_i),
\end{align*}
where
\[
M(p,L,R) := \sum_{c=1}^3 \sigma_c r_{c}^{p} (1-\ell_{c})^{L} (1-r_{c})^{R}.
\]
A simple calculation yields the following table for $M$:
\begin{align}\label{eq:table1}
\begin{array}{c|c|c|c}
p & L & R & M(p,L,R)\\
\hline
0&0&0&1\\
0&1&0&0\\
0&0&1&0\\
0&1&1&-1\\
1&0&0&1\\
1&1&0&1\\
1&0&1&0\\
1&1&1&0 .
\end{array}
\end{align}
Now, if there is some $i$ such that $i$ is both in $P$ and $Q$, then the claim is already proved. So let us assume that $P,Q$ are disjoint. Take any $i\in Q$. Then $L_i=1$ and (since $i\notin P$) $p_i = 0$. If $i$ has no incoming edge in $E$, then $R_i=0$. But then by table \eqref{eq:table1}, $M(p_i,L_i,R_i) =0$, and this is impossible since $C_{P,Q,E}\ne 0$. Thus, $i$ must have an incoming edge in $E$. Let $(i_1,i)$ be such an edge. If $i_1\notin P$, then the same argument shows that $i_1$ must have an incoming edge $(i_2,i_1)$ in $E$. The process cannot continue indefinitely. This yields a directed path from $P$ to $Q$ in $E$, proving the claim.

Thus, we have proved the claim that if  $Q\ne \varnothing$ and $C_{P,Q,E}\ne 0$, then there is a path from $P$ to $Q$ in $E$. Consequently, for any $Q\ne \varnothing$, $Y_{P,Q,E}Z_{P,Q,E} \in \mathcal{I}$. 

Next, suppose that $Q = \varnothing$ but $E \ne \varnothing$, and $C_{P,Q,E}\ne 0$. Take any $i\in P$. Then $p_i=1$. So, for $C_{P,Q,E}$ to be nonzero, $R_i$ must be zero (by table \eqref{eq:table1}). In other words, $i$ cannot have an incoming edge in $E$. On the other hand, if $i\notin P$, then $p_i=0$, and therefore by table \eqref{eq:table1}, we must have $L_i=R_i$. But since $Q=\varnothing$ the condition $L_i=R_i$ implies that $i$ has both an incoming edge in $E$ and an outgoing edge in $E$. To summarize, we have shown that every $i$ has either no incoming edge in $E$, or both an incoming and an outgoing edge in $E$. We claim that this is impossible. To see this, take the largest $i$ that has at least one edge from $E$ that is incident to it. This is well-defined, since $E$ is nonempty. Since $i$ is the largest such vertex, there can be no outgoing edge incident to it. So there must be an incoming edge but no outgoing edge, which contradicts our previous deduction.  This proves that if $Q=\varnothing$ and $E\ne \varnothing$, then $C_{P,Q,E}$ must be zero. 

Thus, the only part of $D_2$ that we have not yet shown to be in $\mathcal{I}$ is 
\[
D_2' := \sum_P Y_{P,\varnothing, \varnothing} Z_{P, \varnothing, \varnothing}. 
\]
We will take care of this part later. Next, writing $U_i = (a_i^-+\delta_i^A)b_i^-$, we get the expansion
\begin{align*}
        D_1
        &:=
        \sum_{c_1,\ldots,c_n\in\{1,2,3\}}\prod_{i=1}^n \sigma_{c_i}\prod_{i=1}^n U_i\prod_{1\le i<j\le n} S_{ij}^{c_i c_j}\\
        &=  \sum_{c_1,\ldots,c_n\in\{1,2,3\}}\sum_{P,E}W_{P,E}\prod_{i=1}^n \sigma_{c_i}\prod_{i\in P}\delta_i^A \prod_{(i,j)\in E} (1-\ell_{c_i})(1-r_{c_j})\Delta_{ij},
\end{align*}
where $W_{P,E}$ is a polynomial in $a_i^-$, $b_i^-$, and $S_{ij}^+$, for $1\le i<j\le n$; $P$ is a subset of $\{1,\ldots,n\}$; and $E$ is a subset of $\{(i,j):1\le i<j\le n\}$.

Since $W_{P,E}$ has no dependence on $c_1,\ldots,c_n$, we can interchange the order of summation and write
\[
D_1 = \sum_{P,E} W_{P,E}V_{P,E},
\]
where 
\begin{align*}
V_{P,E} &:= \sum_{c_1,\ldots,c_n\in\{1,2,3\}} \prod_{i=1}^n \sigma_{c_i}
        \prod_{i\in P} \delta_i^A  \prod_{(i,j)\in E} (1-\ell_{c_i})(1-r_{c_j}) \Delta_{ij}\\
        &= U_{P,E} \prod_{i\in P} \delta_i^A \prod_{(i,j)\in E} \Delta_{ij},
\end{align*}
with 
\[
U_{P,E} := \sum_{c_1,\ldots,c_n\in\{1,2,3\}} \prod_{i=1}^n \sigma_{c_i}  \prod_{(i,j)\in E} (1-\ell_{c_i})(1-r_{c_j}).
\]
Define 
\begin{align*}
L_i &:= 1_{\{i \text{ has an outgoing edge in } E\}},\quad R_i := 1_{\{i \text{ has an incoming edge in } E\}}.
\end{align*}
Then 
\begin{align*}
U_{P,E} &= \sum_{c_1,\ldots,c_n\in\{1,2,3\}} \prod_{i=1}^n \sigma_{c_i}(1-\ell_{c_i})^{L_i} (1-r_{c_i})^{R_i}\\
&= \prod_{i=1}^n \biggl(\sum_{c=1}^3  \sigma_c (1-\ell_c)^{L_i}(1-r_c)^{R_i}\biggr) = \prod_{i=1}^n M(0, L_i, R_i).
\end{align*}
Suppose that $E\ne \varnothing$. Then for $U_{P,E}$ to be nonzero, the above formula shows that each $M(0,L_i,R_i)$ has to be nonzero. By table \eqref{eq:table1}, this means that each $i$ has either no incoming or outgoing edges in $E$, or has an incoming and an outgoing edge in $E$. But this is impossible, as can be seen by considering the largest $i$ which has an edge in $E$ incident to it. Thus, $U_{P,E} = 0$ whenever $E\ne \varnothing$. We conclude that 
\[
D_1 = \sum_P W_{P,\varnothing} V_{P,\varnothing}. 
\]
Our final claim is that \(D_1=D_2'\). Indeed, when \(Q=\varnothing\) and
\(E=\varnothing\), no \(\delta_i^B\)-factor and no \(\Delta_{ij}\)-factor
has been selected. Therefore
\[
        Y_{P,\varnothing,\varnothing}
        =
        \biggl(\prod_{i\notin P} a_i^-\biggr)
        \biggl(\prod_{i=1}^n b_i^-\biggr)
        \biggl(\prod_{1\le i<j\le n} S_{ij}^+\biggr).
\]
Also,
\begin{align*}
Z_{P,\varnothing,\varnothing}
&=
\sum_{c_1,\ldots,c_n\in\{1,2,3\}}
\prod_{i=1}^n\sigma_{c_i}
\prod_{i\in P} r_{c_i}\delta_i^A  \\
&=
\biggl(\prod_{i\in P}\delta_i^A\biggr)
\biggl(\prod_{i\notin P}\sum_{c=1}^3\sigma_c\biggr)
\biggl(\prod_{i\in P}\sum_{c=1}^3\sigma_c r_c\biggr)
=
\prod_{i\in P}\delta_i^A,
\end{align*}
since \(\sum_{c=1}^3\sigma_c=1\) and
\(\sum_{c=1}^3\sigma_c r_c=1\). Similarly, in the expansion of \(D_1\),
\[
        W_{P,\varnothing}
        =
        \biggl(\prod_{i\notin P} a_i^-\biggr)
        \biggl(\prod_{i=1}^n b_i^-\biggr)
        \biggl(\prod_{1\le i<j\le n} S_{ij}^+\biggr),
        \qquad
        V_{P,\varnothing}
        =
        \prod_{i\in P}\delta_i^A .
\]
Thus
\[
        W_{P,\varnothing}V_{P,\varnothing}
        =
        Y_{P,\varnothing,\varnothing}Z_{P,\varnothing,\varnothing}
\]
for every \(P\), and hence \(D_1=D_2'\). Since the preceding argument showed
that \(D_2-D_2'\in\mathcal I\), we conclude that
\[
        D=D_1-D_2=(D_1-D_2')-(D_2-D_2')\in\mathcal I .
\]
This completes the proof.
\end{proof}

\subsubsection{Transitivity of strict timelike separation}
We will say that two points $u = (t,x)$ and $v = (s,y)$ in $M$ are strictly timelike separated if
\[
|t-s| > d_1(x-y). 
\]
If we moreover have that $t< s$, then we will write $u \prec v$. The following lemma shows that this relation is transitive.
\begin{lmm}[Transitivity of timelike separation]\label{lem:transitive}
If $u,v,w\in M$ satisfy $u\prec v$ and $v\prec w$, then $u\prec w$.
\end{lmm}
\begin{proof}
Write $u = (t,x)$, $v=(s,y)$, and $w = (r,z)$. Then $t<s<r$, and 
\[
s - t > d_1(x-y), \quad r-s > d_1(y-z). 
\]
We claim that for any $a,b$, $d_1(a+b)\le d_1(a)+d_1(b)$. To see this, find integers $m,n$ that minimize $|a+2m|$ and $|b+2n|$. Then 
\[
d_1(a+b) \le |a+b+2(m+n)|\le |a+2m|+|b+2n|=d_1(a)+d_1(b).
\]
Having proved the claim, we may now conclude that 
\begin{align*}
r - t > d_1(x-y)+d_1(y-z) \ge d_1(x-z).
\end{align*}
This shows that $u\prec w$. 
\end{proof}

\subsubsection{Proof of Theorem \ref{thm:microcausality}}
Inserting phantom charges if necessary, let us assume that $2\le p\le k-2$. Swapping $v_p,v_{p+1}$ if necessary, let us assume that $t_p\le t_{p+1}$. The case $t_p=t_{p+1}$ is easy, because the formula~\eqref{eq:clorentz} remains unchanged if we swap $v_p,v_{p+1}$ in this case. So let us assume that $t_p<t_{p+1}$.

Write $v_i = (t_i, x_i)$, $i=1,\ldots,k$. Let $\gamma_1$ and $\gamma_2$ denote the contours for the two orders. We will now replace them both by a common contour, using Lemma \ref{lem:be}. For $\gamma_1$, replace the path from $\i t_p$ to $\i t_{p+1}$ by a path that first goes from $\i t_p$ to $\i t_{p+1}$, then back to $\i t_p$, and finally to $\i t_{p+1}$. Call this new path $\gamma$. For $\gamma_2$, replace the path from $\i t_{p-1}$ to $\i t_{p+1}$ by a path that first goes from $\i t_{p-1}$ to $\i t_p$ and then to $\i t_{p+1}$. After this, replace the path from $\i t_{p}$ to $\i t_{p+2}$ by a path that first goes from $\i t_p$ to $\i t_{p+1}$, and then to $\i t_{p+2}$. Notice that this new path is also $\gamma$. 

To summarize, the contour $\gamma$ goes from $\i t_{p-1}$ to $\i t_p$, then to $\i t_{p+1}$, then back to $\i t_p$, then again to $\i t_{p+1}$, and finally to $\i t_{p+2}$. The part $\i t_p \to \i t_{p+1}\to \i t_p \to \i t_{p+1}$ will be called the ``hook''. The hook has three segments, which will be referred to as segments $1$, $2$, and $3$. 

Write the difference between the two correlations as a $w$-fold contour integral over $\gamma$ of the difference between the respective integrands. Let $u_1,\ldots, u_w$ denote the screening charges in this integral, and write $u_i = (\i \tau_i, y_i)$. 

Consider the part of the integral where $\tau_{i_1},\ldots, \tau_{i_n}$ are in the hook and $\tau_l$ is outside the hook for all $l\notin \{i_1,\ldots, i_n\}$. Here $0\le n\le w$ is arbitrary. Clearly, the integral splits as the sum of these parts. We will show that each part is zero. By symmetry, it suffices to prove this when  $i_1=1,\ldots, i_n = n$. Let $I_0$ denote this part of the integral.

Take any point $(u_1,\ldots,u_w)$ within the integral $I_0$.  For $i=1,\ldots,n$, let $c_i\in \{1,2,3\}$ denote the segment number (in the hook) to which $u_i$ belongs. For $c\in \{1,2,3\}$, define
\[
\sigma_c :=
\begin{cases}
1 &\text{ if } c = 1 \text{ or } 3,\\
-1 &\text{ if } c = 2.
\end{cases}
\]
Then the integration of $\tau_1,\ldots,\tau_n$ in the integral $I_0$ can be written (schematically) as 
\begin{align}\label{eq:schematic}
&\i^n \int_{t_p}^{t_{p+1}}\cdots \int_{t_p}^{t_{p+1}} \sum_{c_1,\ldots,c_n\in \{1,2,3\}}\prod_{i=1}^n \sigma_{c_i}\biggl\{\ldots\biggr\} \, \dd \tau_1\cdots \dd \tau_n\notag\\
&= \i^n n! \int_{t_p <\tau_1<\cdots<\tau_n< t_{p+1}} \sum_{c_1,\ldots,c_n\in \{1,2,3\}}\prod_{i=1}^n \sigma_{c_i}\biggl\{\ldots\biggr\} \, \dd \tau_1\cdots \dd \tau_n.
\end{align}
We will now explicitly write down the term in $\{\ldots\}$ above, and then use Lemma \ref{lem:alglmm} and the spacelike separation of $v_p, v_{p+1}$ to prove that the integrand evaluates to zero. The integrand is a difference of contributions coming from the original-order integral and the exchanged-order integral. Each contribution is a product of several exponentials. Let $H$ denote the product of  exponentials that are common to the two orders. Each remaining exponential is one of the following:
\begin{align*}
&a_i^+ := e^{-4\alpha_p b g(\i (\tau_i - t_p), y_i - x_p)}, \quad a_i^- := e^{-4\alpha_p b g(\i (t_p - \tau_i), x_p - y_i)},\\
&b_i^+ := e^{-4\alpha_{p+1} b g(\i (\tau_i - t_{p+1}), y_i - x_{p+1})}, \quad b_i^- := e^{-4\alpha_{p+1} b g(\i (t_{p+1} - \tau_i), x_{p+1} - y_i)},\\
&S_{ij}^+ := e^{-4b^2 g(\i(\tau_j-\tau_i),y_j-y_i)}, \quad S_{ij}^- := e^{-4b^2 g(\i(\tau_i-\tau_j),y_i-y_j)},
\end{align*}
where $1\le i\le n$ and $1\le i<j\le n$. In the original-order integral, each $\tau_i$ comes after $t_p$ and before $t_{p+1}$ in the contour $\gamma$. This gives an external-screening contribution of $U_i:= a_i^+ b_i^-$, irrespective of the value of $c_i$. On the other hand, for the exchanged-order integral the value of $c_i$ matters. If $c_i =1$, then $\tau_i$ is on segment 1 of the hook, where $\tau_i$ comes before both $t_p$ and $t_{p+1}$. This results in a contribution of $W_{1,i} := a_i^-b_i^-$. If $c_i=2$, then $\tau_i$ comes after $t_{p+1}$ and before $t_p$. Thus, the contribution is $W_{2,i} := a_i^-b_i^+$. Lastly, if $c_i=3$, then $\tau_i$ comes after both $t_p$ and $t_{p+1}$, which gives a contribution of $W_{3,i} := a_i^+b_i^+$. 

The screening-screening contributions are the same for both orders. If $1\le i<j\le n$, then the order of placement of $\tau_i$ and $\tau_j$ on the contour depends on the values of $c_i,c_j$. If $c_i < c_j$, then $\tau_i$ comes before $\tau_j$. If $c_i > c_j$, then $\tau_j$ comes before $\tau_i$. If $c_i = c_j$, then there are two cases. If $c_i = c_j =1$ or $3$, then $\tau_i$  comes before $\tau_j$. But if $c_i=c_j=2$, then $\tau_j$ comes before $\tau_i$. The whole situation can be expressed more succinctly as follows. If $c_i=1$ or $c_j =3$, then $\tau_i$ comes before $\tau_j$; otherwise, $\tau_j$ comes before $\tau_i$. Thus the exchanged-order contribution from the pair $(i,j)$ is $S_{ij}^{c_ic_j}$, where $S_{ij}^{cd} = S_{ij}^+$ if $c=1$ or $d=3$, and $S_{ij}^{cd} = S_{ij}^-$ otherwise.

Putting it all together, we see that the integrand in equation \eqref{eq:schematic} is equal to 
\begin{align*}
H \sum_{c_1,\ldots,c_n\in\{1,2,3\}}
        \biggl(\prod_{i=1}^n \sigma_{c_i}\biggr)
        \biggl[
        \prod_{i=1}^n U_i
        -
        \prod_{i=1}^n W_{c_i,i}
        \biggr]
        \prod_{1\le i<j\le n} S_{ij}^{c_i c_j}.
\end{align*}
By Lemma \ref{lem:alglmm}, this integrand can be written as a sum of terms, each of which has a factor of the form 
\[
        \delta_{i_0}^A
        \Delta_{i_0 i_1}
        \Delta_{i_1 i_2}
        \cdots
        \Delta_{i_{l-1}i_l}
        \delta_{i_l}^B,
        \qquad
        1\le i_0<i_1<\cdots<i_l\le n,
\]
where $l\ge0$; for $l=0$, this means the factor $\delta_{i_0}^A\delta_{i_0}^B$. Also,
\[
        \delta_i^A:=a_i^+-a_i^-,
        \quad
        \delta_i^B:=b_i^+-b_i^- ,
        \quad 
        \Delta_{ij}:=S_{ij}^+-S_{ij}^-.
\]
We claim that any such factor evaluates to zero. To see this, note that $\delta_i^A = 0$ if $v_p$ and $u_i$ are spacelike separated, $\delta_i^B = 0$ if $u_i$ and $v_{p+1}$ are spacelike separated, and $\Delta_{ij}=0$ if $u_i$ and $u_j$ are spacelike separated. Thus, for the above factor to be nonzero, $v_p$ and $u_{i_0}$ have to be strictly timelike separated, $u_{i_r}, u_{i_{r+1}}$ have to be strictly timelike separated for $0\le r<l$, and $u_{i_l}$ and $v_{p+1}$ have to be strictly timelike separated (all ignoring null sets of lightlike separation). Moreover, the hook variables in the factor are ordered as $t_p<\tau_{i_0}<\cdots < \tau_{i_l}<t_{p+1}$. Hence any nonzero factor would force $v_p \prec u_{i_0}\prec \cdots \prec u_{i_l}\prec v_{p+1}$. By Lemma \ref{lem:transitive}, we conclude that $v_p \prec v_{p+1}$. But this is false by assumption. Thus, none of the factors can have nonzero value. This completes the proof.

\subsubsection{Locality for the operator algebra}
We obtain the following corollary of Theorem \ref{thm:microcausality}.

\begin{cor}[Spacelike block exchange]\label{cor:spacelike-block-exchange}
Let \(I\) and \(J\) be two adjacent blocks of external insertions, and suppose that every point in the \(I\)-block is spacelike separated from every point in
the \(J\)-block. Then the Lorentzian correlator is unchanged when the two
blocks are exchanged, with their internal orders preserved.
\end{cor}

\begin{proof}
Move the insertions of the \(J\)-block leftward past the insertions of the
\(I\)-block by finitely many adjacent transpositions. Each transposition is
allowed by the spacelike separation hypothesis and
Theorem~\ref{thm:microcausality}.
\end{proof}

We are now ready to prove Theorem \ref{thm:einstein-causality}.
\begin{proof}[Proof of Theorem~\ref{thm:einstein-causality}]
We first prove the algebraic commutativity identity. Let
$A\in\ma(O_1)$, $B\in\ma(O_2)$, and let $X,Y\in\ma$. By linearity in all four
variables, it is enough to treat the case where $A$, $B$, $X$, and $Y$ are
monomials. If either $A$ or $B$ is the unit, the identity is immediate, so
assume that $A$ and $B$ have positive lengths. Write
\[
  A=\mo_{f_1,n_1}\cdots \mo_{f_r,n_r},\qquad
  B=\mo_{g_1,m_1}\cdots \mo_{g_s,m_s},
\]
with $n_i,m_j\in\mathcal I_b$, $\supp(f_i)\subseteq O_1$, and
$\supp(g_j)\subseteq O_2$. Because $X$ and $Y$ are monomials in $\ma$, every
point correlator that appears in the smeared integrals for $\omega(XABY)$ and
$\omega(XBAY)$ still has integer charges. In those integrals, the insertions
belonging to $A$ and the insertions belonging to $B$ occur as adjacent blocks,
possibly with spectator insertions coming from $X$ and $Y$. The total screening
number is the same for the two orders. If it is negative, both ordered point
correlators vanish. If it is
nonnegative, then the Lorentzian correlator formula applies. Because $O_1$ and
$O_2$ are spacelike separated, every point in the $A$-block is spacelike
separated from every point in the $B$-block. Corollary~\ref{cor:spacelike-block-exchange}
therefore gives pointwise equality of the two Lorentzian point correlators
corresponding to the orders $XABY$ and $XBAY$, away from the light-cone
exceptional set. That exceptional set has Lebesgue measure zero, and
Theorem~\ref{thm:smear} supplies absolute integrability of the relevant smeared
correlators. Therefore
\begin{equation}\label{eq:algebraic-locality}
  \omega(XABY)=\omega(XBAY)
\end{equation}
for all $X,Y\in\ma$.

Now let $E,F\in\ma$. Since $\pi(C)\Phi(E)=\Phi(CE)$ for every
$C\in\ma$, we
have
\[
  \pi(A)\pi(B)\Phi(E)=\Phi(ABE),
  \qquad
  \pi(B)\pi(A)\Phi(E)=\Phi(BAE).
\]
For each coordinate $F\in\ma$,
\[
\begin{aligned}
  \bigl(\Phi(ABE)-\Phi(BAE)\bigr)_F
  &=
  \omega(F^*ABE)-\omega(F^*BAE).
\end{aligned}
\]
Since $A\in\ma(O_1)$ and $B\in\ma(O_2)$, equation
\eqref{eq:algebraic-locality}, applied with $X=F^*$ and $Y=E$, shows that the right side above is zero. Thus
\[
  \Phi(ABE)=\Phi(BAE),
\]
and therefore
\[
  \pi(A)\pi(B)\Phi(E)=\pi(B)\pi(A)\Phi(E)
\]
for every $E\in\ma$. Since $\mh_0=\Phi(\ma)$ is dense in
$\mh$, and $\pi(A)$ and $\pi(B)$ are continuous, we obtain
\[
  [\pi(A),\pi(B)]=0
\]
on all of $\mh$. By linearity, the same conclusion holds for all
$A\in\ma(O_1)$ and $B\in\ma(O_2)$. Hence every element of the unital
$*$-subalgebra generated by $\pi(\ma(O_1))$ commutes with every element of the
corresponding unital $*$-subalgebra generated by $\pi(\ma(O_2))$.

It remains to pass to the closed represented  local algebras. Let
\[
  S\in\mathfrak A_{\mathrm{loc}}(O_1),
  \qquad
  T\in\mathfrak A_{\mathrm{loc}}(O_2).
\]
Choose nets \(S_i\) and \(T_j\) in the unital \(\dagger\)-subalgebras
\(\pi(\mathcal A(O_1))\) and \(\pi(\mathcal A(O_2))\), respectively, such that
\[
S_i\to S,\qquad T_j\to T
\]
in the graph-bounded topology \(\mathcal T_\dagger\). In particular, these convergences
hold in the underlying topology \(\mathcal T_b\) of uniform convergence on bounded subsets
of \(\mh\). By the
algebraic part just proved, we have
\[
  S_iT_j=T_jS_i
  \qquad\text{for all }i,j.
\]
Composition with a fixed continuous linear map is continuous in this topology.
Indeed, if $R_\nu\to R$, then $R_\nu U\to RU$ because a continuous linear map
$U$ sends bounded sets to bounded sets; and $UR_\nu\to UR$ because for every
continuous seminorm $p$ on $\mh$, the seminorm $p\circ U$ is again continuous.
Therefore, first letting $i$ tend to the limit gives
\[
  ST_j=T_jS
\]
for every $j$, and then letting $j$ tend to the limit gives
\[
  ST=TS.
\]
This proves the theorem.
\end{proof}

This completes the quantization program announced in Section~\ref{sec:gnsquant}.

\phantomsection
\addcontentsline{toc}{section}{Acknowledgements}
\section*{Acknowledgements}
I thank Ashoke Sen for posing the problem, and Edward Witten for numerous invaluable suggestions. I also thank Beatrix M\"uhlmann, Steve Shenker, and Douglas Stanford for illuminating discussions. The research was supported in part by NSF grant DMS-2450608 and the Simons Collaboration grant on ``Probabilistic Paths to QFT''.

\bibliographystyle{abbrvnat}

\phantomsection
\addcontentsline{toc}{section}{References}
\bibliography{myrefs-hilbertspace}

\appendix

\setcounter{section}{0}

\refstepcounter{section}   
\section*{Appendix}        
\addcontentsline{toc}{section}{Appendix}

This appendix collects material that would have interrupted the main narrative: details on the torus Green's function, technical lemmas used at isolated points in the proofs, and general background on topological vector spaces and the standard AQFT framework.


\subsection{Green's function on the torus}\label{sec:green}
Recall that $M_T$ is the torus $[-T,T]\times [-1,1]$ endowed with the flat metric, and $\Delta$ is its Laplace--Beltrami operator. We define the mean-zero Green's function $G_T$ as the inverse of $-\frac{1}{2\pi}\Delta$ on the space of mean-zero smooth functions on $M_T$, in the following sense. For any $f\in C^\infty(M_T)$, we have 
\begin{align}\label{eq:gtdefinition}
  -\frac{1}{2\pi}G_T \Delta f = -\frac{1}{2\pi}\Delta G_T f = f - c(f),
\end{align}
where 
\[
  c(f) := \frac{1}{4T}\int_{M_T} f(u)\, \dd u
\]
is the zero mode of $f$. It turns out that $G_T$ is an integral kernel, that is, a map from $M_T^2$ into $\R$, acting on smooth functions as 
\[
G_T f(u) = \int_{M_T} G_T(u,v) f(v)\, \dd v.  
\]
In the following subsections, we define $G_T$ and show that it has the required property. Along the way, we also define a useful sequence of approximations of $G_T$.

\subsubsection{The truncated Green's function and its limit}
For $t,x\in \R$ and $N\ge 1$, define 
\begin{align*}
&g_{T,N}(t,x)
:=\frac{\pi}{2T}\sum_{\substack{|m|\le N,\ |n|\le N\\ (m,n)\neq(0,0)}}
\frac{1}
{\lambda_{m,n}}\exp\biggl(\pi \mathrm{i} \frac{m}{T}t+\pi \mathrm{i} n x\biggr),
\end{align*}
where 
\[
  \lambda_{m,n} := \pi^2 \biggl(\frac{m^2}{T^2} + n^2\biggr).
\]
We define the truncated Green's function $G_{T,N}:M_T^2 \to \R$ as 
\[
  G_{T,N}((t,x),(s,y)):= g_{T,N}(t-s,x-y).
\]
The following lemma shows that this $G_{T,N}$ is the same as the $G_{T,N}$ defined in equation~\eqref{eq:gtnalt}.
\begin{lmm}\label{lem:GN-form}
The function $G_{T,N}$ defined above is the same as the function $G_{T,N}$ defined in equation \eqref{eq:gtnalt}.
\end{lmm}
\begin{proof}
Let $u:= t-s$ and $v:= x-y$. Then note that
\begin{align*}
&e_{m,n}^{\mathrm{c}}(t,x) e_{m,n}^{\mathrm{c}}(s,y)+ e_{m,n}^{\mathrm{s}}(t,x) e_{m,n}^{\mathrm{s}}(s,y)\\ 
&= \frac{1}{2T}\biggl\{\cos\biggl(\pi\frac{m}{T}t+\pi nx\biggr)\cos\biggl(\pi\frac{m}{T}s+\pi ny\biggr) \\ 
&\qquad \qquad + \sin\biggl(\pi\frac{m}{T}t+\pi nx\biggr)\sin\biggl(\pi\frac{m}{T}s+\pi ny\biggr)\biggr\}\\
&= \frac{1}{2T}\cos\biggl(\pi \frac{m}{T}u + \pi nv\biggr)\\ 
&= \frac{1}{4T}\biggl\{\exp\biggl(\pi \i \frac{m}{T}u + \pi\i n v\biggr)+ \exp\biggl(-\pi \i \frac{m}{T}u - \pi\i nv\biggr)\biggr\}.
\end{align*}
By the definitions of $G_{T,N}$ and $\Z^2_*$, this shows that 
\begin{align*}
&G_{T,N}((t,x),(s,y)) \\ 
&= \frac{\pi}{2T} \sum_{\substack{|m|\le N, \ |n|\le N\\ (m,n)\in\Z^2_*\setminus\{(0,0)\}}}\frac{1}{\lambda_{m,n}}
\Biggl\{\exp\biggl(\pi \i \frac{m}{T}u + \pi\i nv\biggr)+ \exp\biggl(-\pi \i \frac{m}{T}u - \pi\i nv\biggr)\Biggr\} \\
&= \frac{\pi}{2T} \sum_{\substack{|m|\le N, \ |n|\le N\\ (m,n)\ne (0,0)}}\frac{1}{\lambda_{m,n}}\exp\biggl(\pi \i \frac{m}{T}u + \pi\i nv\biggr)
= g_{T,N}(u,v).
\end{align*}
This completes the proof.
\end{proof}
To show that $G_{T,N}$ converges to the desired Green's function as $N\to \infty$, we first prove pointwise convergence. For that, we need two lemmas from analysis. The first one is the following.

\begin{lmm}\label{lem:1d-lattice-identity}
For any $a>0$ and $\theta \in \R$, we have 
\begin{align*}
\sum_{k\in \Z} \frac{e^{\i k\theta}}{k^2+a^2} &=\frac{\pi}{a}\sum_{k\in \Z} e^{-a |2\pi k + \theta|}.
\end{align*}
As a consequence, we have that for any $a\ge 1$ and $\theta\in[-\pi,\pi]$, 
\begin{equation*}
0\le \sum_{k\in\mathbb{Z}}\frac{e^{\i k\theta}}{k^2+a^2}\le \frac{C}{a}e^{-a|\theta|}
\end{equation*}
for some universal constant $C$.
\end{lmm}

\begin{proof}
Let $f(x) := (\pi/a)e^{-2\pi a|x|}$. The Fourier transform of $f$ is 
\begin{align*}
\hat{f}(\xi) &= \int_{\R} f(x) e^{-2\pi \i x\xi} \dd x\\ 
&= \frac{\pi}{a}\int_0^\infty  e^{-(2\pi a+2\pi \i \xi)x}\dd x + \frac{\pi}{a}\int_0^\infty  e^{-(2\pi a-2\pi \i \xi) x} \dd x\\
&= \frac{\pi}{a(2\pi a+2\pi \i \xi)} + \frac{\pi}{a(2\pi a-2\pi \i \xi)}\\ 
&= \frac{1}{a^2+\xi^2}.
\end{align*}
Then by the Poisson summation formula~\cite[Theorem~3.2.8]{grafakos_classical_2014}, 
\begin{align*}
\sum_{k\in \Z} \frac{e^{\i k\theta}}{k^2+a^2} &= \sum_{k\in \Z} e^{\i k \theta} \hat{f}(k)\notag \\
&= \sum_{k\in \Z}  f \biggl(k + \frac{\theta}{2\pi}\biggr) = \frac{\pi}{a}\sum_{k\in \Z} e^{-a |2\pi k + \theta|}.
\end{align*}
This proves the first claim. Next, suppose that  $\theta \in [-\pi,\pi]$ and $a\ge 1$. Then $2\pi \pm \theta \ge |\theta|$, and hence,
\begin{align*}
\sum_{k\in \Z} e^{-a |2\pi k + \theta|} &= e^{-a|\theta|} + \sum_{k=1}^\infty e^{-a (2\pi k + \theta)} + \sum_{k=1}^\infty e^{-a(2\pi k - \theta)}\\
&= e^{-a|\theta|} + \frac{e^{-a(2\pi + \theta)}}{1-e^{-2\pi a}} + \frac{e^{-a(2\pi - \theta)}}{1-e^{-2\pi a}}\le Ce^{-a|\theta|}, 
\end{align*}
where $C$ is universal. This completes the proof.
\end{proof}
The second lemma from analysis that we need is the following consequence of summation by parts.
\begin{lmm}\label{lem:sumparts}
Let $\{a_n\}_{n\ge 1}$ be a sequence of complex numbers and $\{b_n\}_{n\ge 1}$ be a sequence of positive real numbers decreasing to zero. Suppose that 
\[
  C := \sup_{N\ge 1}\biggl|\sum_{n=1}^N a_n\biggr| <\infty.
\] 
Then the limit 
\begin{align*}
\sum_{n=1}^\infty a_n b_n := \lim_{N\to\infty} \sum_{n=1}^N a_n b_n
\end{align*}
exists, and its absolute value is bounded above by $2Cb_1$.
\end{lmm}
\begin{proof}
Let $A_n:= a_1+\cdots +a_n$ and $A_0:=0$. Then for any $M\le N$,
\begin{align*}
\sum_{n=M}^N a_nb_n &= \sum_{n=M}^N (A_n - A_{n-1})b_n\\ 
&= \sum_{n=M}^N A_n b_n - \sum_{n=M-1}^{N-1} A_n b_{n+1}\\
&= A_N b_N -A_{M-1} b_M+ \sum_{n=M}^{N-1}A_n(b_n - b_{n+1}).
\end{align*}
From the given conditions, this gives
\begin{align*}
\biggl|\sum_{n=M}^N a_nb_n \biggr| &\le Cb_N + C b_M + C\sum_{n=M}^{N-1}|b_n - b_{n+1}|\\ 
&= Cb_N + Cb_M + C(b_M - b_N) = 2Cb_M.
\end{align*}
It is clear that this suffices to prove both assertions.
\end{proof}
We are now ready to prove pointwise convergence of $g_{T,N}$ as $N\to \infty$.
\begin{lmm}\label{lem:gN-to-g}
For every $(t,x)\in\R^2$ such that $t\not\equiv 0\pmod{2T}$ or $x \not\equiv 0 \pmod{2}$, the limit 
\[
  g_T(t,x) := \lim_{N\to\infty} g_{T,N}(t,x)
\]
exists and is finite.
\end{lmm}

\begin{proof}
Since $g_{T,N}(t,x)$ is periodic with period $2T$ in the first coordinate and period $2$ in the second coordinate, we may assume without loss that $(t,x)\in M_T\setminus\{(0,0)\}$. Define $\alpha := \pi t/T$ and $\beta := \pi x$. First, suppose that $\alpha\neq 0$. 
Define, for $n\neq 0$,
\[
B_n:=\sum_{m\in\mathbb{Z}}\frac{e^{\i\alpha m}}{\lambda_{m,n}},
\qquad
B_{n,N}:=\sum_{|m|\le N}\frac{e^{\i\alpha m}}{\lambda_{m,n}},
\]
and for $n=0$ define
\[
B_0^\ast:=\sum_{m\in\mathbb{Z}\setminus\{0\}}\frac{e^{\i\alpha m}}{\lambda_{m,0}},
\qquad
B_{0,N}^\ast:=\sum_{0<|m|\le N}\frac{e^{\i\alpha m}}{\lambda_{m,0}}.
\]
(Note that all of the above series are absolutely convergent.) By definition,
\begin{align}\label{eq:gnexp}
g_{T,N}(t,x)&=\frac{\pi}{2T}\biggl(B_{0,N}^\ast+\sum_{0<|n|\le N}e^{\i\beta n}B_{n,N}\biggr).
\end{align}
We claim that 
\begin{equation}\label{eq:blimit}
\lim_{N\to \infty}\sum_{0<|n|\le N}|B_n-B_{n,N}| = 0,
\qquad
\lim_{N\to \infty}|B_0^\ast-B_{0,N}^\ast|=0.
\end{equation}
To prove this, first note that 
\begin{align*}
B_n - B_{n,N}&=\sum_{|m|> N}\frac{e^{\i\alpha m}}{\lambda_{m,n}},
\end{align*}
that $\lambda_{m,n}$ is positive, $1/\lambda_{m,n}$ decreases to zero as $m\to\infty$, and that for any $0\le M_1\le M_2$, 
\[
  \sum_{m=M_1}^{M_2}e^{\pm\i\alpha m} = \frac{e^{\pm\i \alpha M_1} - e^{\pm\i \alpha (M_2+1)}}{1-e^{\pm \i \alpha}},
\]
whose absolute value is clearly bounded by a constant $C_1(\alpha)$ that depends only on $\alpha$. Thus, by Lemma \ref{lem:sumparts},
\begin{align*}
|B_n - B_{n,N}| &\le \frac{2C_1(\alpha)}{\lambda_{N,n}} = \frac{C_2(\alpha)T^2}{N^2+n^2T^2}.
\end{align*}
This gives
\begin{align*}
\sum_{0<|n|\le N}|B_n-B_{n,N}| &\le C_2(\alpha)T^2\sum_{0<|n|\le N} \frac{1}{N^2+n^2T^2}\le \frac{C_3(\alpha)T^2}{N}. 
\end{align*}
Also, since $\lambda_{m,0}$ grows quadratically in $m$, $B_0^* - B_{0,N}^*\to 0$ as $N\to\infty$. This completes the proof of the claim \eqref{eq:blimit}. So, by equation \eqref{eq:gnexp}, it only remains to show that the limit
\begin{align}\label{eq:bnlim}
\lim_{N\to\infty}\sum_{0<|n|\le N}e^{\i\beta n}B_{n}
\end{align}
exists and is finite. Since $\alpha \in [-\pi,\pi]$ and $\alpha \ne 0$, Lemma \ref{lem:1d-lattice-identity} implies that for $|n|$ so large that $|n|T\ge1$, 
\[
  0\le B_n\le \frac{CT}{|n|}e^{-|\alpha| |n| T},
\]
for some universal constant $C$. 
Clearly, this proves that the limit displayed in equation~\eqref{eq:bnlim} exists and is finite. This completes the proof when $\alpha \ne 0$. The same argument works when $\beta \ne 0$, upon swapping $m,n$ everywhere.
\end{proof}
Having obtained $g_T$, we define, for distinct points $(t,x),(s,y)\in M_T$,
\begin{align}\label{eq:gtdefn}
  G_T((t,x),(s,y)) := g_T(t-s, x-y).
\end{align}
This is our candidate for the Green's function on $M_T$. Note that we have only proved its existence; we still have to show that it satisfies the criterion \eqref{eq:gtdefinition}. This will have to wait until the next subsection.

\subsubsection{Some properties of the Green's function}
Our goal in this subsection is to prove some upper and lower bounds for $G_T$, and use those to show that $G_T$ satisfies the criterion \eqref{eq:gtdefinition}. These estimates will also be useful for other purposes. 
First, the following lemma proves the positivity of the heat kernel on the torus, which we need as a technical input.

\begin{lmm}\label{lem:heat-kernel-positive}
For $t>0$ and $x,y\in \R$, define
\[
p_t(x,y):=\frac{1}{4T}\sum_{(m,n)\in\mathbb{Z}^2}\exp\biggl(-t\lambda_{m,n}+ \pi \i\frac{m}{T} x+ \pi \i ny\biggr).
\]
Then $p_t(x,y)$ has the alternative expression 
\[
  p_t(x,y) = \frac{1}{4\pi t}\sum_{k,l\in\Z}
\exp\biggl(-\frac{(x-2Tk)^2+(y-2l)^2}{4t}\biggr).
\]
In particular, $p_t(x,y)\ge 0$ for all $t>0$ and $x,y\in \R$. 
\end{lmm}

\begin{proof}
Fix $t>0$ and set
\begin{align*}
p_t^{(1)}(x)&:=\frac{1}{2T}\sum_{m\in\mathbb{Z}}\exp\biggl( -\frac{t\pi^2 m^2}{T^2} + \pi\I \frac{m}{T}x\biggr),\\ 
p_t^{(2)}(y)&:=\frac{1}{2}\sum_{n\in\mathbb{Z}}\exp(-t\pi^2 n^2+\pi\I ny).
\end{align*}
Then $p_t(x,y)=p_t^{(1)}(x)\,p_t^{(2)}(y)$.
We will prove a formula for $p_t^{(1)}(x)$, and the corresponding formula for $p_t^{(2)}(y)$ is identical with $T$ replaced by $1$. To that end, consider on $\mathbb{R}$ the function $f(x):=\exp(-t\pi^2 x^2/T^2)$, whose Fourier transform is
\[
\widehat f(\xi)=\int_{\mathbb{R}} f(x)e^{-2\pi \I x\xi}\,\dd x = \frac{T}{\sqrt{\pi t}}\exp\biggl(-\frac{T^2\xi^2}{t}\biggr).
\]
Applying the Poisson summation formula~\cite[Theorem~3.2.8]{grafakos_classical_2014} to the function
$g(x):=f(x)e^{2\pi\I x\xi}$, whose Fourier transform is
$\widehat g(\eta)=\widehat f(\eta-\xi)$, we get
$$\sum_{m\in\mathbb{Z}}f(m)e^{2\pi\I m\xi}=\sum_{k\in\mathbb{Z}}\widehat f(k-\xi).$$
Taking $\xi=x/2T$ and dividing both sides by $2T$ gives
\[
p_t^{(1)}(x) = \frac{1}{2T}\sum_{m\in\mathbb{Z}}\exp\biggl( -\frac{t\pi^2 m^2}{T^2}+\pi\I \frac{m}{T}x\biggr)
=\frac{1}{2\sqrt{\pi t}}\sum_{k\in\mathbb{Z}}\exp\biggl(-\frac{(x-2Tk)^2}{4t}\biggr),
\]
which completes the proof.
\end{proof}

We now use the above lemma to get a uniform lower bound for $g_{T,N}$ (and hence, $g_T$). In the proof of this lemma as well as in the remainder of this subsection, we will use the notation $X\lesssim Y$ to mean that $X\le CY$ for some constant $C$ that may depend only on $T$. 
\begin{lmm}\label{lem:gN-lower}
There exists $C>0$, depending only on $T$, such that
\[
g_{T,N}(t,x)\ge -C\quad\text{for all }N\ge1,\ t,x\in \R.
\]
Consequently, the same lower bound holds for $g_T$.
\end{lmm}

\begin{proof}
Clearly, it suffices to work with $(t,x)\in M_T$. Define 
\begin{align}\label{eq:alphabetadef}
\alpha:=\pi\frac{t}{T},\quad
\beta:=\pi x.
\end{align}
By pairing $(m,n)$ with $(-m,-n)$, we get
\begin{equation}\label{eq:cos-form}
g_{T,N}(t,x)=\frac{\pi}{2T}\sum_{\substack{|m|\le N,\ |n|\le N\\(m,n)\neq(0,0)}}
\frac{1}{\lambda_{m,n}}\cos(\alpha m+\beta n).
\end{equation}
For $\epsilon>0$, define the Abel-regularized kernel
\[
g^\epsilon(t,x):=\frac{\pi}{2T}\sum_{(m,n)\in\mathbb{Z}^2\setminus\{(0,0)\}}
\frac{e^{-\epsilon\lambda_{m,n}}}{\lambda_{m,n}}\,
e^{\i  (\alpha m+\beta n)}.
\]
This series converges absolutely for each $\epsilon>0$. Let $p_s$ be the heat kernel defined in Lemma~\ref{lem:heat-kernel-positive}. By the identity
$\lambda^{-1}e^{-\epsilon \lambda}=\int_\epsilon^\infty e^{-s\lambda}\,ds$ and Fubini (justified by absolute
convergence when $\epsilon>0$), we have
\begin{align}\label{eq:gepsilonform}
g^\epsilon(t,x)
&=\frac{\pi}{2T}\sum_{(m,n)\neq(0,0)}\frac{e^{-\epsilon\lambda_{m,n}}}{\lambda_{m,n}}\,
e^{\i(\alpha m+\beta n)} \notag \\
&=\frac{\pi}{2T}\sum_{(m,n)\neq(0,0)}e^{\i(\alpha m+\beta n)}
\int_\epsilon^\infty e^{-s\lambda_{m,n}}\,\dd s \notag \\
&=\int_\epsilon^\infty \biggl(\frac{\pi}{2T}\sum_{(m,n)\neq(0,0)}
e^{-s\lambda_{m,n}}e^{\i(\alpha m+\beta n)}\biggr)\, \dd s \notag \\
&=2\pi \int_\epsilon^\infty\biggl(p_s(t,x)-\frac{1}{4T}\biggr)\,\dd s.
\end{align}
Now recall that by Lemma~\ref{lem:heat-kernel-positive}, $p_s$ is a nonnegative function. Thus,
\[
\int_\epsilon^1\biggl(p_s-\frac{1}{4T}\biggr)\,\dd s\ge -\frac{1-\epsilon}{4T}\ge -\frac{1}{4T}.
\]
Next, let $\lambda_\star:=\min\{\pi^2/T^2,\pi^2\}$. Then note that for $s\ge1$,
\begin{align}\label{eq:pttail}
\biggl|p_s(t,x)-\frac{1}{4T}\biggr|
&=\frac{1}{4T}\biggl|\sum_{(m,n)\neq(0,0)}e^{-s\lambda_{m,n}}e^{\i(\alpha m+\beta n)}\biggr|\notag\\ 
&\le \frac{1}{4T}\sum_{(m,n)\neq(0,0)}e^{-s\lambda_{m,n}}
\le C_1 e^{-(s-1)\lambda_\star},
\end{align}
where
\[
C_1:=\frac{1}{4T}\sum_{(m,n)\neq(0,0)}e^{-\lambda_{m,n}}<\infty.
\]
Therefore,
\[
\int_1^\infty\biggl(p_s-\frac{1}{4T}\biggr)\,\dd s\ge -\int_1^\infty C_1 e^{-(s-1)\lambda_\star}\,\dd s
= -\frac{C_1}{\lambda_\star}.
\]
Combining with the above estimates yields the uniform lower bound
\begin{equation}\label{eq:abel-lower}
g^\epsilon(t,x)\ge -2\pi \biggl(\frac{1}{4T}+\frac{C_1}{\lambda_\star}\biggr)=:-C_{\mathrm{Abel}}
\quad\text{for all }\epsilon>0,\ (t,x)\in M_T.
\end{equation}
The next step is to compare $g_{T,N}$ with $g^{\epsilon_N}$, where $\epsilon_N:=(N+1)^{-2}$. Note that 
\[
g_{T,N}(t,x)-g^{\epsilon_N}(t,x)
=\frac{\pi}{2T}\sum_{(m,n)\neq(0,0)}
\frac{1_{\{|m|\le N,|n|\le N\}}-e^{-\epsilon_N\lambda_{m,n}}}{\lambda_{m,n}}
e^{\i(\alpha m+\beta n)},
\]
where $1_{\{|m|\le N,|n|\le N\}}$ is $1$ if $|m|\le N$ and $|n|\le N$, and $0$ otherwise. 
Thus, 
\[
|g_{T,N}(t,x)-g^{\epsilon_N}(t,x)|
\le \frac{\pi}{2T}\sum_{(m,n)\neq(0,0)}\frac{\big|1_{\{|m|\le N,|n|\le N\}}-e^{-\epsilon_N\lambda_{m,n}}\big|}{\lambda_{m,n}}
=:R_N.
\]
Split $R_N=R_N^{\mathrm{in}}+R_N^{\mathrm{out}}$ into contributions from the square
$\{|m|\le N,\, |n|\le N\}$ and its complement. Take any $(m,n)\ne 0$ inside the square. Using  $|1-e^{-x}|\le x$, we get
\[
\frac{|1-e^{-\epsilon_N\lambda_{m,n}}|}{\lambda_{m,n}}\le \epsilon_N.
\]
There are at most $(2N+1)^2$ such terms. Hence,
\[
R_N^{\mathrm{in}}\le \frac{\pi(2N+1)^2}{2T}\epsilon_N\le \frac{2\pi}{T}.
\]
Next, take $(m,n)$ outside the square. Then $m^2+n^2\ge (N+1)^2$. Moreover,
\[
\lambda_{m,n}\ge \kappa(m^2+n^2),\qquad
\kappa:=\frac{\pi^2}{\max\{T^2,1\}}.
\]
This gives
\[
R_N^{\mathrm{out}}
\le \frac{\pi}{2T\kappa}\sum_{\text{outside}}\frac{e^{-\epsilon_N\kappa(m^2+n^2)}}{m^2+n^2}.
\]
Comparing the lattice sum to a polar-coordinate integral, we get
\[
\sum_{\text{outside}}\frac{e^{-\epsilon_N\kappa(m^2+n^2)}}{m^2+n^2}
\lesssim 
\int_{r\ge N+1} \frac{e^{-\epsilon_N\kappa r^2}}{r^2}r\,\dd r
=\int_{s\ge 1} e^{-\kappa s^2}\, \frac{\dd s}{s}<\infty.
\]
This gives $R_N^{\mathrm{out}}\lesssim 1$. Since we have already established that $R_N^{\mathrm{in}}\lesssim 1$, we get 
\begin{align}\label{eq:rnbound}
R_N \le C_{\mathrm{compare}}
\end{align} 
for some positive constant $C_{\mathrm{compare}}$ that depends only on $T$.  Combining this with equation~\eqref{eq:abel-lower}, we obtain
\[
g_{T,N}(t,x)\ge g^{\epsilon_N}(t,x)-C_{\mathrm{compare}}\ge -C_{\mathrm{Abel}}-C_{\mathrm{compare}}=:-C,
\]
which proves the claim.
\end{proof}

The last property of $g_T$ we need is a logarithmic upper bound on its absolute value.
\begin{lmm}\label{lem:gN-log-upper}
Take any $t,x\in \R$. Let
\[
d_T(t):=\min_{k\in\Z}|t-2Tk|,\quad d_1(x):=\min_{k\in\Z}|x-2k|,\quad
r:=\sqrt{d_T(t)^2+d_1(x)^2}.
\]
There exists $C>0$, depending only on $T$, such that for all $N\ge1$,
\[
|g_{T,N}(t,x)|\le C+|\log r|.
\]
Since the bound has no dependence on $N$, it holds also for $g_T(t,x)$.
\end{lmm}

\begin{proof}
In this proof, $C_0, C_1,\ldots$ will denote constants depending only on $T$, whose values may change from line to line. First, note that since $g_{T,N}(t,x)$ and $r$ are invariant under adding an integer multiple of $2T$ to $t$ and an even integer to $x$, we may assume without loss that $|t|\le T$ and $|x|\le 1$. 
As in the proof of Lemma \ref{lem:gN-lower}, write 
\[
g_{T,N}(t,x)=g^{\epsilon_N}(t,x)+(g_{T,N}(t,x)-g^{\epsilon_N}(t,x)),
\qquad \epsilon_N:=(N+1)^{-2}.
\]
By the inequality \eqref{eq:rnbound}, the comparison term $|g_{T,N}-g^{\epsilon_N}|$ is uniformly bounded in $N$. Thus, it suffices to bound $|g^{\epsilon_N}(t,x)|$ from above.
Using the representation \eqref{eq:gepsilonform} for $g^\epsilon(t,x)$ 
and the alternative expression for $p_s(t,x)$ from Lemma~\ref{lem:heat-kernel-positive}, we have
\[
g^\epsilon(t,x)=2\pi \int_{\epsilon}^\infty \biggl\{\frac{1}{4\pi s}\sum_{k,l\in\Z}
\exp\biggl(-\frac{(t-2Tk)^2+(x-2l)^2}{4s}\biggr) - \frac{1}{4T}\biggr\} \, \dd s.
\]
First, we get an upper bound on the absolute value of the integral from $\epsilon$ to $1$. 
The $(k,l)=(0,0)$ term yields
\begin{align*}
\int_\epsilon^1 \frac{1}{4\pi s}\exp\biggl(-\frac{r^2}{4s}\biggr)\, \dd s &=\int_{\epsilon/r^2}^{1/r^2}\frac{1}{4\pi q} e^{-1/(4q)}\, \dd q\\
&\le \int_{0}^{1/r^2}\frac{1}{4\pi q} e^{-1/(4q)}\, \dd q\\ 
&\le \int_0^1 \frac{1}{4\pi q} e^{-1/(4q)} \, \dd q + \int_1^{\max\{1,1/r^2\}} \frac{1}{4\pi q} \, \dd q\\ 
&\le C+ \frac{|\log r|}{2\pi}.
\end{align*}
Next, note that since $t\in [-T,T]$ and $x\in [-1,1]$, we have 
\[ 
  (t-2Tk)^2+(x-2l)^2\gtrsim k^2+l^2
\]
when $(k,l)\ne(0,0)$.  Therefore, for $s\in(0,1]$,
\begin{align*}
\sum_{(k,l)\ne(0,0)}\exp\biggl(-\frac{(t-2Tk)^2+(x-2l)^2}{4s}\biggr) &\le \sum_{(k,l)\ne(0,0)} e^{-C_0(k^2+l^2)/s}\\ 
&\lesssim \int_1^\infty e^{-C_1\rho^2/s}\rho \,\dd \rho = \frac{s}{2C_1} e^{-C_1/s} \lesssim s.
\end{align*}
This gives a uniform bound on the contribution of the nonzero lattice terms to the integral over
$s\in(0,1]$. Lastly, by the inequality \eqref{eq:pttail}, $\int_1^\infty|p_s(t,x)-1/(4T)|\, \dd s\lesssim 1$. This completes the proof.
\end{proof}

We are now ready to prove that the function $G_T$ defined in equation \eqref{eq:gtdefn} is indeed the mean-zero Green's function on the torus $M_T$. 
\begin{prop}\label{prop:green}
The function $G_T$ defined above satisfies the criterion \eqref{eq:gtdefinition} for all $f\in C^\infty(M_T)$, in the sense that $G_T f$ is a well-defined smooth function for any $f\in C^\infty(M_T)$, and the two identities in equation \eqref{eq:gtdefinition} hold.
\end{prop}
\begin{proof}
Take any $f\in C^\infty(M_T)$. Define 
\[
G_{T,N} f(u):=\int_{M_T} G_{T,N}(u,v)f(v)\,\dd v.  
\] 
It is easy to see that $G_{T,N}f$ is well-defined and smooth. By the bound from Lemma \ref{lem:gN-log-upper} and dominated convergence,
\[
  \lim_{N\to \infty} G_{T,N} f(u)=\int_{M_T} G_T(u,v)f(v)\,\dd v =: G_T f(u).
\]
This shows, in particular, that $G_T$ is well-defined as an integral operator on smooth functions. (We have not yet shown that $G_T f$ is smooth; this will be shown later in the proof.) Now, let $h:=\Delta f$.  We will first show that 
\begin{align}\label{eq:gtshow1}
-\frac{1}{2\pi}G_T h = f-c(f).
\end{align}
To prove this, take $u=(t,x)\in M_T$. First note that 
\begin{align*}
G_{T,N} h(u) &= \frac{\pi}{2T}\sum_{\substack{|m|\le N,\ |n|\le N\\ (m,n)\neq(0,0)}}
\frac{1}
{\lambda_{m,n}}\int_{M_T}\exp\biggl(\pi \mathrm{i} \frac{m}{T}(t-s)+\pi \mathrm{i} n (x-y)\biggr) h (s,y)\, \dd s \, \dd y.
\end{align*}
Integration by parts gives 
\begin{align*}
&\int_{M_T}\exp\biggl(\pi \mathrm{i} \frac{m}{T}(t-s)+\pi \mathrm{i} n (x-y)\biggr) h (s,y)\, \dd s \, \dd y\\ 
&= \int_{M_T}\Delta\biggl\{\exp\biggl(\pi \mathrm{i} \frac{m}{T}(t-s)+\pi \mathrm{i} n (x-y)\biggr)\biggr\}f (s,y)\, \dd s \, \dd y\\ 
&= -\lambda_{m,n} \int_{M_T}\exp\biggl(\pi \mathrm{i} \frac{m}{T}(t-s)+\pi \mathrm{i} n (x-y)\biggr)f (s,y)\, \dd s \, \dd y.
\end{align*}
Combining the last two displays, we get
\begin{align}\label{eq:gtexpansion}
G_{T,N} h(u) &= -\frac{\pi}{2T}\sum_{\substack{|m|\le N,\ |n|\le N\\ (m,n)\neq(0,0)}} \int_{M_T}\exp\biggl(\pi \mathrm{i} \frac{m}{T}(t-s)+\pi \mathrm{i} n (x-y)\biggr) f (s,y)\, \dd s \, \dd y.
\end{align}
Let $\bt^2$ denote the unit torus $[0,1]^2$, and define $f_0:\bt^2\to \R$ as
\begin{align*}
f_0(t,x) := f(2Tt - T, 2x-1).
\end{align*}
For $m,n\in \Z$,
\[
  \hat{f}_0(m,n) := \int_{\bt^2} f_0(t,x) e^{-2\pi \i (mt + nx)} \, \dd t\, \dd x
\]
denote the Fourier coefficients of $f_0$. Since $f_0$ is smooth, standard results~\cite[Theorem~3.3.9(a)]{grafakos_classical_2014} imply that $\hat{f}_0(m,n)$ is rapidly decaying in $|m|+|n|$. From this, it follows~\cite[Theorem~3.2.5]{grafakos_classical_2014} that for almost every $(t,x)\in \bt^2$,
\[
f_0(t,x) = \sum_{m,n\in \Z} \hat{f}_0(m,n) e^{2\pi \i (mt+nx)}  
\] 
But the rapid decay of the Fourier coefficients implies that the right side above is a continuous function of $(t,x)$, and the left side is continuous anyway. Thus, the identity must hold for all $(t,x)$. 

Translating this identity to the torus $M_T$, we get, for $(t,x)\in M_T$,
\begin{align}\label{eq:ff0}
f(t,x) &= f_0\biggl(\frac{t+T}{2T}, \frac{x+1}{2}\biggr)\notag \\ 
&= \sum_{m,n\in \Z} \hat{f}_0(m,n) \exp\biggl\{2\pi \i \biggl(m \frac{t+T}{2T} +n \frac{x+1}{2}\biggr)\biggr\}\notag \\ 
&= \sum_{m,n\in \Z} \hat{f}_0(m,n)e^{\pi \i (m+n)} \exp\biggl(\pi \i m \frac{t}{T} +\pi \i n x\biggr).
\end{align}
This, together with another application of the rapid decay of the Fourier coefficients, gives 
\begin{align*}
&\int_{M_T}\exp\biggl(\pi \mathrm{i} \frac{m}{T}(t-s)+\pi \mathrm{i} n (x-y)\biggr) f (s,y)\, \dd s \, \dd y\\ 
&= \sum_{m',n'\in \Z}\hat{f}_0(m',n')e^{\pi\i(m'+n')}\exp\biggl(\pi \mathrm{i} \frac{m'}{T}t+\pi \mathrm{i} n'  x\biggr)\\ 
&\qquad \qquad \cdot \int_{M_T}\exp\biggl(\pi \i \frac{m-m'}{T} s + \pi \i (n-n')y\biggr) \, \dd s\, \dd y\\ 
&= 4T\hat{f}_0(m,n)e^{\pi\i(m+n)}\exp\biggl(\pi \mathrm{i} \frac{m}{T}t+\pi \mathrm{i} n  x\biggr).
\end{align*}
Substituting this in equation \eqref{eq:gtexpansion}, we get 
\begin{align*}
G_{T,N}h(u) &=-2\pi\sum_{\substack{|m|\le N,\ |n|\le N\\ (m,n)\neq(0,0)}}\hat{f}_0(m,n)e^{\pi\i(m+n)}\exp\biggl(\pi \mathrm{i} \frac{m}{T}t+\pi \mathrm{i} n  x\biggr).
\end{align*}
Again, the rapid decay of the Fourier coefficients allows us to take $N\to\infty$ on the right. Using equation \eqref{eq:ff0}, this gives 
\begin{align*}
G_{T} h(u) &=-2\pi\sum_{(m,n)\neq(0,0)}\hat{f}_0(m,n)e^{\pi\i(m+n)}\exp\biggl(\pi \mathrm{i} \frac{m}{T}t+\pi \mathrm{i} n  x\biggr)\\ 
&= -2\pi f(u) + 2\pi \hat{f}_0(0,0) = -2\pi(f(u)-c(f)).
\end{align*}
This proves the claim \eqref{eq:gtshow1}. 

Next, let $g := G_T f$. We claim that $g$ is smooth, and 
\begin{align}\label{eq:gtshow2}
-\frac{1}{2\pi}\Delta g = f - c(f).
\end{align}
To see that $g$ is smooth, note that we can write
\[
  g(u) = \int_{M_T} f(u-v) g_T(v) \, \dd v.
\]
By the bound from Lemma \ref{lem:gN-log-upper} and dominated convergence, we can differentiate the right side as many times as we want, by moving the derivatives inside the integral. This shows that $g$ is smooth. Next, take any $\varphi\in C^\infty(M_T)$. Again using Lemma \ref{lem:gN-log-upper} to apply Fubini, and using the observation that $g_T(-u)=g_T(u)$ (which follows from the same property for $g_{T,N}$), we get
\begin{align*}
\int_{M_T}\varphi(u)\Delta g(u)\,\dd u
&=\int_{M_T}g(u)\Delta \varphi(u)\,\dd u \\
&= \int_{M_T}\int_{M_T} f(v) g_T(u-v) \Delta \varphi(u) \,\dd v \, \dd u \\ 
&= \int_{M_T}\int_{M_T} f(v) g_T(v-u) \Delta \varphi(u) \,\dd u \, \dd v \\ 
&= \int_{M_T} f(v) G_T \Delta \varphi(v) \, \dd v.
\end{align*}
But by equation \eqref{eq:gtshow1}, $G_T\Delta \varphi(v) = -2\pi(\varphi(v)-c(\varphi))$.  Thus,
\[
  \int_{M_T}\varphi(u)\Delta g(u)\,\dd u = -2\pi \int_{M_T} f(v)(\varphi(v)-c(\varphi)) \, \dd v = -2\pi \int_{M_T} \varphi(v)(f(v)-c(f)) \, \dd v. 
\]
Rearranging this, we get
\[
\int_{M_T}\varphi(u)  \biggl(\frac{1}{2\pi}\Delta g(u)+f(u)-c(f) \biggr)\, \dd u =0.
\]
Since this holds for every smooth $\varphi$, we can take 
\[
  \varphi(u) = \frac{1}{2\pi}\Delta g(u)+f(u)-c(f).
\]
This choice gives 
\[
\int_{M_T}\biggl(\frac{1}{2\pi}\Delta g(u)+f(u)-c(f) \biggr)^2\, \dd u =0.
\]
Since $f$ and $\Delta g$ are smooth, this proves the claim \eqref{eq:gtshow2}.
\end{proof}

\subsubsection{Passing to infinite volume}\label{sec:infvol}
We now investigate what happens to $G_T$ as we send $T\to\infty$. We need the following basic lemma.
\begin{lmm}\label{lem:KT-explicit}
For $\theta\in [-2\pi,2\pi]$, we have
\begin{align*}
\frac{1}{\pi}\sum_{m\in\Z\setminus\{0\}}\frac{e^{\i m \theta }}{m^2}&=\frac{\pi}{3}-|\theta|+\frac{\theta^2}{2\pi}.
\end{align*}
\end{lmm}
\begin{proof}
Let $g(\theta)$ denote the left side. Pairing $m$ with $-m$, we get
\begin{align}\label{eq:ktuform}
g(\theta)=\frac{2}{\pi}\sum_{m=1}^\infty\frac{\cos(m\theta)}{m^2}.
\end{align}
Let $f(x):= x^2$ for $x\in [-\pi,\pi]$. The Fourier coefficients of $f$ are given by $\hat{f}(0)=\pi^2/3$ and for $m\ne 0$,
\begin{align*}
\hat{f}(m) &= \frac{1}{2\pi}\int_{-\pi}^\pi f(x) e^{\i m x} \dd x\\ 
&= \frac{1}{2\pi}\biggl[\frac{x^2e^{\i mx}}{\i m}\biggr]^\pi_{-\pi} - \frac{1}{\i\pi m}\int_{-\pi}^\pi x e^{\i m x} \dd x\\ 
&= -\frac{1}{\i \pi m}\biggl[\frac{xe^{\i m x}}{\i m}\biggr]^\pi_{-\pi} + \frac{1}{\i^2 \pi m^2}\int_{-\pi}^\pi e^{\i m x} \dd x\\ 
&= \frac{2\cos (m \pi)}{m^2} = \frac{2(-1)^m}{m^2}.
\end{align*}
Since $f\in L^1$ and $\sum_{m\in \Z} |\hat{f}(m)|<\infty$, the Fourier inversion formula~\cite[Section 9.4]{rudin87} shows that for almost all $x\in [-\pi,\pi]$,
\begin{align}\label{eq:xexpan}
x^2 = f(x) &= \sum_{m\in \Z} \hat{f}(m)e^{\i m x} = \frac{\pi^2}{3} + 4\sum_{m=1}^\infty \frac{(-1)^m \cos (mx)}{m^2}. 
\end{align}
The right side is absolutely convergent and is a continuous function of $x$. This shows that the above identity holds for all $x\in [-\pi,\pi]$. 

Replacing $x$ by $\theta-\pi$ in equation \eqref{eq:xexpan}, and using
$\cos(m(\theta-\pi))=(-1)^m\cos(m\theta)$, we get
\[
\sum_{m=1}^\infty\frac{\cos(m\theta)}{m^2}
=\frac{\pi^2}{6}-\frac{\pi\theta}{2}+\frac{\theta^2}{4},
\qquad \theta\in[0,2\pi].
\]
By evenness of the left side, this gives 
\[
\sum_{m=1}^\infty\frac{\cos(m\theta)}{m^2}
=\frac{\pi^2}{6}-\frac{\pi|\theta|}{2}+\frac{\theta^2}{4},
\qquad \theta\in[-2\pi,2\pi].
\]
Substituting this into equation \eqref{eq:ktuform} completes the proof.
\end{proof}

We are now ready to prove the pointwise convergence of $g_T(t,x)$ as $T\to\infty$, after taking off a linear multiple of $T$. We will divide the proof into two cases: $t\ne0$ and $t=0$.

\begin{lmm}\label{lem:gtlimit}
For each $(t,x)\in \R^2$ with $t\ne 0$, we have 
\[
  \lim_{T\to\infty}\biggl(g_T(t,x)-\frac{\pi T}{6}\biggr)=g(t,x) :=-\frac{\pi |t|}{2}-\log|1-e^{-\pi |t|+\pi \i x}|.
\]
\end{lmm}

\begin{proof}
In this proof, $C,C_1,\ldots$ will denote positive universal constants, whose values may change from line to line. 
Since $g_T$ is $2$-periodic in the second argument, it suffices to prove the claim for $(t,x)\in M\setminus\{(0,0)\}$. Let us fix such a $(t,x)$. As in the proof of Lemma~\ref{lem:gN-to-g}, let $\alpha:=\pi t/T$ and  $\beta:=\pi x$, and define
\[
B_0^*:=\sum_{m\in\Z\setminus\{0\}}\frac{e^{\i\alpha m}}{\lambda_{m,0}}, \qquad B_n:=\sum_{m\in\Z}\frac{e^{\i\alpha m}}{\lambda_{m,n}},\quad n\neq 0.
\]
The proof of Lemma~\ref{lem:gN-to-g} shows that
\begin{equation}\label{eq:gtlimit-g-fourier}
g_T(t,x)=\frac{\pi}{2T}\biggl(B_0^*+\sum_{n\in\Z\setminus\{0\}}e^{\i\beta n}B_n\biggr),
\end{equation}
and also that
\begin{align}\label{eq:bnbound}
0\le B_n\le \frac{CT}{|n|}e^{-|n||\alpha|T}=\frac{CT}{|n|}e^{-\pi|n||t|},
\qquad n\neq 0.
\end{align}
Since $t\ne 0$, this shows that the series in the preceding display is absolutely convergent. We first compute $B_0^*$. Since $\lambda_{m,0}=\pi^2m^2/T^2$, we have
\begin{equation}\label{eq:gtlimit-b0-pre}
\frac{\pi}{2T}B_0^*
=
\frac{T}{2\pi}\sum_{m\in\Z\setminus\{0\}}\frac{e^{\i m\pi t/T}}{m^2}.
\end{equation}
Whenever $|t|\le T$, we have $\pi t/T\in[-\pi,\pi]$, and therefore Lemma~\ref{lem:KT-explicit} gives
\begin{equation}\label{eq:gtlimit-b0-eval}
\frac{\pi}{2T}B_0^*
=
\frac{\pi T}{6}-\frac{\pi |t|}{2}+\frac{\pi t^2}{4T}.
\end{equation}
Next, for $n\neq 0$, note that 
\[
\frac{\pi}{2T}B_n=\frac{T}{2\pi}\sum_{m\in\Z}\frac{e^{\i\alpha m}}{m^2+n^2T^2}.
\]
Applying Lemma \ref{lem:1d-lattice-identity} with $a=|n|T$ and $\theta=\alpha=\pi t/T$, we get
\begin{equation}\label{eq:gtlimit-bn-eval}
\frac{\pi}{2T}B_n
=
\frac{1}{2|n|}\sum_{k\in\Z}e^{-\pi |n||t+2kT|}.
\end{equation}
Since $t\ne 0$, the bound \eqref{eq:bnbound} allows us to apply the dominated convergence theorem and get
\begin{align*}
\lim_{T\to\infty}\frac{\pi}{2T}\sum_{n\in\Z\setminus\{0\}}e^{\i\beta n}B_n
&= \sum_{n\in\Z\setminus\{0\}}e^{\i\beta n}\biggl(\lim_{T\to\infty}\frac{\pi}{2T}B_n\biggr)\notag \\ 
&=  \frac12\sum_{n\in\Z\setminus\{0\}}\frac{e^{\pi \i nx}}{|n|}e^{-\pi |n||t|},
\end{align*}
and the series on the right is absolutely convergent. By the power series expansion of logarithm at $1$, 
\begin{align*}
\frac12\sum_{n\in\Z\setminus\{0\}}\frac{e^{\pi \i nx}}{|n|}e^{-\pi |n||t|} &= \frac12\sum_{n=1}^\infty\frac{e^{\pi \i nx}}{n}e^{-\pi n|t|} + \frac12\sum_{n=1}^\infty\frac{e^{-\pi \i nx}}{n}e^{-\pi n|t|}\\ 
&= -\frac12\log(1-e^{\pi \i x - \pi |t|}) -\frac12\log(1-e^{-\pi \i x - \pi |t|})\\ 
&= -\log|1-e^{-\pi |t|+\pi \i x}|. 
\end{align*} 
Combining the steps, this completes the proof.
\end{proof}

Next, we show the same result for $t=0$. We need the following technical lemma.
\begin{lmm}\label{lem:antlmm}
For $n\ge 1$ and $T>0$, define
\[
a_n(T)
:=
\frac{T}{\pi}\sum_{m\in\mathbb{Z}}\frac{1}{m^2+n^2T^2}.
\]
Then we have the alternative expression
\[
  a_n(T) = \frac{1}{n}+ \frac{2}{n(e^{2\pi Tn}-1)}.
\]
\end{lmm}
\begin{proof}
Applying Lemma \ref{lem:1d-lattice-identity} with $\theta=0$ and $a=nT$, we get
\[
\sum_{m\in\Z}\frac{1}{m^2+n^2T^2}
=
\frac{\pi}{nT}\sum_{k\in\Z}e^{-2\pi nT|k|}.
\]
Therefore,
\begin{align*}
a_n(T) &= \frac{T}{\pi}\cdot \frac{\pi}{nT}\sum_{k\in\Z} e^{-2\pi nT|k|}\\ 
&= \frac{1}{n} + \frac{2}{n}\sum_{k=1}^\infty e^{-2\pi Tnk} = \frac{1}{n}+ \frac{2e^{-2\pi Tn}}{n(1-e^{-2\pi Tn})}.
\end{align*}
Rearranging the last term completes the proof.
\end{proof}
We are now ready to calculate the limit of $g_T(0,x)$ as $T\to\infty$.
\begin{lmm}\label{lem:gtlimit2}
For each $x\in \R\setminus 2\Z$,
\[
  \lim_{T\to\infty}\biggl(g_T(0,x)-\frac{\pi T}{6}\biggr)=g(0,x) :=-\log|1-e^{\pi \i x}|.
\]
\end{lmm}
\begin{proof}
Since $g_T(0,\cdot)$ and $x\mapsto -\log|1-e^{\pi\i x}|$ are both $2$-periodic, it suffices to prove the claim for $x\in[-1,1)\setminus\{0\}$. For fixed $T$, write
\[
g_{T,N}(0,x)=A_{T,N}+B_{T,N}(x),
\]
where
\[
A_{T,N}
=
\frac{T}{2\pi}\sum_{0<|m|\le N}\frac{1}{m^2}, \qquad B_{T,N}(x)
=
\frac{T}{2\pi}
\sum_{\substack{|m|\le N\\0< |n|\le N}}
\frac{e^{\pi \i n x}}{m^2+n^2T^2}.
\]
Since the denominator is even in \(n\), pairing \(n\) and \(-n\) gives
\[
B_{T,N}(x)
=
\frac{T}{\pi}
\sum_{n=1}^N
\cos(\pi n x)
\biggl(\sum_{|m|\le N}\frac{1}{m^2+n^2T^2}\biggr).
\]
Clearly,
\[
\lim_{N\to \infty} A_{T,N}=
\frac{T}{\pi}\sum_{m=1}^\infty \frac{1}{m^2} = \frac{\pi T}{6}.
\]
Thus it remains to understand $B_{T,N}(x)$. For $k\ge 1$, set
\[
C_k(x):=\sum_{n=1}^k \cos(\pi n x). 
\]
Note that for any $k$,
\begin{align*}
C_k(x) &= \frac{1}{2}\sum_{n=1}^k e^{\pi\i n x} + \frac{1}{2}\sum_{n=1}^k e^{-\pi\i n x}\\ 
&= \frac{e^{\pi\i x} - e^{\pi\i (k+1)x}}{2(1-e^{\pi\i x})} + \frac{e^{-\pi\i x} - e^{-\pi\i (k+1)x}}{2(1-e^{-\pi\i x})}.
\end{align*}
Since $x\in [-1,1)\setminus\{0\}$, we have $e^{\pm \pi\i x}\ne 1$. Therefore the above identity shows that
\[
M(x) := \sup_{k\ge 1} |C_k(x)|<\infty.  
\]
Let $a_n(T)$ be as in Lemma \ref{lem:antlmm}. 
Since $a_n(T)$ is positive and decreases to zero as $n\to \infty$, and $|C_k(x)|\le M(x)$ for all $k$, Lemma \ref{lem:sumparts} shows that the following function is well-defined:
\[
F_T(x)=\sum_{n=1}^\infty \cos(\pi n x)\,a_n(T) := \lim_{N\to\infty}\sum_{n=1}^N \cos(\pi n x)\,a_n(T).
\]
Now set
\begin{align*}
&b_{n,N}(T)
:=
\frac{T}{\pi}\sum_{|m|\le N}\frac{1}{m^2+n^2T^2},\\ 
&d_{n,N}(T):=a_n(T)-b_{n,N}(T)
=
\frac{T}{\pi}\sum_{|m|>N}\frac{1}{m^2+n^2T^2}.
\end{align*}
Then $d_{n,N}(T)\ge 0$, and $d_{n,N}(T)$ decreases in $n$. Moreover,
\[
0\le d_{1,N}(T)
=
\frac{T}{\pi}\sum_{|m|>N}\frac{1}{m^2+T^2}
\le
\frac{T}{\pi}\sum_{|m|>N}\frac{1}{m^2}
\xlongrightarrow{N\to \infty} 0.
\]
Since
\[
B_{T,N}(x)=\sum_{n=1}^N \cos(\pi n x)\,b_{n,N}(T),
\]
we have
\[
F_T(x)-B_{T,N}(x)
=
\sum_{n=1}^N \cos(\pi n x)\,d_{n,N}(T)
+
\sum_{n=N+1}^\infty \cos(\pi n x)\,a_n(T).
\]
Applying Lemma \ref{lem:sumparts} to each sum yields
\[
|F_T(x)-B_{T,N}(x)|
\le
2M(x)\,d_{1,N}(T)+2M(x)\,a_{N+1}(T).
\]
Now $d_{1,N}(T)\to 0$ as $N\to\infty$, and Lemma \ref{lem:antlmm} shows that $a_{N+1}(T)\to 0$ as $N\to\infty$. Hence,
\[
\lim_{N\to\infty} B_{T,N}(x)= F_T(x).
\]
We conclude that
\[
g_T(0,x) = \lim_{N\to\infty} g_{T,N}(0,x)=\frac{\pi T}{6}+F_T(x).
\]
Now define
\begin{align}\label{eq:sxdef}
S(x):=\sum_{n=1}^\infty \frac{\cos(\pi n x)}{n},
\end{align}
where the series on the right is conditionally convergent by Lemma \ref{lem:sumparts} and the uniform boundedness of $|C_k(x)|$. 
For fixed $T$, set
\[
c_n(T)
:= a_n(T) - \frac{1}{n} = \frac{2}{n(e^{2\pi nT}-1)}.
\]
Then $c_n(T)\ge 0$ and decreases to zero as $n\to\infty$. Therefore, again by
Lemma \ref{lem:sumparts},
\[
\biggl|
\biggl(\lim_{N\to\infty} g_{T,N}(0,x)-\frac{\pi T}{6}\biggr)-S(x)
\biggr|
=
\biggl|\sum_{n=1}^\infty \cos(\pi n x)\,c_n(T)\biggr|
\le
2M(x)\,c_1(T).
\]
But
\[
c_1(T)=\frac{2}{e^{2\pi T}-1}\xlongrightarrow{T\to \infty} 0.
\]
Hence,
\begin{align}\label{eq:prevlim}
\lim_{T\to\infty}
\biggl(g_{T}(0,x)-\frac{\pi T}{6}
\biggr)
=
S(x).
\end{align}
It remains to evaluate $S(x)$. For $r\in (0,1)$, define
\[
  S_r(x) := \sum_{n=1}^\infty \frac{r^n\cos(\pi n x)}{n},
\]
which is an absolutely convergent series. Note that 
\begin{align*}
S_r(x) &= \Re\biggl(\sum_{n=1}^\infty \frac{r^n e^{\pi \i nx}}{n}\biggr)\\ 
&= \Re(-\log(1-r e^{\pi \i x})) = -\log |1-re^{\pi \i x}|.
\end{align*}
Now, define the partial sum
\begin{align*}
U_n(x) := \sum_{j=1}^n \frac{\cos (\pi j x)}{j}.
\end{align*}
By Lemma \ref{lem:sumparts},  $U_n(x)\to S(x)$ as $n\to\infty$. In particular, the sequence is bounded by some constant $C$ in absolute value. Thus, summation by parts gives 
\begin{align*}
S_r(x) &= \sum_{n=1}^\infty r^n(U_n(x)-U_{n-1}(x)) \\ 
&= \sum_{n=1}^\infty U_n(x)(r^n- r^{n+1}) = (1-r)\sum_{n=1}^\infty U_n(x) r^n.
\end{align*}
Take any $\ep>0$ and find $N$ so large that $|U_n(x)-S(x)|\le \ep$ for all $n\ge N$. Then 
\begin{align*}
|S_r(x)-S(x)| &\le (1-r)\sum_{n=1}^N 2C r^n + (1-r)\sum_{n=N+1}^\infty \ep r^n\\ 
&\le 2CN(1-r) + \ep. 
\end{align*}
This shows that 
\[
  \limsup_{r\to 1}|S_r(x)-S(x)| \le \ep.
\]
Since $\ep$ is arbitrary, this shows that $S_r(x)\to S(x)$ as $r\to 1$.
Since $x\in[-1,1)\setminus\{0\}$, we have $e^{\pi\i x}\neq 1$. Hence
\begin{align}\label{eq:sxform}
  S(x)=\lim_{r\uparrow 1}S_r(x)= -\lim_{r\uparrow 1}\log |1-re^{\pi \i x}| = -\log|1-e^{\pi \i x}|.
\end{align}
Combining with equation~\eqref{eq:prevlim} gives the claim.
\end{proof}
This concludes the appendix discussion of the torus Green's function. We next record several technical lemmas used elsewhere in the paper.

\subsection{Auxiliary lemmas}\label{sec:aux}
In this section we collect technical proofs omitted from the main text, together with a few supporting lemmas used in those arguments.

\subsubsection{Proof of Lemma~\ref{lem:denominator-evaluation-integer} (Passing to the zero-mode limit)}\label{sec:aux-zero-mode-limit}
Clearly, the claim follows once we justify moving the limit $\lim_{\ep \to 0}$ inside the infinite sum in the formula for $\inn{F}_{T,N,\ep}$ in Lemma~\ref{lem:denominator-evaluation}. To this end, we first record that $G_{T,N}$ is positive semidefinite.
Indeed, by the spectral representation \eqref{eq:gtnrep}, we see that for any finite family $(u_l)_{1\le l\le n}\subseteq  M_T$ and $(c_l)_{1\le l\le n}\subseteq \rr$,
\[
\sum_{l,l'=1}^n c_l c_{l'} G_{T,N}(u_l,u_{l'})
=\sum_{j=1}^{L_N} \frac{1}{\lambda_j}\biggl(\sum_{l=1}^n c_l e_j(u_l)\biggr)^2\ge 0.
\]
In particular, taking $c_l\equiv 1$ gives
\begin{align*}
\sum_{1\le l<l'\le n} G_{T,N}(u_l,u_{l'})
&=\frac12\sum_{l,l'=1}^n G_{T,N}(u_l,u_{l'})-\frac12\sum_{l=1}^n G_{T,N}(u_l,u_l) \\ 
&\ge -\frac12\sum_{l=1}^n G_{T,N}(u_l,u_l).
\end{align*}
Since $G_{T,N}$ is a bounded function, this shows that the real part of the
exponent satisfies
\[
\Re\biggl(- 4b \sum_{j=1}^k \sum_{l=1}^n\alpha_j G_{T,N}(q_j,u_l) - 4b^2 \sum_{1\le l<l'\le n} G_{T,N}(u_l, u_{l'})\biggr) \le Cn
\]
for some constant $C$ that does not depend on $n$ or $u_1,\ldots,u_n$. Thus the
$n$th integral is bounded by $C_1^n$ for a fixed constant $C_1$, and the factor
$1/n!$ gives a summable majorant. This allows us to use the dominated
convergence theorem to pass the limit $\ep\to 0$ inside the sum.

\subsubsection{Bounds required for the proof of Lemma \ref{lem:bnk}}\label{sec:bnkbounds}
For each $T>0$ and $r\ge 0$, define
\begin{align}\label{eq:ftdef}
  f_T(r):= \frac{\pi}{2} r - \frac{\pi}{4T}r^2.
\end{align}
We will now prove some key properties of $f_T$.
\begin{lmm}\label{lem:ftlemma}
The function $f_T$ is concave on $[0,\infty)$, is nonnegative and increasing on $[0,T]$, and has Lipschitz constant at most $\pi/2$ on $[0,T]$. Moreover, for $r\in [0,T]$, it satisfies the bounds
\[
  \frac{\pi}{4}r \le f_T(r)\le \frac{\pi}{2}r,
\]
and for $a,b\in [0,T]$, it satisfies the subadditive inequality
\[
  f_T(\min\{a+b,T\})\le f_T(a)+f_T(b).
\]
\end{lmm}
\begin{proof}
To see that $f_T$ is concave on $[0,\infty)$, simply observe that its second derivative is nonpositive. Nonnegativity on $[0,T]$ is simple to check (in fact, it is nonnegative on $[0,2T]$), and $f_T$ is increasing on $[0,T]$ because $f_T'$ is nonnegative on $[0,T]$. The bound on the Lipschitz constant follows from the maximum value of $|f_T'|$ in $[0,T]$. The upper and lower bounds on $f_T(r)$ are straightforward from the formula for $f_T$. Lastly, for the subadditive inequality, take any $a,b\in [0,T]$. Let $c:= \min\{a+b,T\}-a$. Note that $c\le b$. Thus, if we can show that 
\begin{align}\label{eq:fttoshow}
f_T(a+c)\le f_T(a)+f_T(c),
\end{align} 
then by the increasing nature of $f_T$ on $[0,T]$, it will follow that $f_T(a+c)\le f_T(a)+f_T(b)$, which is the required inequality. So, let us prove the inequality \eqref{eq:fttoshow}. If $a+c=0$, then the inequality is trivial since $f_T(0)=0$. So, let us assume that $a+c>0$. Let 
\[
  \lambda := \frac{a}{a+c},
\]
so that 
\[
  a = \lambda (a+c), \quad c = (1-\lambda)(a+c).
\]
By the concavity of $f_T$, we get
\begin{align*}
&f_T(a)\ge \lambda f_T(a+c) + (1-\lambda)f_T(0) = \frac{a}{a+c}f_T(a+c),\\ 
&f_T(c)\ge (1-\lambda)f_T(a+c)+\lambda f_T(0) = \frac{c}{a+c} f_T(a+c).
\end{align*}
Adding up the above inequalities gives \eqref{eq:fttoshow}.
\end{proof}
For the next lemma, we need the following technical tool.
\begin{lmm}\label{lem:one-minus-exp-lower}
There exists a universal constant $C\ge1$ such that for all $u\ge0$ and all
$v\in[-\pi,\pi]$,
\[
  C^{-1}\min\{1,\sqrt{u^2+v^2}\}
  \le |1-e^{-u+\i v}|
  \le C\min\{1,\sqrt{u^2+v^2}\}.
\]
\end{lmm}
\begin{proof}
Throughout this proof, $c_1,c_2,\ldots$ will denote positive universal constants whose values may change from line to line. 
Using
\[
|1-e^{-u+\i v}|^2=(1-e^{-u})^2+2e^{-u}(1-\cos v),
\]
consider two cases. If $u^2+v^2\le 1$, then $u\le 1$ and $|v|\le 1$. Hence
\[
1-e^{-u}\ge c_1u,\qquad 1-\cos v\ge c_2v^2,\qquad e^{-u}\ge e^{-1},
\]
which implies
\[
|1-e^{-u+\i v}|^2\ge c_3(u^2+v^2).
\]
If $u^2+v^2>1$, then either $u\ge \frac12$ or $u<\frac12$.
If $u\ge \frac12$, then
\[
|1-e^{-u+\i v}|\ge 1-e^{-u}\ge 1-e^{-1/2}.
\]
If $u<\frac12$, then $|v|\ge \sqrt{1-u^2}>\frac{\sqrt3}{2}$ and $|v|\le \pi$, and so,
\[
|1-e^{-u+\i v}|^2\ge 2e^{-u}(1-\cos v)
\ge 2e^{-1/2}(1-\cos(\sqrt3/2)).
\]
Thus in either subcase, $|1-e^{-u+\i v}|\ge c_4$. This proves the lower bound.
For the upper bound, first use $|e^{-u}|\le1$ to get
\[
  |1-e^{-u+\i v}|\le 1+e^{-u}\le2.
\]
Also,
\[
  |1-e^{-u+\i v}|
  \le |1-e^{-u}|+e^{-u}|1-e^{\i v}|
  \le u+|v|,
\]
where we used $1-e^{-u}\le u$ and $|1-e^{\i v}|\le |v|$. Since
$u+|v|\le\sqrt2\sqrt{u^2+v^2}$, the upper bound follows after increasing the
constant $C$.
\end{proof}

Now define a function $h_T:\R^2 \to \R$ as
\begin{align}\label{eq:htdef}
  h_T(t,x)=\sum_{n=1}^\infty \frac{\cos(\pi n x)}{n}
  \biggl(\sum_{m\in\Z}e^{-\pi n|t+2mT|}\biggr).
\end{align}
Also, for $(t,x)\in \R^2$, let
\[
  r_T(t,x):=\sqrt{d_T(t)^2+d_1(x)^2},
\]
with $d_T,d_1$ as in Lemma~\ref{lem:gN-log-upper}. It is a simple fact that for $(t,x), (s,y)\in M_T$, $r_T(t-s, x-y)$ is the geodesic distance between $(t,x)$ and $(s,y)$ under the flat metric. The following lemma gives upper and lower bounds on $h_T(t,x)$ depending on the geodesic distance between the image of $(t,x)$ on $M_T$ and the origin.
\begin{lmm}\label{lem:htlemma}
Take any $T\ge 1$. There are positive universal constants $C_1$ and $C_2$ such that for all $(t,x)\in \R^2$, we have
\[
  -C_1 \le h_T(t,x)\le C_2 + |\log r_T(t,x)|.
\]
\end{lmm}
\begin{proof}
Throughout this proof, $C,C_1,\ldots$ will denote positive universal constants, whose values may change from line to line. Since both $h_T$ and $r_T$ are $2T$-periodic in the first coordinate and $2$-periodic in the second coordinate, we may assume without loss that $(t,x)\in M_T$. Write
\[
  h_T(t,x)=h_T^{(0)}(t,x)+R_T(t,x),
\]
where
\[
  h_T^{(0)}(t,x):=\sum_{n=1}^\infty \frac{\cos(\pi n x)}{n}e^{-\pi n|t|}
  =-\log|1-e^{-\pi |t|+\pi \i x}|,
\]
and
\[
  R_T(t,x):=\sum_{n=1}^\infty \frac{\cos(\pi n x)}{n}
  \biggl(\sum_{m\in\Z\setminus\{0\}}e^{-\pi n|t+2mT|}\biggr).
\]
Since $|t|\le T$, we have
\[
  |t+2mT|\ge (2|m|-1)T,\qquad m\neq 0.
\]
Since $T\ge 1$, this implies that 
\begin{align}\label{eq:rtbound}
  |R_T(t,x)|
  \le \sum_{n=1}^\infty \frac{1}{n}\biggl(\sum_{m\neq 0}e^{-\pi n|t+2mT|}\biggr)
  \le 2\sum_{n=1}^\infty \frac{1}{n}\biggl(\sum_{m=1}^\infty e^{-\pi n(2m-1)T}\biggr)
  \le C_1.
\end{align}
Since $|1-z|\le 2$ whenever $|z|\le 1$, we also have
\[
  h_T^{(0)}(t,x)\ge -\log 2.
\]
This gives the lower bound $h_T(t,x)\ge -C_1$. Next, let $a:=\pi|t|$ and $b:=\pi x$. Since $(t,x)\in M_T$, we have $d_T(t)=|t|$ and $d_1(x)=|x|$, and so,
\[
  r_T(t,x)=\sqrt{t^2+x^2}=\frac{1}{\pi}\sqrt{a^2+b^2}.
\]
Now note that 
\[
  |1-e^{-a+\i b}|\ge C^{-1}\min\{1,\sqrt{a^2+b^2}\},
\]
by Lemma~\ref{lem:one-minus-exp-lower} (applied with $u=a$ and $v=b\in[-\pi,\pi]$).
Therefore, for $(t,x)\neq (0,0)$,
\begin{align*}
h_T^{(0)}(t,x)
&=-\log|1-e^{-a+\i b}|\\
&\le \log C -\log(\min\{1,\sqrt{a^2+b^2}\})\\
&\le C_1 + |\log\sqrt{a^2+b^2}|\\
&\le C_2 + |\log r_T(t,x)|.
\end{align*}
At $(t,x)=(0,0)$, the right side is $+\infty$, so the same inequality is immediate.
Combining this with the bound \eqref{eq:rtbound} on $R_T$ completes the proof.
\end{proof}

Now recall the function $g_T$ defined in the previous section. For distinct $u,v\in M_T$, define $G_T(u,v) := g_T(u-v)$. Then, define
\[
  \widetilde G_T(u,v):=G_T(u,v)-\frac{\pi}{6}T.
\]
For $u_1,\ldots, u_w\in M_T$, define 
\[
  \Phi_T(u_1,\dots,u_w)
  :=  \exp\biggl(
    -4b\sum_{j=1}^k\sum_{l=1}^w \alpha_j \widetilde G_T(q_j,u_l)
    -4b^2\sum_{1\le l<l'\le w}\widetilde G_T(u_l,u_{l'})
  \biggr).
\]
The following lemma proves a key upper bound on $\Phi_T$.
\begin{lmm}\label{lem:renorm-int-tail}
There exist constants $C,T_0,\eta,\delta>0$ such that for all $T\ge T_0$ and all $(u_1,\dots,u_w)\in M_T^w$, with $u_l=(t_l,x_l)$ and $|u_l|$ denoting the Euclidean norm of $u_l$, we have 
\begin{equation*}
  |\Phi_T(u_1,\dots,u_w)|
  \le
  C \exp\biggl(-\eta\sum_{l=1}^w |u_l|\biggr)
  \prod_{j=1}^k\prod_{l=1}^w d(q_j,u_l)^{-(2-\delta)},
\end{equation*}
where $d$ denotes geodesic distance on $M$.
\end{lmm}

\begin{proof}
If $w=0$, then all sums and products in the definitions are empty, so
$\Phi_T\equiv 1$ and the claim is immediate after taking any fixed
$\delta\in(0,2)$ and increasing $C$ if needed. So assume $w\ge 1$.

Fix $T$ large enough so that all the points $q_1,\dots,q_k$ belong to $M_T$. Fix $u_1,\ldots,u_w\in M_T$, with $u_l=(t_l,x_l)$. Throughout this proof, $C,C_1,\ldots$ will denote constants that do not depend on $T$ or $u_1,\ldots,u_w$, whose values may change from line to line. We will also use $O(1)$ to denote any quantity whose absolute value is bounded above by such a constant. 
Write
\[
  \widetilde g_T(t,x):=g_T(t,x)-\frac{\pi}{6}T,
\]
so that $\widetilde G_T((t,x),(s,y))=\widetilde g_T(t-s,x-y)$.
Take $(t,x)\in M_T\setminus\{(0,0)\}$. Since $|t|\le T$, we may combine equations 
\eqref{eq:gtlimit-g-fourier}, \eqref{eq:gtlimit-b0-eval}, and \eqref{eq:gtlimit-bn-eval} to get
\begin{align*}
\widetilde{g}_T(t,x) &= -\frac{\pi|t|}{2}+\frac{\pi t^2}{4T}+\frac{1}{2}\sum_{n\in \Z\setminus\{0\}} \frac{e^{\pi \i n x}}{|n|}\biggl(\sum_{m\in \Z}e^{-\pi|n||t+2mT|}\biggr).
\end{align*}
After pairing $n$ and $-n$, this gives
\begin{equation}\label{eq:gt-decomp-tail}
  \widetilde g_T(t,x)= -f_T(|t|)+h_T(t,x),
\end{equation}
where $f_T$ is the function defined in equation \eqref{eq:ftdef} and $h_T$ is the function defined in equation \eqref{eq:htdef}. Next, let
\[
  E_T^{(1)}:=-4b\sum_{j=1}^k\sum_{l=1}^w \alpha_j \widetilde G_T(q_j,u_l).
\]
Fix $l$. Write $q_j=(\tau_j,\xi_j)$, and let
\[
  d_{jl}:=d_T(t_l-\tau_j).
\]
Since $d_T$ is a $1$-Lipschitz map, we have 
\[
  |d_{jl}-|t_l||= |d_T(t_l - \tau_j) - d_T(t_l)| \le |\tau_j|.
\]
Using the Lipschitz bound $|f_T(a)-f_T(b)|\le \frac{\pi}{2}|a-b|$ from Lemma \ref{lem:ftlemma}, this implies
\[
  f_T(d_{jl})=f_T(|t_l|)+O(1).
\]
Therefore, using the decomposition~\eqref{eq:gt-decomp-tail},
\begin{align}
-4b\sum_{j=1}^k \alpha_j \widetilde G_T(q_j,u_l)
&=4b\sum_{j=1}^k \alpha_j f_T(d_{jl})
   -4b\sum_{j=1}^k \alpha_j h_T(q_j-u_l) \notag\\
&=4b f_T(|t_l|)\sum_{j=1}^k \alpha_j
   +O(1)-4b\sum_{j=1}^k \alpha_j h_T(q_j-u_l)\notag \\ 
&= -4b^2 w f_T(|t_l|)+O(1)-4b\sum_{j=1}^k \alpha_j h_T(q_j-u_l). \label{eq:insert-split}
\end{align}
Taking real parts in \eqref{eq:insert-split}, using the lower bound from Lemma \ref{lem:htlemma} for those $j$
with $\Re(\alpha_j)\ge 0$, and the upper bound from Lemma \ref{lem:htlemma} for those $j$ with $\Re(\alpha_j)<0$,
we get
\begin{align}
\Re\biggl(-4b\sum_{j=1}^k \alpha_j \widetilde G_T(q_j,u_l)\biggr)
&\le
-4b^2w f_T(|t_l|)
+ C
+(2-\delta)\sum_{j=1}^k |\log r_T(q_j,u_l)|, \label{eq:insert-final}
\end{align}
where
\[
  \delta:=\min_{1\le j\le k}(2+4b\Re(\alpha_j))>0.
\]
Summing equation \eqref{eq:insert-final} over $l=1,\dots,w$ gives
\begin{align}
\Re(E_T^{(1)})
\le
-4b^2w\sum_{l=1}^w f_T(|t_l|)
+ C
+(2-\delta)\sum_{j=1}^k\sum_{l=1}^w |\log r_T(q_j,u_l)|. \label{eq:E1-bound}
\end{align}
Next, let
\[
  E_T^{(2)}:=-4b^2\sum_{1\le l<l'\le w}\widetilde G_T(u_l,u_{l'}).
\]
Using the decomposition \eqref{eq:gt-decomp-tail}, we have
\[
  E_T^{(2)}
  =
  4b^2\sum_{1\le l<l'\le w} f_T(d_T(t_l-t_{l'}))
  -4b^2\sum_{1\le l<l'\le w} h_T(u_l-u_{l'}).
\]
By the lower bound from Lemma \ref{lem:htlemma}, the second term is bounded above by a constant:
\[
  -4b^2\sum_{1\le l<l'\le w} h_T(u_l-u_{l'})\le C.
\]
For the first term, note that
\[
  d_T(t_l-t_{l'})\le \min\{|t_l|+|t_{l'}|,T\}.
\]
Hence, by Lemma \ref{lem:ftlemma},
\[
  f_T(d_T(t_l-t_{l'}))\le f_T(\min\{|t_l|+|t_{l'}|,T\}) \le  f_T(|t_l|)+f_T(|t_{l'}|).
\]
Therefore,
\begin{align}
\Re(E_T^{(2)})
&\le 4b^2\sum_{1\le l<l'\le w}(f_T(|t_l|)+f_T(|t_{l'}|))+C\notag\\
&=4b^2(w-1)\sum_{l=1}^w f_T(|t_l|)+C. \label{eq:E2-bound}
\end{align}
From equations \eqref{eq:E1-bound} and \eqref{eq:E2-bound}, we get
\begin{align*}
\Re(E_T^{(1)}+E_T^{(2)})
&\le
-4b^2\sum_{l=1}^w f_T(|t_l|)
+ C
+(2-\delta)\sum_{j=1}^k\sum_{l=1}^w |\log r_T(q_j,u_l)|.
\end{align*}
But again by Lemma \ref{lem:ftlemma}, 
\[
  f_T(|t_l|)\ge \frac{\pi}{4}|t_l|.
\]
Using this in the previous display, we get 
\begin{align}\label{eq:re1e2}
\Re(E_T^{(1)}+E_T^{(2)})
\le
-\pi b^2\sum_{l=1}^w |t_l|
+ C 
+(2-\delta)\sum_{j=1}^k\sum_{l=1}^w |\log r_T(q_j,u_l)|.
\end{align}
Set
\[
  a:=2-\delta,\qquad \rho_{jl}:=r_T(q_j,u_l),\qquad D_{jl}:=d(q_j,u_l).
\]
Since $q_j=(\tau_j,\xi_j)$ is fixed, by increasing $T_0$ if needed we may assume
\[
  T\ge 1+\max_{1\le j\le k}|\tau_j|.
\]
Fix any $\ep>0$. We claim that there exists $C_\ep$ (independent of $T$ and $u_1,\dots,u_w$) such that
\begin{equation}\label{eq:r-to-d-bound}
  e^{a|\log \rho_{jl}|}\le C_\ep e^{\ep |t_l|} D_{jl}^{-a},
  \qquad 1\le j\le k,\ 1\le l\le w.
\end{equation}
To prove this, consider two cases. If $D_{jl}\le 1$, then $|t_l-\tau_j|\le D_{jl}\le 1<T$, and hence
$d_T(t_l-\tau_j)=|t_l-\tau_j|$. This means that $\rho_{jl}=D_{jl}$. Thus
\[
e^{a|\log \rho_{jl}|}=D_{jl}^{-a}\le e^{\ep|t_l|}D_{jl}^{-a}.
\]
If $D_{jl}>1$, then $\rho_{jl}\ge 1$ and therefore
$e^{a|\log \rho_{jl}|}=\rho_{jl}^a$. (Indeed, if $|t_l-\tau_j|\le T$, then
$\rho_{jl}=D_{jl}>1$; if $|t_l-\tau_j|>T$, then
$|t_l-\tau_j|\le T+|\tau_j|$, so
$d_T(t_l-\tau_j)=2T-|t_l-\tau_j|\ge T-|\tau_j|\ge 1$.)
Also, $d_T(\cdot)\le |\cdot|$ gives
$\rho_{jl}\le D_{jl}\le |t_l-\tau_j|+1\le C(1+|t_l|)$, and so,
\[
e^{a|\log \rho_{jl}|}=\rho_{jl}^a\le D_{jl}^a
=D_{jl}^{-a}D_{jl}^{2a}
\le C(1+|t_l|)^{2a}D_{jl}^{-a}
\le C_\ep e^{\ep|t_l|}D_{jl}^{-a}.
\]
This proves the claim \eqref{eq:r-to-d-bound}. Exponentiating the bound \eqref{eq:re1e2} on $\Re(E_T^{(1)}+E_T^{(2)})$ and applying the inequality 
\eqref{eq:r-to-d-bound} to each pair $(j,l)$ yields
\[
|\Phi_T(u_1,\dots,u_w)|
\le
C_\ep\exp\biggl(-(\pi b^2-k\ep)\sum_{l=1}^w |t_l|\biggr)
\prod_{j=1}^k\prod_{l=1}^w d(q_j,u_l)^{-(2-\delta)}.
\]
Choose $\ep>0$ so small that $\pi b^2-k\ep>0$, and define $\eta:=\pi b^2-k\ep$. Also,
\[
  \sum_{l=1}^w |t_l|
  \ge
  \sum_{l=1}^w (|u_l|-|x_l|)
  \ge
  \sum_{l=1}^w |u_l|-w,
\]
because $|x_l|\le 1$ on $M_T$. Therefore,
\[
  C_\ep\exp\biggl(-\eta\sum_{l=1}^w |t_l|\biggr)
  \le
  C_\ep e^{\eta w}\exp\biggl(-\eta\sum_{l=1}^w |u_l|\biggr).
\]
Absorbing $e^{\eta w}$ into the constant completes the proof.
\end{proof}

\subsubsection{Uniqueness of analytic continuation}\label{sec:unique}
The following lemma proves the uniqueness clause in Theorem \ref{thm:main2}.
\begin{lmm}\label{lem:uniq-analytic-cont}
For $m\ge 1$, define
\[
  \mathbb{H}_+:=\{z\in\C:\Re z>0\},\quad
  D_m:=\C\times \mathbb{H}_+^{m-1},\quad
  B_m:=\R\times(0,\infty)^{m-1}.
\]
If $F$ is analytic on $D_m$ and zero on $B_m$, then $F\equiv 0$ on $D_m$.
\end{lmm}
\begin{proof}
We prove the result by induction on $m$. For $m=1$, we have $D_1=\C$ and
$B_1=\R$. Since $F$ is continuous on $\C\cup\R=\C$ and vanishes on $\R$, the
one-variable identity theorem gives $F\equiv 0$ on $\C$.

Assume the claim holds for $m$, and let $F$ be analytic on $D_{m+1}$ and zero on $B_{m+1}$. Fix any $r>0$ and define
\[
  G_r(z_1,\ldots,z_m):=F(z_1,\ldots,z_m,r),
  \qquad (z_1,\ldots,z_m)\in D_m.
\]
Then $G_r$ is analytic on $D_m$, and $G_r=0$ on
$B_m$ because $F=0$ on $B_{m+1}$. By the induction hypothesis,
\[
  G_r\equiv 0\quad\text{on }D_m\qquad\text{for every }r>0.
\]
Now fix any $(z_1,\ldots,z_m)\in D_m$, and consider the one-variable function
\[
  h(\zeta):=F(z_1,\ldots,z_m,\zeta),
\]
which is analytic on $\mathbb{H}_+$. For each $r>0$,
since $(z_1,\ldots,z_m)\in D_m$ and $G_r\equiv 0$ on $D_m$, we have
\[
  h(r)=G_r(z_1,\ldots,z_m)=0.
\]
Thus $h$ vanishes on $(0,\infty)$. 
Hence the one-variable identity theorem implies that $h\equiv 0$ on $\mathbb{H}_+$. Since $(z_1,\ldots,z_m)$ was arbitrary, $F\equiv 0$ on $D_{m+1}$. This completes
the induction.
\end{proof}

\subsection{Power counting}\label{sec:aux-power-counting}
The following power counting lemma is used to control finite products of
affine distance singularities.

\begin{lmm}\label{lem:1d-power-counting}
Let $d\ge1$ and let $K\subseteq\R^d$ be compact. For $1\le \nu\le N$, let
$\ell_\nu:\R^d\to\R^{m_\nu}$ be an affine map whose zero set
$H_\nu:=\ell_\nu^{-1}(0)$ is either empty or a proper affine subspace of
$\R^d$, and let $a_\nu\in\R$. Suppose that, for every nonempty flat
\[
  F=\bigcap_{\nu\in I}H_\nu,
  \qquad \varnothing\ne I\subseteq\{1,\ldots,N\},
\]
one has
\[
  \sum_{\nu:\,F\subseteq H_\nu}a_\nu>-\operatorname{codim}F.
\]
(Note that here the sum is over all indices whose zero set contains $F$, not merely over
the particular set $I$ used to represent $F$ as an intersection.)
Then
\[
  \int_K\prod_{\nu=1}^N\|\ell_\nu(x)\|^{a_\nu}\,\dd x<\infty.
\]
\end{lmm}
\begin{proof}
If $H_\nu=\varnothing$, then $\|\ell_\nu\|$ is bounded above and below by
positive constants on the compact set $K$, and the corresponding factor is
harmless. If $H_\nu\ne\varnothing$, then $\|\ell_\nu(x)\|$ is comparable to
$\operatorname{dist}(x,H_\nu)$. It therefore suffices to prove the assertion
with
\[\prod_\nu \operatorname{dist}(x,H_\nu)^{a_\nu}\]
in place of the displayed integrand.

We prove the claim by induction on the ambient dimension $d$, starting with
$d=1$. In one dimension, every nonempty proper affine subspace is a point. Fix
$x^0\in K$. The factors whose zero sets do not contain $x^0$ are bounded above
and below by positive constants near $x^0$. If no zero set contains $x^0$, there
is nothing to prove locally. Otherwise, after discarding these bounded factors,
the local integrand is comparable to
\[
  |x-x^0|^A,
  \qquad A:=\sum_{\nu:\,x^0\in H_\nu}a_\nu .
\]
The flat $\{x^0\}$ has codimension one, so the hypothesis gives $A>-1$. Hence
$|x-x^0|^A$ is locally integrable, proving the base case.

Now assume $d\ge2$ and that the result is known in dimensions
$1,\ldots,d-1$. Fix a point $x^0\in K$. By compactness of $K$, it is enough to
prove local integrability near $x^0$. Let
\[
  I(x^0):=\{\nu:x^0\in H_\nu\}.
\]
If $I(x^0)=\varnothing$, then no singular factor vanishes near $x^0$, so local
integrability is immediate. Factors with $\nu\notin I(x^0)$ are bounded above
and below by positive constants in a small neighborhood of $x^0$. Translating $x^0$ to the origin,
write $V_\nu:=H_\nu-x^0$ for $\nu\in I(x^0)$, and let
\[
  W:=\bigcap_{\nu\in I(x^0)}V_\nu .
\]
Note that $W$ is a proper subspace of  $\R^d$. Indeed, $I(x^0)$ is 
nonempty by assumption, and for each $\nu\in I(x^0)$ the set $H_\nu$ is a nonempty proper
affine subspace; hence each $V_\nu$ is a proper subspace of $\R^d$, and therefore so is $W$. 

If $W\ne\{0\}$, split $\R^d=W\oplus W^\perp$ and write a point near the origin
as $w+y$, with $w\in W$ and $y\in W^\perp$. Since $W\subseteq V_\nu$ for all
$\nu\in I(x^0)$, each singular subspace decomposes as
\[
  V_\nu=W\oplus (V_\nu\cap W^\perp).
\]
Indeed, if $v\in V_\nu$ and $v=w+y$ is its orthogonal decomposition, then
$w\in W\subseteq V_\nu$, and hence $y=v-w\in V_\nu\cap W^\perp$. Consequently
translation in the $W$ direction does not change any of the singular distances:
\[
  \operatorname{dist}(w+y,V_\nu)
  =\operatorname{dist}(y,V_\nu\cap W^\perp),
  \qquad \nu\in I(x^0).
\]
Thus the local integrand depends on the transverse variable $y$ through the
arrangement of subspaces
\[
  V_\nu^\perp:=V_\nu\cap W^\perp\subset W^\perp,
  \qquad \nu\in I(x^0),
\]
up to the bounded nonsingular factors already discarded.

The power-counting hypotheses are inherited by this transverse arrangement.
To see this, let $M\subset W^\perp$ be a flat obtained as an intersection of
some of the subspaces $V_\nu^\perp$. Then $W\oplus M$ is the corresponding flat
of the original arrangement. Moreover,
\[
  \operatorname{codim}_{W^\perp} M
  =\operatorname{codim}_{\R^d}(W\oplus M),
\]
and $M\subseteq V_\nu^\perp$ holds exactly when
$W\oplus M\subseteq V_\nu$. Therefore the exponent sum attached to $M$ in the
transverse arrangement is the same as the exponent sum attached to
$W\oplus M$ in the original arrangement, while the codimension is also the
same. The original hypothesis therefore gives precisely the hypothesis needed
for the arrangement in $W^\perp$.

Since $\dim W^\perp<d$, the induction hypothesis gives local integrability in
the transverse variable $y$. The remaining $w$ variable ranges over a bounded
set, so Fubini's theorem gives local integrability in the full neighborhood.

It remains to treat the case $W=\{0\}$. Then all singular subspaces pass
through the origin and have total intersection $\{0\}$. In polar coordinates
$x=r\theta$, $0<r<\epsilon$, $\theta\in\mathbb S^{d-1}$, the singular part is
\[
  r^{A}\prod_{\nu\in I(x^0)}\operatorname{dist}(\theta,V_\nu)^{a_\nu},
  \qquad A:=\sum_{\nu\in I(x^0)}a_\nu .
\]
The hypothesis applied to the flat $\{0\}$ gives $A>-d$, so the radial
integral $\int_0^\epsilon r^{A+d-1}\,\dd r$ is finite. It remains to prove
integrability of the angular factor on $\mathbb S^{d-1}$.

Fix $\theta^0\in\mathbb S^{d-1}$ and let
$J:=\{\nu\in I(x^0):\theta^0\in V_\nu\}$. The factors with $\nu\notin J$ are
bounded near $\theta^0$. In a smooth coordinate chart centered at $\theta^0$,
the factors with $\nu\in J$ are comparable to distances to the linear subspaces
$V_\nu\cap(\theta^0)^\perp$ of the tangent space $(\theta^0)^\perp$. Indeed,
after rotating so that $\theta^0$ is the last coordinate vector, each
$V_\nu$ with $\theta^0\in V_\nu$ is defined by linear equations involving only
the first $d-1$ coordinates in this chart. For any flat $M$ in this tangent
arrangement, $M\oplus\R\theta^0$ is a flat of the original arrangement,
\[
  \operatorname{codim}_{(\theta^0)^\perp}M
  =\operatorname{codim}_{\R^d}(M\oplus\R\theta^0),
\]
and the factors containing $M$ in the tangent arrangement are exactly the
factors whose original subspaces contain $M\oplus\R\theta^0$. Hence the
power-counting hypothesis is inherited by the tangent arrangement. By the
induction hypothesis in dimension $d-1$, the angular factor is locally
integrable near $\theta^0$. Compactness of $\mathbb S^{d-1}$ gives angular
integrability, and hence local integrability near $x^0$. This proves the lemma.
\end{proof}

\subsection{Topological vector spaces and completeness}\label{sec:app-tvs}
We briefly recall the notions used in \secref{sec:gnsquant}. A
topological vector space over $\C$ is a complex vector space $V$ equipped with
a topology such that vector addition
\[
  V\times V\to V,\qquad (u,v)\mapsto u+v,
\]
and scalar multiplication
\[
  \C\times V\to V,\qquad (\lambda,v)\mapsto \lambda v,
\]
are continuous.

A directed set is a preordered set $(D,\le)$ such that for every $i,j\in D$
there exists $k\in D$ with $i\le k$ and $j\le k$. A net in a set $X$ is simply
a family $(x_i)_{i\in D}$ indexed by a directed set. The phrase
``eventually'' means ``for all sufficiently large indices'': a property $P(i)$
holds eventually if there exists $i_0\in D$ such that $P(i)$ holds for all
$i\ge i_0$. A net $(v_i)_{i\in D}$ in $V$ converges to $v\in V$ if for every
open neighborhood $U$ of $v$, one has $v_i\in U$ eventually. Nets are the
right replacement for sequences in arbitrary topological spaces, because they
characterize continuity and closure without first-countability assumptions. In
particular, if $E\subseteq X$ and $x\in\bar E$, then there exists a net
$(x_i)$ in $E$ such that $x_i\to x$. This is the fact used in the proof of
Theorem~\ref{thm:innlmm}, where an arbitrary state
$z\in\overline{\Phi(\ma)}$ is approximated by a net from $\Phi(\ma)$.

Since a general topological vector space need not be first countable,
completeness is naturally formulated using nets rather than sequences. A net
$(v_i)_{i\in D}$ in $V$ is called Cauchy if for every open neighborhood $U$ of
$0$ there exists $i_0\in D$ such that
\[
  v_i-v_j\in U\qquad\text{for all }i,j\ge i_0.
\]
The space $V$ is called complete if every Cauchy net in $V$ converges to some
point of $V$.

A subset $B\subset V$ is called bounded if for every open neighborhood $U$ of
$0$ there exists $r>0$ such that $B\subset \,\lambda U$ whenever
$|\lambda|\ge r$. If $V$ is locally convex, one natural topology on the space
$\mathrm{End}_{\mathrm{cont}}(V)$ of continuous linear endomorphisms is the
topology of uniform convergence on bounded subsets: a net $(T_i)$ converges to
$T$ if for every bounded $B\subset V$ and every neighborhood $U$ of $0$, one
has
\[
  (T_i-T)(B)\subset U
\]
eventually. This is the operator topology used in
\secref{sec:gnsquant} to define the closed represented local
algebras $\mathfrak A_{\mathrm{loc}}(O)$.

Now let $I$ be an arbitrary index set, and consider the product space
\[
  \C^I:=\{z=(z_i)_{i\in I}: z_i\in \C\text{ for all }i\in I\}
\]
with the product topology. Equivalently, this is the smallest topology making
each coordinate projection
\[
  \pi_i:\C^I\to\C,\qquad \pi_i(z)=z_i,
\]
continuous. A basic neighborhood of a point $z\in\C^I$ is obtained by choosing
finitely many indices $i_1,\dots,i_n\in I$, open sets $U_1,\dots,U_n\subseteq
\C$ with $z_{i_m}\in U_m$, and then requiring only that the corresponding
coordinates lie in these prescribed open sets. In particular, all but finitely
many coordinates are left unrestricted.

With this topology, $\C^I$ is a topological vector space: addition and scalar
multiplication are defined coordinatewise, and continuity is checked coordinate
by coordinate. It is also Hausdorff, since two distinct points in $\C^I$ differ
in some coordinate, and the corresponding projection separates them.

We now explain why $\C^I$ is complete. Let $(z^{(\alpha)})_{\alpha\in D}$ be a
Cauchy net in $\C^I$, where
\[
  z^{(\alpha)}=(z_i^{(\alpha)})_{i\in I}.
\]
Fix any $i\in I$. Because the projection $\pi_i$ is continuous, the scalar net
$(z_i^{(\alpha)})_{\alpha\in D}$ is Cauchy in $\C$. Since $\C$ is complete,
there exists $z_i\in \C$ such that
\[
  z_i^{(\alpha)}\to z_i\qquad\text{in }\C.
\]
Doing this for each $i\in I$ defines a point $z=(z_i)_{i\in I}\in \C^I$.

To show that $z^{(\alpha)}\to z$ in the product topology, let $U$ be any basic
neighborhood of $z$. Then $U$ only constrains finitely many coordinates, say
$i_1,\dots,i_n$. For each $m$, the convergence
$z_{i_m}^{(\alpha)}\to z_{i_m}$ in $\C$ implies that eventually
$z_{i_m}^{(\alpha)}\in U_m$. Intersecting these finitely many eventuality
conditions, we conclude that $z^{(\alpha)}\in U$ eventually. Hence
$z^{(\alpha)}\to z$ in $\C^I$.

Therefore $\C^I$ is a complete Hausdorff topological vector space for every
index set $I$. In particular, the space $\C^{\ma}$ used in
\secref{sec:gnsquant} is complete.
With this topological background in place, we conclude the appendix by recalling the standard AQFT
framework used for comparison in the main text.

\subsection{Quick recap of the standard AQFT framework}\label{sec:app-aqft}
We briefly recall the Haag--Kastler framework of algebraic quantum field
theory~\cite{haagkastler64,haag96}. Fix a spacetime $M$ and let $\mathcal K$
be a collection of suitable spacetime regions (for example, bounded causally
convex open sets). A Haag--Kastler net assigns to each $O\in\mathcal K$ a
unital $C^*$-algebra $\mathfrak A(O)$ of local observables. The basic axioms
include:
\begin{enumerate}
\item \textit{Isotony.} If $O_1\subset O_2$, then
\[
  \mathfrak A(O_1)\subset\mathfrak A(O_2).
\]
\item \textit{Locality.} If $O_1$ and $O_2$ are spacelike
separated, then every element of $\mathfrak A(O_1)$ commutes with every element
of $\mathfrak A(O_2)$.
\item \textit{Covariance.} If a symmetry group $G$ acts on spacetime, then one
has a corresponding action by $*$-automorphisms
\[
  \alpha_g:\mathfrak A(O)\to\mathfrak A(gO),\qquad g\in G.
\]
\end{enumerate}
Because each $\mathfrak A(O)$ is a $C^*$-algebra, it is already complete in its
$C^*$-norm. One also forms the quasilocal algebra $\mathfrak A$ by taking the
norm-closed $C^*$-algebra generated by the local algebras. This is the
representation-independent level of AQFT.

A state is a positive normalized linear functional
$\omega:\mathfrak A\to\C$. By the GNS theorem, such a state yields a Hilbert
space representation $(\pi_\omega,\mathcal H_\omega,\Omega_\omega)$ with
$\Omega_\omega$ cyclic and
\[
  \omega(A)=\langle\Omega_\omega,\pi_\omega(A)\Omega_\omega\rangle,
  \qquad A\in\mathfrak A.
\]
At this stage one can look at the concrete represented images
\[
  \pi_\omega(\mathfrak A(O))\subset\mathcal B(\mathcal H_\omega).
\]
In many AQFT arguments one then passes to the associated local von Neumann
algebras
\[
  \mathcal M_\omega(O):=\pi_\omega(\mathfrak A(O))''\subset\mathcal B(\mathcal H_\omega).
\]
Here, for any subset $X\subset\mathcal B(\mathcal H_\omega)$, its commutant is
\[
  X':=\{T\in\mathcal B(\mathcal H_\omega): TA=AT\text{ for all }A\in X\},
\]
and its double commutant is $X'':=(X')'$. Since $\pi_\omega(\mathfrak A(O))$
is a unital $*$-subalgebra of $\mathcal B(\mathcal H_\omega)$, the von
Neumann bicommutant theorem says that $\pi_\omega(\mathfrak A(O))''$
coincides with the weak operator closure of $\pi_\omega(\mathfrak A(O))$
(equivalently, with its strong operator closure). Thus the double commutant is
indeed a closure, but in the weak operator topology rather than the norm topology. 

Both formulations are standard: the abstract Haag--Kastler net is
naturally $C^*$-algebraic, while a fixed vacuum representation often leads one
to work with the corresponding local von Neumann algebras.
This extra closure is not merely cosmetic. After the state space
$\mathcal H_\omega$ has been constructed, one wants local operator algebras on
that space that are stable under the natural operator limits and on which the
symmetry implementation acts concretely. The passage from
$\pi_\omega(\mathfrak A(O))$ to $\mathcal M_\omega(O)$ is the standard AQFT
way to do this. The main text follows the same logic in a locally convex setting without positivity: after constructing $\mh$, it passes from the algebraic
images $\pi(\ma(O))$ to the closed represented local algebras
$\mathfrak A_{\mathrm{loc}}(O)$.

If the state is invariant under the symmetry group and the corresponding
regularity hypotheses hold, the automorphisms are implemented by a unitary
representation $U_\omega(g)$ satisfying
\[
  U_\omega(g)\pi_\omega(A)U_\omega(g)^{-1}=
  \pi_\omega(\alpha_g(A)).
\]
and likewise preserve the von Neumann closures $\mathcal M_\omega(O)$.
For the translation subgroup, one usually also requires the spectrum condition:
the joint spectrum of the generators lies in the forward light cone.

This is the standard reference picture against which the AQFT-type
construction in~\secrefrange{sec:aqft}{sec:timeev} may be compared.
In the main text, we retain isotony, a vacuum functional, and a represented
local net obtained from the algebraic images $\pi(\ma(O))$ by locally convex
operator closure, together with translation covariance and Einstein locality.
But we do not obtain positivity, abstract local
$C^*$-algebras, represented local von Neumann algebras on a Hilbert space, or
the usual spectrum condition.

\end{document}